\documentclass[a4paper,UKenglish]{lipics-v2019}
 
\usepackage{microtype}

\usepackage{amsmath,amssymb}

\usepackage{xspace}
\usepackage{tikz}
\usepackage{graphicx}
\usepackage{color}
\usepackage{xcolor}
\usepackage{mathrsfs}
\usetikzlibrary{shapes,snakes, calc}
\usepackage{calc}

\title{On Kernelization for Edge Dominating Set under Structural Parameters}
\titlerunning{Edge Dominating Set} 

\author{Eva-Maria C. Hols}{Department of Computer Science, Humboldt-Universit{\"a}t zu Berlin, Germany}{hols@informatik.hu-berlin.de}{https://orcid.org/0000-0002-2832-0722}{Supported by DFG Emmy Noether-grant (KR 4286)}
\author{Stefan Kratsch}{Department of Computer Science, Humboldt-Universit{\"a}t zu Berlin, Germany}{kratsch@informatik.hu-berlin.de}{https://orcid.org/0000-0002-0193-7239
}{}
\authorrunning{E.C. Hols and S. Kratsch} 

\Copyright{Eva-Maria C. Hols, Stefan Kratsch}


\category{}

\supplement{}

\funding{}

\acknowledgements{}

\nolinenumbers 

\hideLIPIcs  

\EventEditors{John Q. Open and Joan R. Access}
\EventNoEds{2}
\EventLongTitle{42nd Conference on Very Important Topics (CVIT 2016)}
\EventShortTitle{CVIT 2016}
\EventAcronym{CVIT}
\EventYear{2016}
\EventDate{December 24--27, 2016}
\EventLocation{Little Whinging, United Kingdom}
\EventLogo{}
\SeriesVolume{42}
\ArticleNo{23}

\theoremstyle{definition}
\newtheorem{redrule}{Reduction Rule}

\theoremstyle{remark}

\newcommand{\prob}[1]{\textsc{\lowercase{#1}}}

\newcommand{\Oh}{\mathcal{O}}
\newcommand{\VC}{\prob{Vertex Cover}\xspace}
\newcommand{\EDS}{\prob{Edge Dominating Set}\xspace}
\newcommand{\MCC}{\prob{Multicolored Clique}\xspace}
\newcommand{\NP}{\ensuremath{\mathsf{NP}}\xspace}
\newcommand{\containment}{\ensuremath{\mathsf{NP\subseteq coNP/poly}}\xspace}
\newcommand{\ncontainment}{\ensuremath{\mathsf{NP\nsubseteq coNP/poly}}\xspace}
\newcommand{\SAT}{$\mathrm{SAT}$\xspace}
\newcommand{\C}{\ensuremath{\mathcal{C}}\xspace}
\DeclareMathOperator{\MEDS}{\textsc{eds}}
\newcommand{\cost}{\mathrm{cost}}

\newcommand{\R}{\ensuremath{\mathcal{R}}\xspace}
\newcommand{\Q}{\ensuremath{\mathcal{Q}}\xspace}
\newcommand{\N}{\ensuremath{\mathbb{N}}\xspace}

\newcommand{\cH}{\ensuremath{\mathcal{H}}\xspace}

\definecolor{lgray}{gray}{0.6}
\definecolor{bgray}{gray}{0.9}
\definecolor{Green}{rgb}{0.0, 0.5, 0.0}
\definecolor{Blue}{rgb}{0.08, 0.38, 0.74}

\begin{document}

\maketitle

\begin{abstract}
In the \NP-hard \EDS problem (EDS) we are given a graph $G=(V,E)$ and an integer $k$, and need to determine whether there is a set $F\subseteq E$ of at most $k$ edges that are incident with all (other) edges of $G$. It is known that this problem is fixed-parameter tractable and admits a polynomial kernel when parameterized by $k$. A caveat for this parameter is that it needs to be large, i.e., at least equal to half the size of a maximum matching of $G$, for instances not to be trivially negative. Motivated by this, we study the existence of polynomial kernels for EDS when parameterized by \emph{structural} parameters that may be much smaller than $k$.

Unfortunately, at first glance this looks rather hopeless: Even when parameterized by the deletion distance to a disjoint union of paths $P_3$ of length two there is no polynomial kernelization (under standard assumptions), ruling out polynomial kernels for many smaller parameters like the feedback vertex set size. In contrast, somewhat surprisingly, there is a polynomial kernelization for deletion distance to a disjoint union of paths $P_5$ of length \emph{four}. As our main result, we fully classify for all finite sets $\cH$ of graphs, whether a kernel size polynomial in $|X|$ is possible when given $X$ such that each connected component of $G-X$ is isomorphic to a graph in $\cH$. 
\end{abstract}

\section{Introduction}

In the \EDS problem (EDS) we are given a graph $G=(V,E)$ and an integer $k$, and need to determine whether there is a set $F\subseteq E$ of at most $k$ edges that are incident with all (other) edges of $G$. It is known that this is equivalent to the existence of a maximal matching of size at most $k$. The \EDS problem is \NP-hard but admits a simple $2$-approximation by taking any maximal matching of $G$. It can be solved in time $\Oh^*(2.2351^k)$\footnote{$\Oh^*$-notation hides factors that are polynomial in the input size.}~\cite{IwaideN16}, making it fixed-parameter tractable for parameter $k$. Additionally, for EDS any given instance $(G,k)$ can be efficiently reduced to an equivalent one $(G',k')$ with only $\Oh(k^2)$ vertices and $\Oh(k^3)$ edges \cite{XiaoKP13} (this is called a \emph{kernelization}).

The drawback of choosing the solution size $k$ as the parameter is that $k$ is large on many types of easy instances. 
This has been addressed for many other problems by turning to so called \emph{structural parameters} that are independent of the solution size. Two lines of research in this direction have yielded polynomial kernels for several other \NP-hard problems. One possibility is to choose the parameter as the size of a set $X$ such that $G-X$ belongs to some class $\C$ where the problem in question can be efficiently solved; such sets $X$ are called \emph{modulators}. The other possibility is to parameterize above some lower bound for the solution, i.e., the parameter is the difference between the solution size $k$ and the lower bound. 

The \VC problem, where, given a graph $G$ and an integer $k$, we are asked whether there are $k$ vertices that are incident with all edges, has been successfully studied under different structural parameters. 
It had been observed that \VC is FPT parameterized by the size of a modulator to a class $\C$ when one can solve vertex cover on graphs that belong to $\C$ in polynomial time; e.g. if $\C$ is the graph class of forests or, more generally, of bipartite or K\H{o}nig graphs. 
Furthermore, there also exist kernelizations for \VC parameterized by modulators to some graph classes $\C$. The first of a number of such results is due to Jansen and Bodlaender~\cite{JansenB13} who gave a kernelization with $\Oh(\ell^3)$ vertices where $\ell$ is the size of a (minimum) feedback vertex set of the input graph. Clearly, the solution size $k$ cannot be bounded in terms of $\ell$ alone because forests already have arbitrarily large minimum vertex covers. This result has been generalized, e.g., for parameterization by the size of an odd cycle transversal~\cite{KratschW12}.

There are also parameterized algorithms for \VC above lower bounds that address the specific complaint about the seemingly unnecessarily large parameter value $k$ in many graph classes. It was first shown that \VC parameterized by $\ell=k-MM$ where $MM$ stands for the size of a maximum matching is FPT~\cite{RamanRS11}. In other words, the parameter value $\ell$ is the difference between $k$ and the obvious lower bound. This has been improved to work also for parameterization by $\ell=k-LP$ where $LP$ stands for the minimum \emph{fractional vertex cover} (as determined by the LP relaxation)~\cite{CyganPPW13,LokshtanovNRRS14} and, recently, even for parameter $\ell=k-(2LP-MM)$~\cite{GargP16}. All of these above lower bound parameterizations of \VC also have randomized polynomial kernels~\cite{KratschW12,Kratsch16}.

Motivated by the number of positive results for \VC parameterized by structural parameters we would like to know whether some of these results carry over to the related but somewhat more involved \EDS problem. 

\subparagraph{Our results.}
For kernelization subject to the size of a modulator to some tractable class $\C$ there is bad news: Even if $\C$ contains only the disjoint unions of paths of length two (consisting of three vertices each) we show that there is no polynomial kernelization for parameterization by $|X|$ with $G-X\in\C$ unless \containment (and the polynomial hierarchy collapses). The same is true when $\C$ contains at least all disjoint unions of triangles. Thus, for the usual program of studying modulators to well-known \emph{hereditary} graph classes $\C$ there is essentially nothing left to do because the only permissible connected components would have one or two vertices.\footnote{This very modest case actually admits a polynomial kernelization.} That said, as the next result shows, this perspective would ignore an interesting landscape of positive and negative results that can be obtained by permitting certain forms of connected components in $G-X$ but not necessarily all induced subgraphs thereof, i.e., by dropping the requirement that $\C$ needs to be hereditary (closed under induced subgraphs).

Indeed, there is, e.g., a polynomial kernelization for parameter $|X|$ when all connected components of $G-X$ are paths of length \emph{four}. This indicates that the structure even of constant-sized components permitted in $G-X$ determines in a nontrivial way whether or not there is a polynomial kernelization. Note the contrast with \VC where a modulator to component size $d$ admits a kernelization with $\Oh(k^d)$ vertices for each fixed $d$. Naturally, we are interested in finding out exactly which cases admit polynomial kernels.

This brings us to our main result. For $\cH$ a set of graphs, say that $G$ is an $\cH$-component graph if each connected component of $G$ is isomorphic to some graph in $\cH$. We fully classify the existence of polynomial kernels for parameterization by the size of a modulator to the class of $\cH$-component graphs for all finite sets $\cH$. To clarify, the input consists of $(G,k,X)$ such that $G-X$ is an $\cH$-component graph and the task is to determine whether $G$ has an edge dominating set of size at most $k$; the parameter is $|X|$.
Note that these problems are fixed-parameter tractable for all finite sets $\cH$ because $G$ has treewidth at most $|X|+\Oh(1)$.

\begin{theorem}\label{theorem:classification}
For every finite set \cH of graphs, the \EDS problem parameterized by the size of a given modulator $X$ to the class of $\cH$-component graphs falls into one of the following two cases:
\begin{enumerate}
 \item It has a kernelization with $\Oh(|X|^{d})$ vertices, $\Oh(|X|^{d+1})$ edges, and size $\Oh(|X|^{d+1}\log |X|)$. Moreover, unless \containment, there is no kernelization to size $\Oh(|X|^{d-\varepsilon})$ for any $\varepsilon>0$. Here $d=d(\cH)$ is a constant depending only on the set $\cH$.
 \item It has no polynomial kernelization unless \containment.
\end{enumerate}
\end{theorem}

To obtain the classification one needs to understand how connected components of $G-X$ that are isomorphic to some graph $H\in\cH$ can interact with a solution for $G$, and to derive properties of $H$ that can be leveraged for kernels or lower bounds for kernelization. Crucially, edge dominating sets for $G$ may contain edges between $X$ and components of $G-X$. From the perspective of such a component (isomorphic to $H$) this is equivalent to first covering edges incident with some vertex set $B \subseteq V(H)$ (the endpoints of chosen edges to $X$) and then covering the remaining edges by a minimum edge dominating set for $H-B$. Depending on the size of a minimum edge dominating set of $H-B$ and further properties of $H$, such a set $B$ may be used to rule out any polynomial kernels or to give a lower bound of $\Oh(|X|^{d-\varepsilon})$ for the kernel size, where $d=|B|$. Conversely, absence of such sets or an upper bound for their size can be leveraged for kernels. Some sets $B$ may make others redundant, complicating both upper and lower bounds.

For a given finite set $\cH$ of graphs, the lower bound obtained from the classification is simply the strongest one over all $H\in\cH$. If this does not already rule out a polynomial kernelization then, for each $H\in\cH$, we can reduce the number of components isomorphic to $H$ to $\Oh(|X|^{d(H)})$ where $d(H)$ depends only on $H$. Moreover, we also have the almost matching lower bound of $\Oh(|X|^{d(H)-\varepsilon})$, assuming \ncontainment. The value $d(\cH)$ is the maximum over all $d(H)$ for $H\in\cH$ that yield such a polynomial lower bound; it can be computed in time depending only on $\cH$, i.e., in constant time for each fixed $\cH$.

Regarding parameterization above lower bounds, we prove that it is \NP-hard to determine whether a graph $G$ has an edge dominating set of size equal to the lower bound of half the size of a maximum matching. This rules out any positive results for parameter $\ell=k-\frac12MM$.

\subparagraph{Related work.}
The parameterized complexity of \EDS has been studied in a number of papers~\cite{Fernau06,FominGSS09,WangCFC09,Xiao10,XiaoKP13,XiaoN13a,EscoffierMPX15,IwaideN16}. Structural parameters were studied, e.g., by Escoffier et al.~\cite{EscoffierMPX15} who obtained an $\Oh^*(1.821^\ell)$ time algorithm where $\ell$ is the vertex cover size of the input graph, and by Kobler and Rotics~\cite{KoblerR03} who gave a polynomial-time algorithm for graphs of bounded clique-width. It is easy to see that EDS is fixed-parameter tractable with respect to the treewidth of the input graph. Prieto~\cite{Prieto05} was the first to find a kernelization to $\Oh(k^2)$ vertices for the standard parameterization by $k$; this was improved to $\Oh(k^2)$ vertices and $\Oh(k^3)$ edges by Xiao et al.~\cite{XiaoKP13} and further tweaked by Hagerup~\cite{Hagerup12}. Our work appears to be the first to study the existence of polynomial kernels for EDS subject to structural parameters, though some lower bounds, e.g., for parameter treewidth are obvious.

Classically, \EDS remains \NP-hard on planar cubic graphs, bipartite graphs with maximum degree three~\cite{YannakakisG80}. This implies \NP-hardness already for $|X|=0$ when considering parameterization by a modulator to any graph class containing this special case. \EDS has also been studied from the perspective of approximation~\cite{FujitoN02,ChlebikC06,CardinalLL09,SchmiedV12,EscoffierMPX15}, enumeration~\cite{KanteLMN12,GolovachHKV15,KanteLMNU15}, and exact exponential-time algorithms~\cite{RamanSS07,Xiao10,RooijB12,XiaoN14}.

\subparagraph{Organization.}
We begin with some preliminaries in Section~\ref{section::preliminaries}. Section~\ref{section::lowerbounds} provides some intuition for the main result by proving the lower bound for \EDS parameterized by the size of a modulator to a $P_3$-component graph as well as the polynomial kernelization for parameterization by the size of a modulator to a $P_5$-component graph. Section~\ref{section::classification} gives a detailed statement of the main result including the required definitions to determine which result applies for any given set $\cH$. Section~\ref{section::aboveLB} contains the hardness proof for parameter $\ell=k-\frac12MM$. We conclude in Section~\ref{section::conclusion}.

\section{Preliminaries}\label{section::preliminaries}

We use standard graph notation as given by Diestel~\cite{Diestel12}. In particular, for a graph $G=(V,E)$ we let $N(v)=\{u\in V\mid \{u,v\}\in E\}$ and $N[v]=N(v)\cup\{v\}$; similarly, $N[X]=\bigcup_{x\in X}N[x]$ and $N(X)=N[X]\setminus X$. We let $E(X,Y)=\{\{x,y\}\mid x\in X, y\in Y\}$ and we let $\delta(v)=\{\{u,v\}\mid u\in V, \{u,v\}\in E\}$. By $G[X]$ we denote the induced subgraph of $G$ on vertex set $X$ and by $G-X$ the induced subgraph on vertex set $V\setminus X$; we let $G-v=G-\{v\}$. We denote the size of a minimum edge dominating set of a graph $G$ by $\MEDS(G)$. 

Let \cH be a set of graphs. We say that a graph $G$ is an \emph{\cH-component graph} if each connected component of $G$ is isomorphic to some graph in $\cH$. Clearly, disconnected graphs in $\cH$ do not affect which graphs $G$ are $\cH$-component graphs and, thus, our proofs need only consider the connected graphs $H\in\cH$. 
We write $H$-component graph rather than $\{H\}$-component graph for single (connected) graphs $H$.

Let $[n]$ denote the set $\{1, 2,\ldots, n\}$.

\subparagraph{Parameterized complexity.}
A \emph{parameterized problem} \Q is a subset of $\Sigma^*\times\N$ where $\Sigma$ is any finite set. The second component $k$ of instances $(x,k)$ is called the \emph{parameter}. A parameterized problem \Q is \emph{fixed-parameter tractable} if there is an algorithm that correctly solves all instances $(x,k)$ in time $f(k)|x|^c$ where $f$ is a computable function and $c$ is a constant independent of $k$. A \emph{kernelization} for \Q is an efficient algorithm that, given an instance $(x,k)$, takes time polynomial in $|x|+k$ and returns an instance $(x',k')$ of size at most $f(k)$ such that $(x,k)\in\Q$ if and only if $(x',k')\in\Q$ where $f$ is a computable function. The function $f$ is also called the \emph{size} of the kernelization and a kernelization is polynomial (resp.\ linear) if $f(k)$ is polynomially (resp.\ linearly) bounded in $k$.

We use the notion of a cross-composition~\cite{BodlaenderJK14}, which is a convenient front-end for the seminal kernel lower bound framework of Bodlaender et al.~\cite{BodlaenderDFH09} and Fortnow and Santhanam~\cite{FortnowS11}. A relation $\R\subseteq\Sigma^*\times\Sigma^*$ is a \emph{polynomial equivalence relation} if equivalence of two strings $x,y\in\Sigma^*$ can be tested in time polynomial in $|x|+|y|$ and if \R partitions any finite set $S\subseteq\Sigma^*$ into a number of classes that is polynomially bounded in the largest element of $S$.

\begin{definition}[(OR-)cross-composition \cite{BodlaenderJK14}]
Let $L\subseteq\Sigma^*$ be a language, let \R be a polynomial equivalence relation on $\Sigma^*$, and let $\Q\subseteq\Sigma^*\times\N$ be a parameterized problem. An \emph{\mbox{(OR-)}cross-composition of $L$ into $\Q$} (with respect to $\R$) is an algorithm that, given $t$ instances $x_1,\ldots,x_t\in\Sigma^*$ of $L$ belonging to the same equivalence class of \R, takes time polynomial in $\sum_{i=1}^t|x_i|$ and outputs an instance $(y,k)\in\Sigma^*\times\N$ such that the following hold:
\begin{itemize}
 \item ``PB'': The parameter value $k$ is polynomially bounded in $\max_{i=1}^t |x_i|+\log t$. 
 \item ``OR'': The instance $(y,k)$ is yes for \Q if and only if \emph{at least one} instance $x_i$ is yes for $L$.
\end{itemize}
An \emph{(OR-)cross-composition of $L$ into $\Q$ of cost $f(t)$} instead satisfies ``OR'' and ``CB'':
\begin{itemize}
 \item ``CB'': The parameter value $k$ is bounded by $\Oh(f(t)\cdot (\max_{i=1}^t|x_i|)^c)$, where $c$ is some constant independent of $t$.
\end{itemize}
\end{definition}

If $L$ is \NP-hard then both forms of cross-compositions are known to imply lower bounds for kernelizations for $\Q$. Theorem~\ref{theorem:polynomiallowerbound} additionally builds on Dell and van Melkebeek~\cite{DellM14}.

\begin{theorem}[{\cite[Corollary 3.6.]{BodlaenderJK14}}]\label{theorem:kernellowerbound:crosscomposition}
If an \NP-hard language L has a cross-composition to \Q then \Q admits no polynomial kernelization or polynomial compression unless \containment.
\end{theorem}

\begin{theorem}[{\cite[Theorem 3.8.]{BodlaenderJK14}}]\label{theorem:polynomiallowerbound}
Let $d,\varepsilon>0$. If an \NP-hard language L has a cross-composition into \Q of cost $f(t)=t^{1/d+o(1)}$, where $t$ is the number of instances, then \Q has no polynomial kernelization or polynomial compression of size $\Oh(k^{d-\varepsilon})$ unless \containment.
\end{theorem}

All our composition-based proofs use for $L$ the \NP-hard \MCC problem. Therein we are given a graph $G=(V,E)$, an integer $k$, and a partition of $V$ into $k$ sets $V_1,\ldots,V_k$ of equal size; we need to determine whether there is a clique of size $k$ in $G$ that contains exactly one vertex from each set $V_i$. Such a set $X$ is called a \emph{multicolored $k$-clique}.

\section[EDS parameterized by the size of a modulator to a P3- resp. P5-component graph]{EDS parameterized by the size of a modulator to a $\boldsymbol{P_3}$- resp.\ $\boldsymbol{P_5}$-component graph}\label{section::lowerbounds}

In this section we study the difference of \EDS parameterized by the size of a modulator to a $P_3$-component graph and \EDS parameterized by the size of a modulator to a $P_5$-component graph, which are both more restrictive than parameterization by size of a feedback vertex set (modulator to a forest). Note that the latter is FPT, because the treewidth is at most the size of the feedback vertex set plus one and \EDS parameterized by the treewidth is FPT. Hence, \EDS parameterized by the above modulators is FPT too. 

First, we show that \EDS parameterized by the size of a modulator to a $P_3$-component graph has no polynomial kernel unless $\containment$. This rules out polynomial kernels for a large number of interesting parameters like feedback vertex set size or size of a modulator to a linear forest. 
Somewhat surprisingly, we then show that when parameterized by the modulator to a $P_5$-component graph we do get a polynomial kernel.

\subsection[Lower bound for EDS parameterized by the size of a modulator to a P3-component graph]{Lower bound for EDS parameterized by the size of a modulator to a $\boldsymbol{P_3}$-component graph}

We give a kernelization lower bound for \EDS parameterized by the size of a modulator $X$, such that deleting $X$ results in a disjoint union of $P_3$'s. To prove this we give a cross-composition from \MCC.

\begin{theorem} \label{theorem::P3}
  \EDS parameterized by the size of a modulator to a $P_3$-component graph (and thus also parameterized by the size of a modulator to a linear forest) does not admit a polynomial kernel unless $\containment$.
\end{theorem}

\begin{proof}
  To prove the theorem we give a cross-composition from the \NP-hard \MCC problem to \EDS parameterized by the size of a modulator to a $P_3$-component graph. Input instances are of the form $(G_i,k_i)$ where $G_i$ comes with a partition of the vertex set into $k$ color classes. (Since the color classes are of equal size it holds that $k \leq |V(G_i)|$.)  
  For the polynomial equivalence relation $\mathcal{R}$ we take the relation that puts two instances $(G_1,k_1)$, $(G_2,k_2)$ of \MCC in the same equivalence class if $k_1=k_2$ and $|V(G_1)|=|V(G_2)|$. It is easy to check that $\mathcal{R}$ is a polynomial equivalence relation. (Instances with size at most $N$ have at most $N$ vertices. Thus, we get at most $N^2$ classes for instances of size at most $N$.)
  
  Let a sequence of instances $I_i=(G_i,k)_{i=1}^{t}$ of \MCC be given that are equivalent under $\mathcal{R}$.
  We identify the color classes of the input graphs so that all graphs have the same vertex set $V$ and the same color classes $V_1, V_2, \ldots, V_k$. Let $n:=|V_i|$ be the number of vertices of each color class; thus, each instance has $|V|=n\cdot k$ vertices. We assume w.l.o.g.\ that every instance has at least one edge in $E(V_p,V_q)$ for all $1 \leq p < q \leq k$; otherwise, this instance would be a trivial no instance and we can delete it. Furthermore, we can assume w.l.o.g.\ that $t=2^s$ for an integer $s$, since we may copy some instances if needed (while at most doubling the number of instances and increasing $\log t$ by less than one).

  \begin{figure}[t]
   \centering
  \begin{tikzpicture}[scale=1, transform shape]
     \node[draw, ellipse, minimum width=3cm, minimum height=0.5cm] (Z') at (0,0)[label=right:$Z'$] {};
     \node[draw, ellipse, minimum width=3cm, minimum height=0.5cm] (Z) at (0,-1) [label=right:$Z$] {};
     \node[draw, rectangle, rounded corners, minimum width=6cm, minimum height=0.7cm] (W) at (0,-2) [label=right:$W$] {};
     \node[draw, rectangle, minimum width=12cm, minimum height=4cm, fill, color=bgray] at (0,-4.5) [] {}; 
     \foreach \i/\j/\place in {0/1/left,1/2/left,2/3/left,3/4/left,5/{t-1}/left,6/{t}/left}
	\node[draw, rectangle, minimum width=1cm, minimum height=3cm] (I_\i) at (1.6*\i-4.8,-4.5) [label=\place:$I_{\j}$] {};
     \node at (1.6*5-6.7,-4.5) {$\ldots$}; 
     \node[draw, rectangle, rounded corners, minimum width=4cm, minimum height=1cm] (V) at (-2.5,-7.5) [label=left:$V$] {};
     \node[draw, rectangle, rounded corners, minimum width=4cm, minimum height=0.55cm] (S) at (2.5,-7.12) [label=right:$S$] {};
     \node[draw, rectangle, rounded corners, minimum width=4cm, minimum height=0.55cm] (S) at (2.5,-7.88) [label=right:$S'$] {};
     \node[draw, ellipse, minimum width=3cm, minimum height=0.5cm] (T) at (-2.5,-9) [label=left:$T$] {};
     \node[draw, ellipse, minimum width=3cm, minimum height=0.5cm] (T') at (-2.5,-10) [label=left:$T'$] {};
     \foreach \i in {-2,-1,...,3}{
	\node[draw, circle, inner sep=1pt, fill] (zp_\i) at (0.5*\i-0.25,0) {};
	\node[draw, circle, inner sep=1pt, fill] (z_\i) at (0.5*\i-0.25,-1) {};
	\draw (zp_\i) -- (z_\i);}
     \foreach \i in {-2,-1,...,1}{
	\node[draw, circle, inner sep=1pt, fill] (t_\i) at (0.5*\i -2.5+0.25,-9) {};
	\node[draw, circle, inner sep=1pt, fill] (tp_\i) at (0.5*\i -2.5+0.25,-10) {};
	\draw (t_\i) -- (tp_\i);
      }
      \foreach \j in {-5,-4,...,6}{
	  \node[draw, circle, inner sep=1pt, fill] (w_\j) at (0.5*\j -0.25,-2) {};	
	  \foreach \i in {-2,-1,...,3}{
	      \draw (z_\i) -- (w_\j);}}
      \foreach \j in {-5,-4,...,6}{
	  \node[draw, circle, inner sep=1pt, fill] (v_\j) at (0.3*\j -2.65,-7.5) {};
      }
      \foreach \i in {-2,-1,...,1}{
           \foreach \j in {0,1,2}{
              \pgfmathtruncatemacro{\number}{3 *\i +1 +\j};
	      \draw (t_\i) -- (v_\number);}}
      \foreach \j in {1,2,...,6}{
	  \node[draw, circle, inner sep=1pt, fill] (s1_\j) at (0.7*\j +0.05,-7.1) {};
	  \node[draw, circle, inner sep=1pt, fill] (s2_\j) at (0.7*\j +0.05,-7.9) {};
	  \draw (s1_\j) -- (s2_\j);
	}
      \foreach \i in {0,1,2,3,5,6}{
	  \foreach \j in {1,3,7,9}{
	    \node[draw, circle, inner sep=1pt, fill] (x_\i_\j) at (1.6*\i +1.2-6.25,-3-0.3*\j) {};
	    \node[draw, circle, inner sep=1pt, fill] (e_\i_\j) at (1.6*\i +1.2-6,-0.3*\j-3) {};
	    \node[draw, circle, inner sep=1pt, fill] (y_\i_\j) at (1.6*\i +1.2-5.75,-3-0.3*\j) {};
	    \draw (x_\i_\j) -- (e_\i_\j);
	    \draw (y_\i_\j) -- (e_\i_\j);}
	    \node at (1.6*\i +1.2-6,-2.6-0.3*6) {$\vdots$};
      }    
      \foreach \i in {-5,-4,...,0}{
	  \draw[color=lgray] (w_\i) -- (x_0_1);
	  \draw[color=lgray] (w_\i) -- (x_0_3);}
      \foreach \i in {-5,-4,-3,5,6}{
	  \draw[color=lgray] (w_\i) -- (x_3_1);
	  \draw[color=lgray] (w_\i) -- (x_3_3);}
      \draw[color=lgray] (x_0_9) -- (v_-1);
      \draw[color=lgray] (x_0_9) -- (v_1);
      \draw[color=lgray] (e_0_9) -- (s1_1);
      \draw[color=lgray] (x_0_7) -- (v_-5);
      \draw[color=lgray] (x_0_7) -- (v_-2);
      \draw[color=lgray] (e_0_7) -- (s1_2);
      \draw[color=lgray] (x_3_9) -- (v_3);
      \draw[color=lgray] (x_3_9) -- (v_5);
      \draw[color=lgray] (e_3_9) -- (s1_3);
      \draw[color=lgray] (x_3_7) -- (v_0);
      \draw[color=lgray] (x_3_7) -- (v_1);
      \draw[color=lgray] (e_3_7) -- (s1_1);
      \foreach \k in {40,60,...,140}{
	\foreach \i in {1,2,5,6}{
	  \foreach \j in {1,3,7,9}{
	    \draw[color=lgray] (x_\i_\j) -- +(\k:0.23);}}
	\foreach \i in {0,3}{
	  \foreach \j in {7,9}{
	    \draw[color=lgray] (x_\i_\j) -- +(\k:0.23);}}
      }
      \foreach \k in {260,290}{
	\foreach \i in {1,2,5,6}{
	  \foreach \j in {1,3,7,9}{
	    \draw[color=lgray] (x_\i_\j) -- +(\k:0.23);
	    \draw[color=lgray] (e_\i_\j) -- +(-30-15*\i:0.23);}}
	\foreach \i in {0,3}{
	  \foreach \j in {1,3}{
	    \draw[color=lgray] (x_\i_\j) -- +(\k:0.23);
	    \draw[color=lgray] (e_\i_\j) -- +(-30-15*\i:0.23);}}
      }
  \end{tikzpicture}
  \caption{Construction of the graph $G'$ with $k=4$, where $X'=W \cup Z \cup Z' \cup V \cup T \cup T' \cup S\cup S'$}
  \label{figure::lower_bound}
  \end{figure}
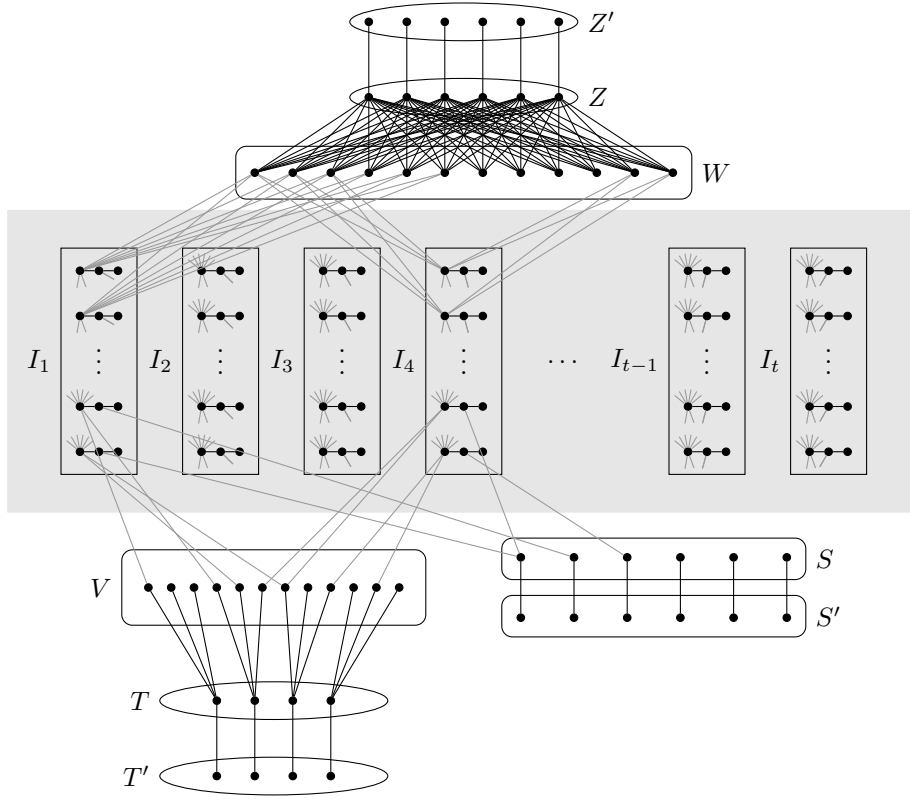
  Now, we construct an instance $(G',k',X')$ of \EDS parameterized by the size of a modulator to a $P_3$-component graph, where the size of $X'$ is polynomially bounded in $n+k+s$ (see Figure \ref{figure::lower_bound} for an illustration). 
  We add a set $V$ consisting of $k \cdot n$ vertices to graph $G'$ which represents the vertices of the $t$ instances. The set $V$ is partitioned into the $k$ color classes $V_1,V_2,\ldots, V_k$. To choose which vertices are contained in a clique of size $k$, we add a set $T=\{t_1,t_2,\ldots,t_k\}$ and a set $T'=\{t_1',t_2',\ldots,t_k'\}$, each of size $k$, to $G'$. We make $t_j \in T$, with $j \in [k]$, adjacent to all vertices in $V_j$ and to vertex $t_j' \in T'$.
  Next, we add two sets $Z$, $Z'$, each of size $s$, and a set $W$ of size $2s$ to $G'$ and add edges to $G'$ such that each vertex in $Z$ has exactly one private neighbor in $Z'$ and is adjacent to all vertices in $W$. The set $W$ contains $\binom{2s}{s} \geq 2^s$ different subsets of size $s$. For each instance $(G_i,k)$, with $i \in [t]$, we pick a different subset of size $s$ of $W$ and denote it by $W(i)$.
  For all $1\leq p<q \leq k$ we add a vertex $s_{p,q}$ and a vertex $s_{p,q}'$ to $G'$; these will correspond to edge sets $E(V_p,V_q)$. Let $S=\{s_{p,q} \mid 1\leq p < q \leq k\}$ and $S'=\{s_{p,q}' \mid 1 \leq p < q \leq k\}$. We make vertex $s_{p,q}$ adjacent to vertex $s_{p,q}'$ for all $1 \leq p < q \leq k$.
  For each graph $G_i$, for $i \in [t]$, we add $|E(G_i)|$ paths of length two to the graph $G'$; every $P_3$ represents exactly one edge of the graph $G_i$. Let $P_i^e=u_{i,1}^e u_i^e u_{i,2}^e$ denote the path of instance $i \in [t]$ that represents edge $e \in E(G_i)$. 
  Finally, we make vertices in $P_i^e$, with $i \in [t]$ and $e \in E(G_i)$, adjacent to vertices in the sets $W$, $V$, and $S$ as follows:
  We make vertex $u_{i,1}^e$ of path $P_i^e$, with $i \in [t]$, which represents edge $e=\{x,y\} \in E(G_i)$ adjacent to the vertices $x,y$ in $V$ and to all vertices in the set $W(i) \subseteq W$. Additionally, we make vertex $u_i^e$ adjacent to vertex $s_{p,q}$ where $1 \leq p < q \leq k$ such that $e \in E(V_p,V_q)$.
  
  The set $X'$ is defined to contain all vertices that do not participate in the paths $P_i^e$, i.e., $X'=W \cup Z \cup Z' \cup V \cup T \cup T' \cup S\cup S'$. Clearly, $G-X'$ is a $P_3$-component graph and $|X'|=4s+k \cdot n+2k+2\cdot\binom{k}{2}$. Let $k'=k+s+\sum_{i=1}^t |E(G_i)|$. Note that the size of $k'$ can depend linearly on the number of instances, because our parameter is the size of $X'$, which is polynomially bounded in $n+s$, as $k\leq n$.
  We return the instance $(G',k',X')$; clearly, this instance can be generated in polynomial time.
  
  Now, we have to show that $(G',k',X')$ is a $\mathrm{YES}$-instance of \EDS if and only if there exists an $i^* \in [t]$ such that $(G_{i^*},k)$ is a $\mathrm{YES}$-instance of \MCC.
  
  $(\Rightarrow:)$ Assume first that $(G',k',X')$ is yes for EDS and that there exists an edge dominating set $F$ of size at most $k'$ in $G'$. 
  We can always pick $F$ such that it fulfills the following properties (most hold for all solutions of size at most $k'$): 
  \begin{enumerate}
   \item The vertex sets $S$, $T$, and $Z$ must be subsets of $V(F)$: E.g., for each edge $\{z,z'\}$ with $z\in Z$ and $z'\in Z'$ the set $V(F)$ must contain $z$ or $z'$; if it contains $z'$ then $\{z,z'\}\in F$ as it is the only edge incident with $z'$; either way we get $z\in V(F)$. The same applies for $S$ and $S'$, and for $T$ and $T'$.
   \item Because $S,T,Z\subseteq V(F)$ but $S\cup T\cup Z$ is an independent set, the set $F$ must contain at least $|S|$ edges incident with $S$, $|T|$ edges incident with $T$, and $|Z|$ edges incident with $Z$. By straightforward replacement arguments we may assume that $F$ contains exactly the following edges incident with $S\cup T\cup Z$: $|T|$ edges between $T$ and $V$, $|Z|$ edges between $Z$ and $W$, and $|S|$ edges between $S$ and middle vertices $u_i^e$ of $P_3$'s in $G'-X'$. Furthermore, we can assume that these edges are a matching, because no color class is empty, no edge set $E(V_p,V_q)$ is empty, and $Z$ is adjacent to all vertices in $W$.
   \item For each $P_i^e=u^e_{i,1} u^e_i u^e_{i,2}$, which represents the edge $e$ of instance $(G_i,k)$, at least vertex $u^e_i$ must be an endpoint of an edge in $F$: Indeed, to cover the edge $\{u^e_i,u^e_{i,2}\}$ one of its two vertices must be in $V(F)$. Similar to Property 1 above, if $u^e_{i,2}\in V(F)$ then $F$ must contain its sole incident edge $\{u^e_i,u^e_{i,2}\}$ and, hence, $u^e_i\in V(F)$.
   \item An edge in $F$ cannot have its endpoints in two different $P_3$'s of $G'-X'$ because no such edges exist. 
  \end{enumerate}
  Let $F_T=F \cap E(T,V)$, let $F_Z=F \cap E(Z,W)$, let $F_S=F \cap E(S, \{u^e_i\mid i\in[t], e\in E(G_i)\})$, and let $F_R=F \setminus (F_T \cup F_Z \cup F_S)$. Hence, due to Properties 1 and 2, we have
  \[
  |F_R| \leq k' - |F_T| -|F_Z| - |F_S|\leq \sum_{i=1}^t |E(G_i)| - \binom{k}{2}.
  \] 
  By Property 3, all vertices $u^e_i$ are endpoints of edges in $F$. Among $F_T\cup F_Z\cup F_S$ this can only be true for the $|S|=\binom{k}{2}$ edges in $F_S$. Since there are exactly $\sum_{i=1}^t |E(G_i)|$ vertices $u^e_i$, which is (greater or) equal to $|F_R|+|F_S|$, and there are no edges connecting different such vertices, each edge in $F_R\cup F_S$ is incident with a private vertex $u^e_i$. This also implies that all edges in $F_R$ have no endpoints in $V\cup W$ as those sets are not adjacent to any vertex $u^e_i$.
  Thus, in $W$ exactly the $|Z|=s$ endpoints of $F_Z$ are endpoints of $F$. Similarly, in $V$ exactly the $|T|=k$ endpoints of $F_T$ are endpoints of $F$; let $X\subseteq V$ denote this set of $k$ vertices. Observe that by construction of $G'$ the set $X$ contains exactly one vertex from each color class, because $t_j \in T$, for $j \in [k]$, is only adjacent to vertices of $V_j$.

  Now, consider any path $P_i^e=u^e_{i,1} u^e_i u^e_{i,2}$ where $u_i^e$ is an endpoint of an edge $f\in F_S$. Clearly, the other endpoint of $f$ lies in $S$, and, by the above accounting, no other edge of $F$ is incident with $u^e_{i,1}$ or $u^e_{i,2}$. In particular, this implies that all neighbors of $u^e_{i,1}$ in $W$ and $V$ must be endpoints of edges in $F$. If $e=\{x,y\}$ then these neighbors of $u^e_{i,1}$ are the set $W(i)\subseteq W$ and the vertices $x,y\in V$, and, by construction of $G'$, the edge $\{x,y\}$ must exist in $G_i$. Thus, $W(i) \cup \{x,y\} \subseteq V(F)$ which implies that $x,y \in X$.
  
  Repeating this argument for all $|S|=\binom{k}{2}$ paths of this type, we can conclude the following: (1) All paths correspond to the same instance $i^*\in[t]$ because we require $W(i)\subseteq V(F)$, but exactly $|Z|=|W(i^*)|=s$ such vertices are in $V(F)$. (Different values of $i$ would require different sets $W(i)$, exceeding size $s$.) (2) There are $\binom{k}{2}$ edges of $G_{i^*}$ represented by the paths and all their endpoints must be in $X=V\cap V(F)$. Since $|X|=k$, the edges must form a clique of size $k$ on vertex set $X$ in $G_{i^*}$. We already observed above that $X$ contains exactly one vertex per color class, hence, instance $(G_{i^*},k)$ is yes, as claimed.
    
  $(\Leftarrow:)$
  For the other direction, assume that for some $i^* \in [t]$ the \MCC instance $(G_{i^*},k)$ is a $\mathrm{YES}$-instance. Let $X=\{x_1,x_2,\ldots,x_k\} \subseteq V$ be a multicolored clique of size $k$ in $G_{i^*}$ with $x_j \in V_j$ for $j \in [k]$, let $E'$ be the set of edges of the clique $X$, and let $e_{p,q}=\{x_p,x_q\}$ for $1\leq p < q \leq k$. We construct an edge dominating set $F$ of $G'$ of size at most $k'$ as follows:
  First we add the $k$ edges $\{t_j,x_j\}$ for $j \in [k]$ between $T$ and $X\subseteq V$; thus, $T\cup X\subseteq V(F)$. We then add a maximum matching (of size $s$) between $W(i^*) \subseteq W$ and $Z$ to the set $F$. This matching saturates $W(i^*)$ and $Z$ because $|Z|=|W(i^*)|=s$; thus, $W(i^*)\cup Z\subseteq V(F)$. Next, we add the edges $\{u^{e_{p,q}}_{i^*},s_{p,q}\}$ for all edges $e_{p,q} \in E'$, with $1\leq p < q \leq k$, to the set $F$; hence $S \subseteq V(F)$. Finally, for all other paths $P_i^e$, with $i \in [t]$, $e \in E(G_i)$, and $i \neq i^*$ or $e \notin E'$, we add the edge $\{u^e_{i,1}, u^e_i\}$ to $F$. (We have thus selected exactly one edge incident with each path of $G'-X'$.) By construction, it holds that $|F|=k+s+ \sum_{i=1}^t |E(G_i)| = k'$.
  
  It remains to show that $F$ is indeed an edge dominating set of $G'$. To prove this, it suffices to show that $V(G')-V(F)$ is an independent set in $G'$. We already know that $S\cup T\cup W(i^*)\cup X\cup Z\subseteq V(F)$. Moreover, $V(F)$ contains the middle vertex $u^e_i$ for all $P_3$'s in $G'-X'$ and it contains $u_{i,1}^e$ for all $P_3$'s that do not correspond to an edge of the clique $X$ (i.e., with $i\neq i^*$ or with $i=i^*$ but $e\neq e_{p,q}$ for any $1\leq p<q\leq k$). The sets $S'$, $T'$, and $Z'$ are independent sets whose neighborhoods $S$, $T$, and $Z$ are subsets of $V(F)$. Similarly, all vertices $u_{i,2}^e$ have their single neighbor $u^e_i$ in $V(F)$. Thus, only vertices in $W\setminus W(i^*)$ and $V\setminus X$ could possibly be adjacent to vertices $u_{i^*,1}^{e_{p,q}}$, which correspond to the edges of $G_{i^*}[X]$, in $G'-V(F)$, but this can be easily refuted: Indeed, each $u_{i^*,1}^{e_{p,q}}$ is adjacent only to $x_p$ and $x_q$ in $V$, which are both in $X\subseteq V(F)$, and to the vertices in $W(i^*)$ in $W$, but $W(i^*)\subseteq V(F)$ as well. Thus $V(G')-V(F)$ is an independent set in $G'$ and hence $F$ is an edge dominating set for $G'$ of size at most $k'$. Thus, $(G',k',X')$ is yes, which completes the cross-composition.  

  By Theorem~\ref{theorem:kernellowerbound:crosscomposition} the cross-composition from \MCC implies the claimed lower bound for kernelization.
\end{proof}

We proved that \EDS parameterized by the size of a modulator to a $P_3$-component graph has no polynomial kernelization unless \containment. A similar proof establishes the same lower bound for modulators to $K_3$-component graphs. As mentioned in the introduction this rules out polynomial kernels using modulators to essentially all interesting hereditary graph classes.\footnote{It certainly does completely settle the question for modulators to $\cH$-component graphs for all \emph{hereditary} classes $\cH$. If $\cH$ contains any connected graph with at least three vertices then we get a lower bound; else all connected components have one or two vertices and there is a polynomial kernel.}

\subsection[Polynomial kernel for EDS parameterized by the size of a modulator to a P5-component graph]{Polynomial kernel for EDS parameterized by the size of a modulator to a $\boldsymbol{P_5}$-component graph}

To illustrate why other, non-hereditary, sets $\cH$ may well allow polynomial kernels for parameterization by the size of a modulator $X$ to an $\cH$-component graph, we sketch a simple kernelization for the case of $\cH=\{P_5\}$, i.e., when components of $G-X$ are isomorphic to the path of length four. This does not use the full generality of the kernelization obtained in Section \ref{section::classification} because $P_5$ does not have any (later called) uncovered vertices or (later called) strongly beneficial sets (which are the main source of complication).

For the kernelization we need the following theorem which is due to Hopcroft and Karp \cite{DBLP:journals/siamcomp/HopcroftK73}. The second claim of the theorem is not standard (but well known).
\begin{theorem}[\cite{DBLP:journals/siamcomp/HopcroftK73}] \label{theorem::matching}
    Let $G$ be an undirected bipartite graph with partition $R$ and $S$, on $n$ vertices and $m$ edges. Then we can find a maximum matching of $G$ in time $\Oh(m \sqrt{n})$. Furthermore, in time $\Oh(m \sqrt{n})$ we can find either a maximum matching that saturates $R$ or a set $Y \subseteq R$ such that $|N_G(Y)|<|Y|$ and such that there exists a maximum matching $M$ in $G - N_G[Y]$ that saturates $R \setminus Y$.
\end{theorem}

\begin{theorem} \label{theorem::kernel_P5}
    \EDS parameterized by the size of a given modulator $X$ to a $P_5$-component graph admits a kernel with $\Oh(|X|)$ vertices.
\end{theorem}
\begin{proof}
    Let $(G,k,X)$ be an instance of \EDS parameterized by the size of a modulator to a $P_5$-component graph, and let $\C$ be the set of connected components of $G-X$.
  We construct a bipartite graph $G_B$ where one part is the set $X$, the other part consists of one vertex $s_P$ for every connected component $P$ in $\C$, and where there is an edge between $x \in X$ and $s_P$ with $P=w_1w_2w_3w_4w_5 \in \C$ if and only if $x$ is adjacent to a vertex of $P$ that is not the middle vertex $w_3$. Now, we apply Theorem~\ref{theorem::matching} to obtain either a maximum matching in $G_B$ that saturates $X$ or a set $Y \subseteq X$ such that $|N_{G_B}(Y)| < |Y|$ and such that there exists a maximum matching in $G_B - N_{G_B}[Y]$ that saturates $X \setminus Y$. If there exists a maximum matching in $G_B$ that saturates $X$ then let $X_1=X$ and $X_2 = \emptyset$. Otherwise, if there exists a set $Y$ with the above properties then let $X_1=X \setminus Y$ and $X_2 = Y$. Observe that $X_2$ also contains the vertices in $X$ that are only adjacent to middle vertices of components in $\C$, and the vertices in $X$ that are not adjacent to any component in $\C$. Let $M$ be a maximum matching in $G_B-N_{G_B}[X_2]$ that saturates $X_1$. The partition $X_1 \dot\cup X_2$ of $X$ fulfills the following properties:
  \begin{itemize}
      \item Let $\C_2$ be the set of connected components $P$ in $\C$ where $s_P$ is a vertex in $N_{G_B}(X_2)$, i.e., $\C_2=\{P=w_1w_2w_3w_4w_5 \in~\C \mid N_G(\{w_1,w_2,w_4,w_5\}) \cap X_2 \neq \emptyset\}$.
      It holds either that $\C_2$ is the empty set (when $X_2 = \emptyset$) or that it contains less than $|X_2|$ connected components of $\C$, i.e., $|\C_2|< |X_2|$ (when $Y=X_2 \neq \emptyset$).
      \item For every vertex $x \in X_1$, let $P_x=w_1^xw_2^xw_3^xw_4^xw_5^x$ be the connected component in $\C_1:=\C \setminus \C_2$ that is paired to $x$ by $M$, i.e., $\{x,s_{P_x}\} \in M$. It holds that there exists a vertex $w^x \in \{w_1^x,w_2^x,w_4^x,w_5^x\}$ such that $\{w^x,x\} \in E(G)$ (definition of $G_B$). Note that $\C_1$ also contains all connected components that are not adjacent to any vertex in $X$ or where only the middle vertex of a path in $\C$ is adjacent to a vertex in $X$.
  \end{itemize}
  Using the above partition, one can show that there exists an optimum solution $S$ that contains for each path $P_x$ with $x \in X_1$ the locally optimal solution $\{\{x,w^x\},\{w_3^x,w_2^x\}\}$ resp.\ $\{\{x,w^x\},\{w_3^x,w_4^x\}\}$ depending on whether $w^x \in \{w_4^x,w_5^x\}$ or $w^x \in \{w_1^x,w_2^x\}$. More generally, for every vertex $w$ of a path $P \in \C$, except the middle vertex, and every vertex $x \in X$ that is adjacent to $w$ there exists a local optimum solution to $P$ that uses edge $\{w,x\}$ and has the middle vertex of $P$ as an endpoint of the second solution edge.
  This is the crucial difference to a path $P'=v_1v_2v_3$ of length two. Here, the only locally optimal solution that dominates $P'$ and contains an edge between $P'$ and $X$ is $\{\{v_2,x\}\}$ with $x \in X$, but this local solution does not contain the vertices $v_1$ and $v_3$. We used this in our lower bound construction to control which $P_3$'s may be used to ``buy'' vertices in $X$.
  \begin{redrule} \label{rule::P5_X}
    Delete $X_1$ from $G$, i.e., let $G'=G-X_1$, $X'=X \setminus X_1=X_2$, and $k'=k$.
  \end{redrule}
  \begin{claim}
    Reduction Rule~\ref{rule::P5_X} is safe.
  \end{claim}
  \begin{claimproof}
    Let $F$ be an edge dominating set of size at most $k$ in $G$. We construct an edge dominating set $F'$ of size at most $k'=k$ in $G'$ by deleting every edge $e=\{x,y\} \in F$ if both endpoints of $e$ are contained in $X_1$, or if exactly one endpoint is contained in $X_1$ and the other endpoint is isolated in $G'$; and by replacing every edge $e=\{x,y\} \in F$ with $x \in X_1$ and $y \notin X_1$ by exactly one edge in $\delta_{G'}(y)$ if $\delta_{G'}(y) \neq \emptyset$.
    It holds that $F'$ has size at most $k=k'$ because we either delete edges in $F$ or replace them one for one by a new edge. Since every vertex in $V(G') \cap V(F)$ is either contained in $V(F')$ or isolated in $G'$ it holds that $F'$ is an edge dominating set in $G'$.

    For the other direction, let $F'$ be an edge dominating set of size at most $k'$ in $G'$.
    Consider the path $P_x=w_1^xw_2^xw_3^xw_4^xw_5^x$ for some vertex $x \in X_1$. It holds that the only vertex in $P_x$ that can be adjacent to a vertex in $X'=X\setminus X_1 = X_2$ is vertex $w_3^x$; otherwise $P_x$ would be a component in $\C_2$ and not in $\C_1$ (by definition of $\C_1$ and $\C_2$).
    Furthermore, the edge dominating set $F'$ must dominate the two non-adjacent edges $\{w_1^x,w_2^x\}$ and $\{w_4^x,w_5^x\}$. Since $w_1^x$, $w_2^x$, $w_4^x$, and $w_5^x$ are only adjacent to vertices in $P_x$ the set $F'$ must contain one of the two edges $e^x_{1,2}=\{w_1^x,w_2^x\}$, $e^x_{2,3}=\{w_2^x,w_3^x\}$ and one of the two edges $e^x_{3,4}=\{w_3^x,w_4^x\}$, $e^x_{4,5}=\{w_4^x,w_5^x\}$. To obtain an edge dominating set of size at most $k$ in $G$ we replace for each vertex $x \in X_1$ these edges with the local optimum solution $\{\{x,w^x\},\{w_3^x,w_2^x\}\}$ resp.\ $\{\{x,w^x\},\{w_3^x,w_4^x\}\}$ depending whether $w^x \in \{w_4^x,w_5^x\}$ or $w^x \in \{w_1^x,w_2^x\}$.
    It holds that $|F| \leq |F'|$ because for every vertex $x \in X_1$ we replace the at least two edges in $F' \cap \{e^x_{1,2}, e^x_{2,3}, e^x_{3,4}, e^x_{4,5}\}$ by the two edges of the locally optimal solution $\{\{x,w^x\},\{w_3^x,w_2^x\}\}$ resp.\ $\{\{x,w^x\},\{w_3^x,w_4^x\}\}$.

    It remains to show that $F$ is indeed an edge dominating set in $G$. The set $V(F)$ contains all vertices in $V(F')$, except some vertices in the connected components $P_x$ with $x \in X_1$ where we change the edge dominating set $F'$. Furthermore, $V(F)$ contains all vertices in $X_1$ because for every vertex $x \in X_1$ the edge $\{w^x,x\}$ is contained in $F$.
    Thus, the only edges that are possibly not dominated by $F$ have one endpoint in a path $P_x$ with $x \in X_1$. Since $w_3^x$ is contained in $V(F)$ (by construction), since every edge in $P_x$ is dominated by $F$ (by construction), and since the vertices in $\{w^x_1,w^x_2,w^x_4,w^x_5\}$ are only adjacent to vertices in $P_x \cup X_1$, it follows that $F$ is an edge dominating set in $G$.
  \end{claimproof}

  After applying Reduction Rule~\ref{rule::P5_X} it holds that for each path $P=w_1w_2w_3w_4w_5 \in \C_1$ only the vertex $w_3$ can be adjacent to a vertex in $X$, and we can assume that every (optimum) solution contains the edges $\{w_2,w_3\}$ and $\{w_3,w_4\}$.
  Additionally, one can show that there exists an optimum solution that does not contain any edge between $\C_1$ and $X$ because we can replace any such edge $e=\{x,v\}$ with $v \in V(\C_1)$ by the edge $\{x,u\}$ with $u \in N_G(x) \setminus V(\C_1)$ (or delete this edge when $N_G(x) \setminus V(\C_1)=\emptyset$).
  This allows us to delete $\C_1$ from $G$.
  \begin{redrule} \label{rule::P5_cc}
    Delete all connected components in $\C_1$ and decrease $k$ by the size of a minimum edge dominating set in $\C_1$, i.e., let $G'=G-\C_1$, $X'=X$, and $k'=k-\MEDS(\C_1)$.
  \end{redrule}
  \begin{claim}
    Reduction Rule~\ref{rule::P5_cc} is safe.
  \end{claim}
  \begin{claimproof}
    First, we will show that there exists an edge dominating set $F$ of size at most $k$ in $G$ such that no edge in $F$ has one endpoint in a connected component of $\C_1$ and the other endpoint in $X$.
    Let $F$ be an edge dominating set of size at most $k$ in $G$ with $F \cap E(\C_1,X)$ minimal, and let $P=w_1w_2w_3w_4w_5$ be a path in $\C_1$.
    We can assume, w.l.o.g., that $F$ contains the edges $\{w_2,w_3\}$ and $\{w_3,w_4\}$ because $F$ must dominate the non-adjacent edges $\{w_1,w_2\}$, $\{w_4,w_5\}$, and the vertices $w_1$, $w_2$, $w_4$, $w_5$ are only adjacent to vertices in $P$; otherwise, $P$ is contained in $\C_2$ and not $\C_1$.
    Now, assume for contradiction that there exists an edge $e=\{x,y\} \in F \cap E(\C_1,C)$ with $x \in X$ and $y \in P$ where $P=w_1w_2w_3w_4w_5$ is a path in $\C_1$. It holds that $y=w_3$ because $w_3$ is the only vertex in $P$ that is adjacent to a vertex in $X$.
    If every vertex $u \in N_G(x)$ is contained in $V(F)$ then let $\widetilde{F}=F \setminus \{e\}$. Otherwise, let $\widetilde{F}=F \setminus \{e\} \cup \{\{x,u\}\}$, where $u \in N_G(x) \setminus V(F)$. It holds that $\widetilde{F}$ is an edge dominating set in $G$ because $y=w_3$ is still a vertex in $V(\widetilde{F})$ which implies $V(F) \subseteq V(\widetilde{F})$. Furthermore, $u$ is not contained in a connected component of $\C_1$ because for every path $P=w_1w_2w_3w_4w_5$ in $\C_1$ the vertex $w_3$ is contained in $V(F)$ and no other vertex is adjacent to a vertex in $X$.
    Now, the set $\widetilde{F}$ is an edge dominating set of size at most $k$ in $G$ with $\widetilde{F} \cap E(\C_1,X) \subsetneq  F \cap E(\C_1,X)$ which contradicts the minimality of $F \cap E(\C_1,X)$ and proves that there exists an edge dominating set $F$ of size at most $k$ in $G$ with $F \cap E(\C_1,X) = \emptyset$.
    This implies that $F'=F \setminus E(\C_1)$ is an edge dominating set of size at most $k'$ in $G'$ when $F$ is a solution to $(G,k,X)$ with $F \cap E(\C_1,X) = \emptyset$.

    For the other direction, let $F'$ be an edge dominating set of size at most $k'$ in $G'$. To obtain an edge dominating set $F$ of size at most $k$ in $G$ we add for every path $P=w_1w_2w_3w_4w_5$ in $\C_1$ the two edges $\{w_2,w_3\}$ and $\{w_3,w_4\}$, which are a minimum edge dominating set of $P$, to $F'$. It follows that $F$ has size $|F'|+\MEDS(\C_1) \leq k$.
    The set $F$ dominates all edges in $G-X$ as well as all edges between $\C_2$ and $X$ because $F' \subseteq F$, and because $F$ contains an edge dominating set of $\C_1$. Additionally, $F$ dominates all edges between $\C_1$ and $X$ because $F$ dominates all middle vertices of the paths in $\C_1$ which are the only vertices in $\C_1$ that are adjacent to $X$. Hence, $F$ is an edge dominating set of size at most $k$ in $G$.
  \end{claimproof}
  Let $(G',k',X')$ be the reduced instance. It holds that the set of connected components in $G'-X'$ is $\C_2$ because we delete all other connected components during Reduction Rule~\ref{rule::P5_cc}. Since $|\C_2| \leq |X_2|=|X'|$ it follows that $G'$ has at most $5 \cdot |\C_2|+|X'| \leq 6|X'|$ vertices.
  It remains to show that we can perform the reduction in polynomial time. We apply each Reduction Rule at most once. Furthermore, we can apply the Reduction Rules in polynomial time because we can compute the partition of $X$ as well as the sets $\C_1$ and $\C_2$ in polynomial time, and because we can delete sets of vertices from $G$ and $X$ in polynomial time.
\end{proof}

While this is not the full story about the classification in the following section, it hopefully shows the spirit of how upper and lower bounds for kernelization can arise. Solution edges between components of $G-X$ and $X$ play a crucial role and they affect the solutions for components in nontrivial ways, e.g., apart from control opportunities, it depends on how much budget is needed for $H-B$ when edges between $B$ and $X$ are in the solution.

\section[EDS parameterized by the size of a modulator to an H-component graph]{EDS parameterized by the size of a modulator to an $\boldsymbol{\cH}$-component graph}  \label{section::classification}

In this section, we develop a complete classification of \EDS parameterized by the size of a modulator to an $\cH$-component graph regarding existence of polynomial kernels for all finite sets $\cH$. This is motivated by the observed difference between modulating to $P_3$-component graphs (no polynomial kernel unless \containment) vs.\ modulating to $P_5$-component graphs (polynomial kernelization). 
To this end, we will study which properties graphs $H\in\cH$ must have, such that \EDS parameterized by the size of a modulator to an $\cH$-component graph has resp.\ does not have a polynomial kernel. To recall, the input of our problem is a tuple $(G,k,X)$ where $G-X$ is an $\cH$-component graph and we ask whether $G$ has an edge dominating set of size at most $k$; the parameter is $|X|$.

In contrast to \VC, where we can delete a vertex in the modulator if we know that this vertex must be in a solution of certain size, this is not the case for \EDS because we do not necessarily know which incident edge should be chosen. Of course, we can check for a vertex $x$ in the modulator $X$ how not having this vertex as an endpoint of a solution edge influences the size of a minimum edge dominating set of $G-X$. But, even if we find out that a vertex $x$ in the modulator $X$ must be an endpoint of a solution edge, we do not know if the other endpoint of the solution edge incident with $x$ is in $X$ or in a connected component of $G-X$. If there would be a connected component $C$ in $G-X$ with the property that there exists a vertex $v \in N(x) \cap V(C)$ with $\MEDS(C)=\MEDS(C-v) +1$, then it could be possible to have $x$ as an endpoint of a solution edge without paying more than the cost of a minimum edge dominating set in $C$. 
Thus, instead of finding vertices in the modulator that must be endpoint of a solution edge, we want to find vertices in the modulator that can be endpoints of a solution edge without spending more budget than the size of a minimum edge dominating set in $G-X$. Similarly, getting edges to $r$ vertices in $X$ while increasing the cost in $C$ by less than $r$ is of interest (cost equal to $r$ can always be had). The following definition classifies relevant vertices and vertex sets in a graph $H$, which may occur as a component of $G-X$.

\newcommand{\helpfulcomment}[1]{\emph{(#1)}}

\begin{definition}\label{definition:basicterms}
    Let $H=(V,E)$ be a connected graph.
    \begin{itemize}
        \item We call a vertex $v \in V$ \emph{extendable} if $\MEDS(H-v)+1=\MEDS(H)$. We denote the set of extendable vertices of $H$ by $Q(H)$. 
        \helpfulcomment{Intuitively, these vertices allow a local solution for an $H$-component in $G-X$ that includes an edge $\{v,x\}$ with $x\in X$ and $v \in V(H)$.}
        \item We call a set $Y \subseteq Q(H)$ \emph{free} if for all vertices $v \in Y$ and for all minimum edge dominating sets $F$ in $H$ there exists a minimum edge dominating set $F'$ in $H-v$ of size $|F|-1$ and with $V(F) \setminus Y \subseteq V(F')$. By $W(H)$ we denote the unique maximum free set of $H$. We call a vertex $w \in W(H)$ \emph{free}. \footnote{We show in Proposition~\ref{proposition::properties}~(\ref{proposition::free}) that $W(H)$ is unique.}
        \helpfulcomment{Intuitively, vertices in $Y$ can be used for solution edges between components and $X$, while covering the same vertices of $H-Y$ as any local optimum solution; thus, they cannot be used for lower bounds like for $P_3$-components.}
        \item We call a vertex $v \in V$ \emph{uncovered} if no minimum edge dominating set $F$ of $H$ contains an edge incident with $v$, i.e.\ $v \notin V(F)$. We denote the set of uncovered vertices by $U(H)$.
        \helpfulcomment{Intuitively, $H$-components with any $v\in U(H)$ adjacent to $x\in X$ are easy to handle because $x\notin V(F)$ would imply that the local cost for $H$ increases above $\MEDS(H)$.}
        \item For any $Y\subseteq V$ define $\cost(Y):=|Y|+\MEDS(H-Y)-\MEDS(H)$.
        
        \helpfulcomment{Intuitively, $\cost(Y)$ is equal to the additional budget that is needed for an $H$-component of $G-X$ when exactly the vertices in $Y$ have solution edges to $X$. Note that $\cost(\{v\})=0$ for all extendable vertices $v$.}
        \item We call a set $B\subseteq V\setminus W(H)$ \emph{beneficial} if for all $\widetilde{B}\subsetneq B$ we have $|B|-\cost(B)>|\widetilde{B}|-\cost(\widetilde{B})$ or, equivalently, $\MEDS(H-B)<\MEDS(H-\widetilde{B})$. Note that this must also hold for $\widetilde{B}=\emptyset$ which implies that for all beneficial sets we have $|B|-\cost(B) >0$ or, equivalently, $\MEDS(H-B)<\MEDS(H)$.
        
        \helpfulcomment{Intuitively, the solution may include $|B|$ edges between $B$ and some $X'\subseteq X$ while increasing the cost for the $H$-component by exactly $\cost(B)$; this saves $|B|-\cost(B)>0$ over taking any $|B|$ edges incident with $X'$. The condition for all $\widetilde{B}\subsetneq B$ ensures that the savings of getting $|B|$ edges at cost $\cost(B)$ is greater than for any  proper subset.} 
        \item We call a beneficial set $B$ \emph{strongly beneficial} if $\cost(B) < \sum_{i=1}^h \cost(B_i)$ holds for all covers $B_1, B_2,\ldots,B_h \subsetneq B$ of $B$. 
        \helpfulcomment{Intuitively, for a strongly beneficial set $B$ we cannot get the same number of edges to $X$ by using sets $B_i$ in several different $H$-components.}
    \end{itemize}
\end{definition}

\begin{figure}
    \centering
    \begin{tikzpicture}[scale=1, transform shape] 
    \node[circle, inner sep=1.5pt, fill] (v1) at (1,1) [label=above:$b$] {};
    \node[circle, inner sep=1.5pt, fill] (v2) at (0,1) [label=above:$a$] {};
    \node[circle, inner sep=1.5pt, fill] (v3) at (0,0) [label=120:$f$] {};
    \node[circle, inner sep=1.5pt, fill] (v4) at (0,-1) [label=below:$k$] {};
    \node[circle, inner sep=1.5pt, fill] (v5) at (1,-1) [label=below:$l$] {};
    \node[circle, inner sep=1.5pt, fill] (w1) at (2.5,1) [label=above:$c$] {};
    \node[circle, inner sep=1.5pt, fill] (w2) at (3.5,1) [label=above:$d$] {};
    \foreach \i/\j/\k in {1/-1,2/0,3/1,4/2,5/3,6/4}
    \node[circle, inner sep=1.5pt, fill] (u1) at (-1,0) [label=below:$e$] {};
    \node[circle, inner sep=1.5pt, fill] (u2) at (0,0) {};
    \node[circle, inner sep=1.5pt, fill] (u3) at (1,0) [label=below:$g$] {};
    \node[circle, inner sep=1.5pt, fill] (u4) at (2,0) [label=below:$h$] {};
    \node[circle, inner sep=1.5pt, fill] (u5) at (3,0) [label=below:$i$] {};
    \node[circle, inner sep=1.5pt, fill] (u6) at (4,0) [label=below:$j$] {};
    \foreach \i/\j in {1/2,2/3,3/4,4/5}
        \draw (v\i) -- (v\j);
    \foreach \i/\j in {1/2,2/3,3/4,4/5,5/6}
        \draw (u\i) -- (u\j);
    \draw (w1) -- (u4);
    \draw (w1) -- (u5);
    \draw (w2) -- (u5);
    \draw (w2) -- (u6);
    \foreach \i in {1,2,4,5}
        \node[draw=Green, thick, circle, inner sep=3pt] at (v\i) {};
    \foreach \x/\i in {u/4,u/6,w/2}
        \node[draw=orange, thick, rectangle, inner sep=4pt] at (\x\i) {};
    \node[draw=Blue, thick, regular polygon, regular polygon sides=3, inner sep=2.5pt] at (u1) {};
    \draw[decorate,decoration={snake,amplitude=.4mm,segment length=1mm}] (v1) -- (v2);
    \draw[decorate,decoration={snake,amplitude=.4mm,segment length=1mm}] (v3) -- (v4);
    \draw[decorate,decoration={snake,amplitude=.4mm,segment length=1mm}] (w1) -- (u4);
    \draw[decorate,decoration={snake,amplitude=.4mm,segment length=1mm}] (u5) -- (u6);
    \end{tikzpicture}
    \caption{Example of an $H$-component with $\MEDS(H)=4$. The wavy edges are a possible minimum edge dominating set of $H$.}
    \label{figure::example_def}
\end{figure}
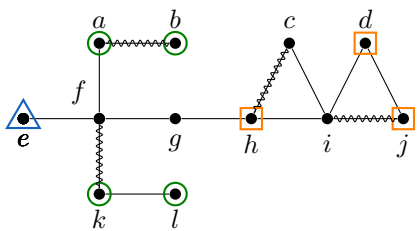

\begin{example}[Illustration of Definition \ref{definition:basicterms}]
    Figure~\ref{figure::example_def} shows a connected graph $H$. The size of an edge dominating set in $H$ is at least four because a solution has to dominate the four pairwise non-adjacent edges $\{a,b\}, \{k,l\}, \{j,d\}$ and $\{g,h\}$. Thus, $\MEDS(H)=4$ because the wavy edges are an edge dominating set of $H$. 
    
    The vertices $\{a,b,k,l\}$, marked with a green cycle, as well as the vertices $\{d,h,j\}$, marked with an orange rectangle, are extendable. But only the green marked vertices $\{a,b,k,l\}$ are free: Let $F$ be any minimum edge dominating set in $H$. The set $F$ must contain exactly one of the two edges $e_1=\{a,b\}$ and $e_2=\{a,f\}$, and exactly one of the two edges $e_3=\{k,l\}$ and $e_4=\{k,f\}$. Now, $F' = F \setminus \{e_1,e_2,e_3,e_4\} \cup \{f,k\}$ is an edge dominating set in $H-a$ and $H-b$ of size $|F|-1$, and $F' = F \setminus \{e_1,e_2,e_3,e_4\} \cup \{a,f\}$ is an edge dominating set in $H-k$ and $H-l$ of size $|F|-1$ which implies that the vertices $\{a,b,k,l\}$ are free. The vertices $\{d,h,j\}$ are not free because no minimum edge dominating set $F'$ in $H-d$, resp.\ $H-h$, resp.\ $H-j$ has vertex $c$, which is not extendable, as an endpoint of a solution edge, but the graph $H$ has a minimum edge dominating set that has $c$ as an endpoint, namely the one containing the wavy edge $\{a,b\},\{h,c\},\{d,j\}$.
    The vertex $e$, marked with a blue triangle, is uncovered. 
    
    The set $\{c,g\}$ is strongly beneficial, whereas the set $\{c,g,i,j\}$ is only beneficial, but not strongly beneficial: The set $\{c,g\}$ is beneficial because $\MEDS(H-\{c,g\})=3$ and $\MEDS(H-c)=\MEDS(H-g)=\MEDS(H) = 4$, and strongly beneficial because the only possible non-trivial cover of $\{c,g\}$ is $\{c\},\{g\}$ and $\cost(\{c,g\}) = 1 < 2 = \cost(\{c\}) + \cost(\{g\})$. The set $\{c,g,i,j\}$ is beneficial because $\MEDS(H-\{c,g,i,j\}) = 2$ and $\MEDS(H-B) \geq 3$ for all $B \subsetneq \{c,g,i,j\}$. But $\{c,g,i,j\}$ is not strongly beneficial because $\cost(\{c,g,i,j\}) = 2 = 1 + 1 + 0 = \cost(\{c,g\}) + \cost(\{i\}) + \cost(\{j\})$.
    Observe that the set $\{c,g,i\}$ is not beneficial even though $\MEDS(H-\{c,g,i\})=3 < 4=\MEDS(H)$, because $\{c,g\} \subsetneq \{c,g,i\}$ and $\MEDS(H-\{c,g,i\}) = 3 = \MEDS(H-\{c,g\})$.
\end{example}

We are now able to give a more detailed version of Theorem~\ref{theorem:classification}, which specifies for each finite set $\cH$ of connected graphs the kernelization complexity of \EDS parameterized by the size of a modulator to $\cH$-component graphs.

\begin{theorem}\label{theorem:detailedclassification}
Let $\cH$ be any finite set of connected graphs. The \EDS problem parameterized by the size of a modulator to $\cH$-component graphs behaves as follows:
\begin{enumerate}
 \item If $\cH$ contains any graph $H$ fulfilling one of the following items then there is no polynomial kernelization unless \containment:\label{item:lowerbound}
 \begin{enumerate}
  \item There is an extendable vertex in $H$ that is not free, i.e., $Q(H)\setminus W(H)\neq\emptyset$.\label{item:lowerbound:CP_Q}
  \item There is a strongly beneficial set $B$ in $H$ that contains an uncovered vertex, i.e., $B\cap U(H)\neq\emptyset$.\label{item:lowerbound:CP_U}
  \item There is a vertex in $H$ that is neither uncovered, free, nor neighbor of a free vertex, i.e., $V(H)\setminus (N[W(H)]\cup U(H))\neq\emptyset$.\label{item:lowerbound:CP_R}
  \item There is a strongly beneficial set $B\subseteq N(W(H))$ in $H$ such that no minimum edge dominating set $F_B$ of $H-B$ covers all vertices of $N(W(H))\setminus B$.\label{item:lowerbound:CP_other}
 \end{enumerate}
 \item Else, if $\cH$ contains at least one graph that has a strongly beneficial set, then there is a kernelization to $\Oh(|X|^d)$ vertices, $\Oh(|X|^{d+1})$ edges, and size $\Oh(|X|^{d+1}\log|X|)$, and there is no kernelization to size $\Oh(|X|^{d-\varepsilon})$, for any $\varepsilon>0$, unless \containment where $d$ is the size of the largest strongly beneficial set in any $H\in\cH$.\label{item:polynomialkernel}
 \item Else, there is a kernelization to $\Oh(|X|^2)$ vertices, $\Oh(|X|^3)$ edges, and size $\Oh(|X|^3\log|X|)$, and there is no kernelization to size $\Oh(|X|^{2-\varepsilon})$, for any $\varepsilon>0$, unless \containment.\label{item:smallkernel}
\end{enumerate}
\end{theorem}

The rest of this section is devoted to proving Theorem~\ref{theorem:detailedclassification}, following the proof outline below. From this, Theorem~\ref{theorem:classification} directly follows because disconnected graphs in $\cH$ do not affect the resulting class of $\cH$-component graphs, i.e., given \emph{any} finite set $\cH$ of graphs we can take the subset $\cH'$ of connected graphs in \cH and apply Theorem~\ref{theorem:detailedclassification} to $\cH'$. As an example for applying the theorem, for $\cH=\{P_3\}$ we get Item~\ref{item:lowerbound:CP_Q}, for $\cH=\{P_4\}$ we get Item~\ref{item:lowerbound:CP_U}, for $\cH=\{K_3\}$ and $\cH=\{K_5\}$ we get Item~\ref{item:lowerbound:CP_R}, and for $\cH=\{P_2\}=\{K_2\}$, $\cH=\{K_4\}$, $\cH=\{P_5\}$, as well as $\cH = \{E = \begin{tikzpicture}[scale=0.11, transform shape] 
    \node[draw, circle, inner sep=3.5pt, fill] (v1) at (1,1) {};
    \node[draw, circle, inner sep=3.5pt, fill] (v2) at (0,1) {};
    \node[draw, circle, inner sep=3.5pt, fill] (v3) at (0,0) {};
    \node[draw, circle, inner sep=3.5pt, fill] (v4) at (0,-1) {};
    \node[draw, circle, inner sep=3.5pt, fill] (v5) at (1,-1) {};
    \node[draw, circle, inner sep=3.5pt, fill] (u) at (1,0) {};
    \foreach \i/\j in {1/2,2/3,3/4,4/5}{
        \draw[line width=0.3] (v\i) -- (v\j);}
    \draw[line width=0.3] (v3) -- (u);
    \end{tikzpicture}\}$ we get Item~\ref{item:smallkernel}.
    
\begin{remark*}
    We showed that \EDS parameterized by the size of a given modulator $X$ to a $P_5$-component graph admits a kernel with $\Oh(|X|)$ vertices (see Theorem~\ref{theorem::kernel_P5}). The reason why we the kernelization procedure of Item~\ref{item:smallkernel} only reduces to $\Oh(|X|^2)$ vertices instead of $\Oh(|X|)$ vertices is that $H$-components can have uncovered vertices. This leads to a different marking argument similar to the case for \EDS parameterized by solution size. Note that EDS parameterized by solution size is covered by Item~\ref{item:smallkernel}.
\end{remark*}

\begin{proof}[Proof outline for Theorem~\ref{theorem:detailedclassification}.]
We begin by establishing a number of useful properties of the terms introduced in Definition~\ref{definition:basicterms}, e.g., that each graph $H$ containing a beneficial set $B$ also contains a strongly beneficial set $B'\subseteq B$ (Proposition~\ref{proposition::properties} (\ref{proposition::strongly})).

The kernelization lower bound of Item~\ref{item:lowerbound} is proved by generalizing the lower bound obtained for $P_3$-component graphs in Theorem~\ref{theorem::P3}. We define so-called \emph{control pairs} by abstracting properties of $P_3$-components used in the proof (Definition~\ref{definition:controlpair}) and show that there is no polynomial kernelization when any graph $H\in\cH$ has a control pair (Theorem~\ref{theorem::LB}). We then show that graphs $H$ fulfilling Items \ref{item:lowerbound:CP_Q}, \ref{item:lowerbound:CP_U}, \ref{item:lowerbound:CP_R}, or \ref{item:lowerbound:CP_other} have control pairs (Lemmas \ref{lemma::CP_Q}, \ref{lemma::CP_U}, \ref{lemma::CP_R}, and \ref{lemma::CP_other}). 

In Item \ref{item:lowerbound:CP_other}, and in the items below, we (may) use that no graph in $\cH$ fulfills Items \ref{item:lowerbound:CP_Q}, \ref{item:lowerbound:CP_U}, or \ref{item:lowerbound:CP_R}. Accordingly, each graph $H\in\cH$ has $V(H)=N[W(H)]\cup U(H)$, i.e., each vertex of $H$ is uncovered, free, or neighbor of a free vertex. Moreover, every extendable vertex is also free, i.e., $Q(H)=W(H)$, and strongly beneficial sets contain no (uncovered) vertices of $U(H)$. This implies that all strongly beneficial sets are subsets of $N(W(H))$, the neighborhood of the free vertices, as neither uncovered nor free vertices can be contained and no further vertices except those in $N(W(H))$ exist in $H$ (in this case).

For Item~\ref{item:polynomialkernel} we have that no graph in \cH fulfills any of the Items \ref{item:lowerbound:CP_Q} through \ref{item:lowerbound:CP_other} and that at least one graph in $\cH$ has a strongly beneficial set. Thus, in addition to the above restrictions on $H\in\cH$, we know that for each strongly beneficial set $B$, which here must be a subset of $N(W(H))$, there is a minimum edge dominating set $F_B$ of $H-B$ that covers all vertices in $N(W(H))\setminus B$. We give a general kernelization procedure that reduces the number of components in $G-X$ to $\Oh(|X|^d)$ where $d$ is the size of the largest strongly beneficial set among graphs $H\in\cH$ (Lemma \ref{lemma::kernel2}). We then rule out kernels of size $\Oh(|X|^{d-\varepsilon})$ using only $H$-components, where $H$ is any graph in $\cH$ that exhibits the largest size $d$ of strongly beneficial sets (Lemma \ref{theorem::polynomiallowerbound}). Note that in the present item $d$ is always at least two because having a strongly beneficial set $B$ of size one would mean that $v\in B$ is an extendable vertex that is not free (because beneficial sets are disjoint from the set $W(H)$ of free vertices), which is handled by Item~\ref{item:lowerbound:CP_Q}.

Finally, for Item~\ref{item:smallkernel}, it remains to consider the case that no graph $H\in\cH$ fulfills any of the Items \ref{item:lowerbound:CP_Q} through \ref{item:lowerbound:CP_other} and that no graph in $H$ has a strongly beneficial set. It follows that no graph in $H$ has any beneficial sets (Proposition~\ref{proposition::properties} (\ref{proposition::strongly})) and, as before, we have $V(H)=N[W(H)]\cup U(H)$. We obtain a kernelization to $\Oh(|X|^2)$ vertices, $\Oh(|X|^3)$ edges, and size $\Oh(|X|^3\log|X|)$ (Lemma \ref{lemma::kernel1}). The lower bound ruling out kernels of size $\Oh(|X|^{2-\varepsilon})$ for any $\varepsilon>0$, and in fact for any set $\cH$, follows easily by a simple reduction from \VC for which a lower bound ruling out size $\Oh(n^{2-\varepsilon})$ is known \cite{DellM14} (Lemma \ref{lemma::quadraticlowerbound}).

Since the arguments required for the kernelization in Item~\ref{item:smallkernel} are simpler than for that of Item~\ref{item:polynomialkernel} and can serve as an introduction to it, the proofs are given in the order of Item~\ref{item:lowerbound}, Item~\ref{item:smallkernel}, followed by Item~\ref{item:polynomialkernel}.
\end{proof}

Before starting on the lower bound part of Theorem~\ref{theorem:detailedclassification}, we establish a few basic properties of the terms defined in Definition~\ref{definition:basicterms}; these mostly follow readily from their definition. We also justify the definition of $W(H)$ as the unique maximum cardinality free set in $H$.

\begin{proposition}[\footnote{The proof of Proposition~\ref{proposition::properties} is deferred to Section \ref{section::proposition}.}] \label{proposition::properties}
  Let $H=(V,E)$ be a connected graph, let $W=W(H)$ be the set of free vertices, let $Q=Q(H)$ be the set of extendable vertices, and let $U=U(H)$ be the set of uncovered vertices.
  \begin{enumerate}
   \item The set $W$ is well defined. \label{proposition::free}
   \item The set $U$ is an independent set and no vertex in $Q$ is adjacent to a vertex in $U$; hence $N_H(U) \cap (Q \cup U ) = \emptyset$. \label{proposition::neighborhood}
   \item If $v \in N_H(U)$ is a vertex that is adjacent to a vertex in $U$, then $v$ is an endpoint of an edge in every minimum edge dominating set of $H$. \label{proposition::N(U)}
   \item It holds for all vertices $v \in V$ that $\MEDS(H)-1 \leq \MEDS(H-v) \leq \MEDS(H)$. \label{proposition::vertex}
   \item Let $Y \subseteq V$. It holds for all subsets $X \subseteq Y$ that $\MEDS(H-X) - |Y \setminus X| \leq \MEDS(H-Y) \leq \MEDS(H-X)$, and that $\cost(X) \leq \cost(Y)$. \label{proposition::monoton}
   \item Let $F$ be a minimum edge dominating set in $H$. There exists a minimum edge dominating set $F'$ in $H$ with $(V(F) \cup N_H(W) ) \setminus W \subseteq V(F')$. \label{proposition::N(W)IN} 
   \item Every set that consists of one vertex $v \in Q \setminus W$ is strongly beneficial. Furthermore, these are the only beneficial sets of size one. \label{proposition::Q-W} 
   \item If $B$ is a beneficial set of size at least two then $B$ contains no extendable vertex; hence $B \cap Q=\emptyset$. \label{proposition::disjointQ} 
   \item If there exists a set $Y \subseteq V \setminus W$ with $\MEDS(H-Y) < \MEDS(H)$, then there exists a beneficial set $B \subseteq Y$ with $\MEDS(H-B) = \MEDS(H-Y)$. \label{proposition::BFS}
   \item If there exists a set $Y \subseteq V \setminus W$ with $\MEDS(H-Y) < \MEDS(H)$, then there exists a beneficial set $B \subseteq Y$ with $\MEDS(H-B) +1 = \MEDS(H)$. Furthermore, $B$ is strongly beneficial. \label{proposition::BFScost1}
   \item If $H$ has a beneficial set $B$, then $H$ has also a strongly beneficial set $B' \subseteq B$. \label{proposition::strongly}
   \item Let $F$ be a minimum edge dominating set in $H$. If $e=\{x,y\}$ is an edge in $F$ with $x, y \notin Q$, then $\{x,y\}$ is a strongly beneficial set. \label{proposition::BFSedge}
   \item Let $B$ be a beneficial set. $B$ is strongly beneficial if and only if for every non-trivial partition $B_1,B_2,\ldots,B_h$ of $B$ it holds that $\cost(B) < \sum_{i=1}^h \cost(B_i)$. \label{proposition::definition}
   \item Let $Y \subseteq V \setminus W$. There exists a partition $B_1,B_2,\ldots,B_h$ of $Y$ where  $B_i$ is either strongly beneficial or where $B_i$ has $\cost(B_i)=|B_i|$, for all $i \in [h]$, such that $\cost(Y) \geq \sum_{i=1}^h \cost(B_i)$. (Note that we also allow trivial partitions.) \label{proposition::decomposeB}
  \end{enumerate}
\end{proposition}

\subsection[Generalizing the lower bound obtained for P3-component graphs]{Generalizing the lower bound obtained for $\boldsymbol{P_3}$-component graphs}

We want to generalize Theorem \ref{theorem::P3} to get a lower bound that covers a variety of different $H$-components. In the proof of Theorem \ref{theorem::P3}, we used one endpoint of each $P_3$ to control that we choose the edges only from one instance and to make sure that the $\binom{k}{2}$ edges have their endpoints in a set of size $k$. The middle vertex of each $P_3$ is extendable (but not free, so $P_3$ fits Item~\ref{item:lowerbound:CP_Q} of Theorem~\ref{theorem:detailedclassification}) and the set consisting of this single vertex is beneficial. Hence, we were able to add edges between the middle vertices of the $P_3$'s and the set $S$ without spending more budget. Accordingly, to generalize Theorem~\ref{theorem::P3}, we define what we call control pairs consisting of a set of control vertices and a beneficial set.

\begin{definition}\label{definition:controlpair}
 Let $H=(V,E)$ be a connected graph and let $B \subseteq V$, $C \subseteq V \setminus (Q(H) \cup B)$. We call the pair $(C,B)$ \emph{control pair}, if 
 \begin{itemize}
    \item $B$ is strongly beneficial,
    \item no vertex $c \in C$ is extendable in $H-B$, i.e., $C \cap Q(H-B) = \emptyset$,
    \item there exists a minimum edge dominating set $F$ in $H$ such that $C \subseteq V(F)$, and
    \item for all minimum edge dominating sets $F_B$ in $H-B$ it holds that $C \nsubseteq V(F_B)$.
 \end{itemize}
\end{definition}

Let $H$ be a connected graph that contains a control pair. We show that \EDS parameterized by the size of a modulator to an $\cH$-component graph has no polynomial kernel unless \containment for all $\cH\ni H$. The lower bound construction generalizes the construction used for Theorem \ref{theorem::P3}, making the proof more complicated. Observe that for $H=P_3=v_1v_2v_3$ the set $B$ is the vertex $v_2$ and the set $C$ is the vertex $v_1$ (or $v_2$).

\begin{theorem} \label{theorem::LB}
 Let $H$ be a connected graph and let $B \subseteq V$, $C \subseteq V \setminus (Q(H) \cup B)$ such that $(C,B)$ is a control pair. For all sets $\cH\ni H$ of graphs, the \EDS problem parameterized by a modulator to an $\cH$-component graph admits no polynomial kernelization unless \containment.
\end{theorem}

\begin{proof}
    We give a cross-composition from \prob{Multicolored-Clique}; the theorem then follows directly from Theorem~\ref{theorem:kernellowerbound:crosscomposition}. In $G-X$ we use only components isomorphic to $H$, so $X$ is a modulator to $\cH$-component graphs for all $\cH$ with $H\in\cH$.

    We choose the same polynomial equivalence relation $\mathcal{R}$ as in the proof of Theorem \ref{theorem::P3}. Assume that we are given a sequence $I_i=(G_i,k)_{i=1}^t$ of \MCC instance that are in the same equivalence class of $\mathcal{R}$. Since all color classes have the same size we identify the vertex sets of each color class. Let $V$ be the vertex set (of size $k \cdot n$) of the $t$ instances and let $V_1,V_2,\ldots, V_k$ be the different color classes (each of size $n$). We assume w.l.o.g.\ that every instance has at least one edge in $E(V_p,V_q)$ for all $1 \leq p < q \leq k$; otherwise, this instance would be a trivial no instance and we can delete it. Furthermore, we can assume w.l.o.g.\ that $t=2^s$.
    
    We construct an instance $(G',k',X')$ of \EDS parameterized by a modulator to an $H$-component graph; thus $X'$ is a modulator to $\cH$-component graphs for all \cH with $H\in\cH$. As in the proof of Theorem \ref{theorem::P3} we add sets $V$, $T=\{t_1,t_2,\ldots,t_k\}$, $T'=\{t_1',t_2',\ldots,t_k'\}$, $W$, $Z$, and $Z'$ to $G'$ and connect them in the same way. Again, for each instance $G_i$, with $i \in [t]$, we pick a different subset of size $s$ of $W$ and denote it by $W(i)$.
    
    Instead of adding one vertex corresponding to each edge set $E(V_p,V_q)$ with $1\leq p < q \leq k$, we add a set $S_{p,q}$ of size $d:=|B|$ to $G'$ as well as a copy $S_{p,q}'$. Let $S_{p,q}=\{s_{p,q}^1,s_{p,q}^2,\ldots,s_{p,q}^{d}\}$, let $S_{p,q}'=\{s'^1_{p,q},s'^2_{p,q},\ldots,s'^d_{p,q}\}$, let $S=\bigcup_{1\leq p < q \leq k} S_{p,q}$, and let $S'=\bigcup_{1\leq p < q \leq k} S_{p,q}'$. We make every vertex $s_{p,q}^j$ with $1 \leq p < q \leq k$ and $j \in [d]$ adjacent to vertex $s'^j_{p,q}$.
    For each graph $G_i$, with $i \in [t]$, we add $|E(G_i)|$ copies of graph $H$ to $G'$. We denote by $H_i^e$ the copy of $H$ that represents edge $e=\{x,y\} \in E(G_i)$ of instance $i \in [t]$ and by $(C_i^e,B_i^e)$ the control pair of $H_i^e$. For $e=\{x,y\}$ in instance $i$, we make every vertex in $C_i^e$ adjacent to all vertices in $W(i)$ and to the vertices $x,y \in V$. 
    
    To be able to refer to single vertices of $B$ in copies $H_i^e$ of $H$, let $B=\{b_1,\ldots,b_d\}$ and let $B_i^e=\{b_{i,1}^e,b_{i,2}^e,\ldots,b_{i,d}^e\}$ where $b_{i,j}^e$ corresponds to $b_j$ in $B$ (i.e., this correspondence constitutes an isomorphism between $H$ and $H_i^e$).
    For all $i \in [t]$, $e \in E(G_i)$, and $j \in [d]$ we make vertex $b_{i,j}^e \in B_i^e$ adjacent to vertex $s_{p,q}^j \in S_{p,q}$ if $e \in E(V_p,V_q)$ for $1\leq p < q \leq k$. Note that every vertex $b_{i,j}^e$ is adjacent to exactly one vertex in $S$.
    
     The modulator $X'$ contains all vertices that are not contained in a copy of $H$; thus $X'=V(G') \setminus \bigcup_{i=1}^t \bigcup_{e \in E(G_i)} V(H_i^e) = V \cup T \cup T' \cup W \cup Z \cup Z' \cup S \cup S'$ and $X'$ has size $|X'|=k \cdot n+2k+4s+2 \cdot \binom{k}{2} |B|$. Let $k'=s+k+\sum_{i=1}^t |E(G_i)| \cdot \MEDS(H) + \binom{k}{2} \cdot \cost(B)$. Note that $0 \leq \cost(B) < |B|$.
    
    We will show that $(G',k',X')$ is a $\mathrm{YES}$-instance of \EDS if and only if there exists an $i^* \in [t]$ such that $(G_{i^*},k)$ is a $\mathrm{YES}$-instance of \MCC. 
    
    $(\Rightarrow:)$ Let $F$ be an edge dominating set of size at most $k'$ in $G'$. 
    Analogously to the proof of Theorem \ref{theorem::P3} we can observe, that $Z \cup T \cup S \subseteq V(F)$ and that we can choose the $|T|$ resp.\ $|Z|$ resp.\ $|S|$ edges that have one endpoint in $T$ resp.\ $Z$ resp.\ $S$ always from the edge set $E(T,V)$ resp.\ $E(Z,W)$ resp.\ $E(S, \bigcup_{i=1}^t \bigcup_{e \in E(G_i)} H_i^e)$. Additionally, we can assume w.l.o.g.\ that the edges in $F \cap E(T,V)$ resp.\ $F \cap E(Z,W)$ resp.\ $F \cap E(S, \bigcup_{i=1}^t \bigcup_{e \in E(G_i)} H_i^e)$ are a matching (simple replacement argument). Furthermore, an edge in $F$ cannot have its endpoints in different copies of $H$.
    
    Let $F_T=F \cap E(T,V)$, let $F_Z=F \cap E(Z,W)$, let $F_S=F \cap E(S, \bigcup_{i=1}^t \bigcup_{e \in E(G_i)} H_i^e)$, and let $F_R=F \setminus (F_T \cup F_Z \cup F_S)$. Recall, the edge sets $F_T$, $F_Z$ and $F_S$ are matchings in $G'$.
    We can assume that every edge in $F_R$ has at least one endpoint in $V(G') \setminus X' = \bigcup_{i=1}^t \bigcup_{e \in E(G_i)} V(H_i^e)$, because every edge in $E(G'[X'])$ is dominated by an edge in $F_T \cup F_Z \cup F_S$. This follows from the fact that $Z \cup T \cup S$ are covered by the edge dominating set $F$ and that $X' \setminus (Z \cup T \cup S)$ is an independent set.
    
    Let $B_i^e(F) = \{v \in B_i^e \mid \exists s \in S \colon \{v,s\} \in F_S\}$ be the set of vertices in $B_i^e$ that are incident with an edge in $F_S$, where $i \in [t]$ and $e \in E(G_i)$. 
    Since $F_S$ is a matching between $S$ and $\bigcup_{i=1}^t \bigcup_{e \in E(G_i)} H_i^e$ which covers $S$ and contains all edges of $F$ that are incident with $S$, and since $|F_S|=|S|$ it holds that no edge in $F \setminus F_S$ has an endpoint in $S$ and that $\sum_{i=1}^t \sum_{e \in E(G_i)} |B_i^e(F)| = |S| = \binom{k}{2} \cdot |B|$.
    For $i \in [t]$ and $e \in E(G_i)$, let $F_i^e = \{ f \in F_R \mid f \cap V(H_i^e) \neq \emptyset\}$ be the set of edges in $F_R$ that have at least one endpoint in $H_i^e$. It holds that $\bigcup_{i=1}^t \dot{\bigcup}_{e \in E(G_i)} F_i^e =F_R$, because every edge in $F_R$ is incident with a vertex in $\bigcup_{i=1}^t \bigcup_{e \in E(G_i)} V(H_i^e)$. Moreover, the sets $F_i^e$ are a partition of $F_R$ because no edge is incident with vertices of different graphs $H_i^e$.
    
    Every edge in the graph $H_i^e$, with $i \in [t]$ and $e \in E(G_i)$, must be dominated by edges in $F_R$ and $F_S$, because they cannot be dominated by edges in $F_T \cup F_Z\subseteq E(T,V)\cup E(Z,W)$. Thus, $F_i^e$ must dominate all edges in $H_i^e - B_i^e(F)$, because the set $B_i^e(F)$ contains all vertices in $H_i^e$ that are incident with edges in $F_S$. This implies that for all $i \in [t]$ and $e \in E(G_i)$, the set $F_i^e$ has at least the size of a minimum edge dominating set in $H_i^e - B_i^e(F)$; hence $\MEDS(H_i^e - B_i^e(F) ) \leq |F_i^e|$ .
    Combining all this, we get
    \begin{align*}
        |F| &= |F_T| + |F_Z| + |F_S| + |F_R| \\
            &= k + s + \binom{k}{2} |B| + \sum_{i=1}^t \sum_{e \in E(G_i)} |F_i^e| \quad\quad\quad\quad\quad\quad \text{// because the sets $F_i^e$ partition $F_R$}\\
            &\geq k + s + \binom{k}{2} |B| + \sum_{i=1}^t \sum_{e \in E(G_i)} \MEDS(H_i^e - B_i^e(F) ) \quad\quad \text{// bc.\ } \MEDS(H_i^e - B_i^e(F) ) \leq |F_i^e| \\
            &= k + s + \binom{k}{2} |B| + \sum_{i=1}^t \sum_{e \in E(G_i)} \left(\MEDS(H_i^e) - |B_i^e(F)| + \cost(B_i^e(F))\right) \quad \text{// def.\ } \cost(B_i^e(F)) \\
            &=k + s + \binom{k}{2} |B| + \sum_{i=1}^t |E(G_i)|  \MEDS(H) - \binom{k}{2} |B| + \sum_{i=1}^t \sum_{e \in E(G_i)} \cost(B_i^e(F)) \\
            &= k + s + \sum_{i=1}^t |E(G_i)|  \MEDS(H) + \sum_{i=1}^t \sum_{e \in E(G_i)} \cost(B_i^e(F)) \\
            &= k + s + \sum_{i=1}^t |E(G_i)|  \MEDS(H) + \sum_{1 \leq p < q \leq k} \sum_{\substack{x \in V_p\\ y \in V_q}} \sum_{\substack{i\in[t]\\ \{x,y\} \in E(G_i)}} \cost(B^{\{x,y\}}_i(F))\\
    \intertext{Because we also have that}
        |F| &\leq k'=k+s+ \sum_{i=1}^t |E(G_i)|  \MEDS(H)+ \binom{k}{2}  \cost(B), \\
    \end{align*}
    it follows directly that
    \begin{align}
        \sum_{1 \leq p < q \leq k} \sum_{\substack{x \in V_p\\ y \in V_q}} \sum_{\substack{i\in[t]\\ \{x,y\} \in E(G_i)}} \cost(B^{\{x,y\}}_i(F))
        \leq \binom{k}{2} \cdot \cost(B) \text{.} \label{align::minimal_beneficial}
    \end{align}
    Therefore, there exist $1 \leq \bar{p} < \bar{q} \leq k$ such that
    \begin{align*}
        \sum_{\substack{x \in V_{\bar{p}}\\ y \in V_{\bar{q}}}} \sum_{\substack{i\in[t]\\ \{x,y\} \in E(G_i)}} \cost(B^{\{x,y\}}_i(F)) \leq \cost(B) \text{.}
    \end{align*}
    \begin{claim} \label{claim::cover}
        The set $B_C:=\bigcup_{\substack{x \in V_{\bar{p}}\\ y \in V_{\bar{q}}}} \dot\bigcup_{\substack{i \in [t]\\ \{x,y\} \in E(G_i)}} B_i^{\{x,y\}}(F)$ contains exactly one copy of every vertex in $B$.
    \end{claim}
    \begin{claimproof}
        Let $j \in [|B|]$. The vertex $s_{\bar{p},\bar{q}}^j \in S_{\bar{p},\bar{q}}$ is endpoint of an edge $f$ in $F_S$. Let $v$ be the other endpoint of this edge; hence $f=\{s_{\bar{p},\bar{q}}^j,v\}$. Since $F_S \subseteq E(S,\bigcup_{i=1}^t \bigcup_{e \in E(G_i)} B_i^e)$ and $s^j_{\bar{p},\bar{q}}$ is only adjacent to vertices in $\bigcup_{i=1}^t \bigcup_{e \in E(G_i)} B_i^e$ that correspond to vertex $b_j$ in $B$, it holds that $v=b^{\{x,y\}}_{i^*,j}$ for some $i^* \in [t]$, $x \in V_{\bar{p}}$, $y \in V_{\bar{q}}$ with $\{x,y\} \in E(G_{i^*})$. Thus, $B_C$ contains at least one copy of $b_j$ for all $j \in [|B|]$. 
        
        Assume, $B_C$ contains at least two copies of a vertex $b_j$ in $B$, with $j \in [|B|]$. Let $i_1,i_2 \in [t]$, let $e_1 \in E(V_{\bar{p}},V_{\bar{q}}) \cap E(G_{i_1})$, and let $e_2 \in E(V_{\bar{p}},V_{\bar{q}}) \cap E(G_{i_2})$ such that $b^{e_1}_{i_1,j}$, $b^{e_2}_{i_2,j}$ are contained in $B_C$ and either $i_1 \neq i_2$ or $e_1 \neq e_2$. Since both vertices $b^{e_1}_{i_1,j}$ and $b^{e_2}_{i_2,j}$ are only adjacent to vertex $s_{\bar{p},\bar{q}}^j \in S_{\bar{p},\bar{q}}$, the set $F_S$ is not a matching, which is a contradiction and proves the claim. 
    \end{claimproof}
    
    Let $B_1, B_2,\ldots, B_h \subseteq B$ be the sets in $B$ that correspond to the nonempty sets in $\{ B_i^{\{x,y\}}(F) \mid x \in V_{\bar{p}}, y \in V_{\bar{q}}, i \in [t] \colon \{x,y\} \in E(G_i) \}$. That is, for each nonempty set $B_i^{\{x,y\}}(F)$, with $i\in[t]$, there is a set $B_r=\{b_j \mid b_{i,j}^{\{x,y\}}\in B_i^{\{x,y\}}(F)\}$.
    It holds that the sets $B_1,B_2,\ldots, B_h$ are a partition of $B$ (Claim \ref{claim::cover}) and that $\sum_{r=1}^h \cost(B_r) \leq \cost(B)$ (inequality (\ref{align::minimal_beneficial})). This implies that $h=1$: Otherwise $B_1,B_2, \ldots, B_h$ would be a non-trivial partition of $B$ with $\cost(B) \geq \sum_{r=1}^h \cost(B_i)$ which implies that $B$ is not strongly beneficial (Proposition~\ref{proposition::properties}~(\ref{proposition::definition})). 
    
    Thus, there exists exactly one vertex $x \in V_{\bar{p}}$, exactly one vertex $y \in V_{\bar{q}}$, and exactly one $i^* \in [t]$ with the property that $\{x,y\} \in E(G_{i^*})$ and that $B_{i^*}^{\{x,y\}}(F)$ is not the empty set. Furthermore, it holds that $B_{i^*}^{\{x,y\}}(F) = B_{i^*}^{\{x,y\}}$ (Claim \ref{claim::cover}).
    It follows that for each graph $H_i^e$ with $i \in [t]$ and $e \in E(G_i)$ either $B_i^e(F) = B_i^e$ or $B_i^e(F) = \emptyset$.
    Therefore, for all $1 \leq p < q \leq k$ there exists exactly one edge $e \in E(V_p,V_q)$ and exactly one $i \in [t]$ such that $e \in E(G_i)$ and $B_i^e(F) = B_i^e$: Either all $1 \leq p < q \leq k$ fulfill inequation \ref{align::minimal_beneficial} with equality or there exist $1 \leq p'<q' \leq k$ such that inequation \ref{align::minimal_beneficial} holds with "$<$". This would imply that $B$ is not beneficial (see proof of Claim \ref{claim::cover} and definition of strongly beneficial).
    Consequently, the edges in $F_S$ are incident with $\binom{k}{2}$ different copies of $H$ and cover the copy of $B$ in these $\binom{k}{2}$ copies.
    
    Now, we consider which vertices are contained or not contained in $V(F_R)$. 
    It holds that $|F_R| \leq \sum_{i=1}^t |E(G_i)| \cdot \MEDS(H) - \binom{k}{2} (\cost(B)-|B|)$; to see this consider $k'\geq |F|=|F_S|+|F_T|+|F_Z|+|F_R|$.
    For each graph $H_i^e$, with $i \in [t]$ and $e\in E(G_i)$, with $B_i^e(F) = \emptyset$ we need at least $\MEDS(H)$ edges to dominate all edges in $E(H_i^e)$ and for each graph $H_i^e$, here $i \in [t]$ and $e\in E(G_i)$, with $B_i^e(F) = B_i^e$ we need at least $\MEDS(H-B)$ edges to dominate all edges in $E(H_i^e-B_i^e)$. Since no two different copies of $H$ are adjacent, since $\MEDS(H-B)=\MEDS(H)-|B|+\cost(B)$, and since there are exactly $\binom{k}{2}$ copies of $H$ where the vertices that correspond to vertices in $B$ are covered by edges in $F_S$, it holds that we have exactly $\MEDS(H)$ resp.\ $\MEDS(H-B)$ edges of $F_R$ to dominate all edges in $H_i^e$ resp.\ $H_i^e - B_i^e$.
    
    It follows, that $F_R$ contains no edge that has one endpoint in $W$ or $V$, because the vertex sets $V$ and $W$ are only adjacent to copies of vertices in $C$ in $G-X$ and a vertex $c \in C$ is neither extendable in $H$ nor in $H-B$. Hence, without using more than $\MEDS(H)$ resp.\ $\MEDS(H-B)$ edges to dominate all edges in $H$ resp.\ $H-B$ we cannot have an edge in $F_R$ that has one endpoint in $V \cup W$.
    
    Let $i \in [t]$ and $e=\{x,y\} \in E(G_i)$ such that $B_i^e(F)=B_i^e$. It holds that $|F_i^e| = \MEDS(H-B)$, and therefore $C_i^e \nsubseteq V(F_i^e)$ (by definition of a control pair). Furthermore, $C_i^e \cap B_i^e$ is empty by the choice of $C$ and $B$. This implies that $N(C_i^e) = \{x,y\} \cup W(i) \subseteq V(F)$, i.e., that all neighbors of $C_i^e$ in $V$ and $W$ must be endpoints of $F$. It holds that $\{x,y\} \subseteq V(F_T)$ and $W(i) \subseteq V(F_Z)$ because neither edges in $F_S$ nor in $F_R$ have endpoints in $V\cup W$. Since, $|F_Z|=|W(i)|$ and $F_Z \subseteq E(Z,W)$ it follows that $V(F_Z) \cap W = W(i)$. Thus, all graphs $H_i^e$, here $i \in [t]$ and $e \in E(G_i)$, with $B_i^e(F)=B_i^e$ must belong to the same instance (because only edges in $F_Z$ contain edges that have endpoints in $W$ and because $F_Z$ has size $s=|W(i)|$).
    Let $i^* \in [t]$ be the number of this instance and let $e=\{x,y \} \in E(G_{i^*})$ such that $B_{i^*}^e(F)=B_{i^*}^e$. It must hold that $x,y \in V(F_T)$ (because no other edges in $F$ have an endpoint in $V$). Since $|F_T|=k$ and $F_T \subseteq E(T,V)$ it follows that the $\binom{k}{2}$ graphs $H^e_{i^*}$, for $e \in E(G_{i^*})$, with $B_{i^*}^e(F)=B_{i^*}^e$ must correspond to a set of edges in $E(G_{i^*})$ that have their endpoints in the set $X=V(F_T) \cap V$ of size $k$. 
    Consequently, the set $X$ is a clique in $G_{i^*}$ and $(G_{i^*},k)$ is a $\mathrm{YES}$-instance of \prob{Multicolored-Clique}.
    
    $(\Leftarrow:)$ This direction of the correctness proof is similar to the corresponding one in the proof of Theorem \ref{theorem::P3} and follows easily from the construction; we sketch this only briefly. If $(G_{i^*},k)$ is yes and $X=\{x_1,\ldots,x_k\}\subseteq V$ is a multicolored $k$-clique in $G_{i^*}$ with $x_j\in V_j$, for all $j \in [k]$, then select $k+s$ solution edges for $F$ in $E(T,V)$ and $E(Z,W)$ as for Theorem~\ref{theorem::P3}. In particular, this ensures that $W(i^*)\cup X\subseteq V(F)$. For each set $S_{p,q}$ with $1\leq p<q\leq k$ add the edges between the $d=|B|$ vertices between $S_{p,q}$ and the copy $B_{i^*}^{\{x_p,x_q\}}$ of $B$ in $H_{i^*}^{\{x_p,x_q\}}$. At this point, we have used up the budget (intuitively) intended for edges incident with $S$, $T$, and $Z$. All edges in $E(G[X'])$ are already dominated by $F$ as well as all edges incident with $S$; all edges in graphs $H_i^e$ and some edges between those graphs and $V\cup W$ remain.
    
    In each graph $H_{i^*}^{\{x_p,x_q\}}$ for $1\leq p<q\leq k$ select an edge dominating set for $H_{i^*}^{\{x_p,x_q\}}-B_{i^*}^{\{x_p,x_q\}}$ of cost $\MEDS(H-B)$. Together with previously added edges incident with $B_{i^*}^{\{x_p,x_q\}}$, this dominates all edges in this copy of $H$. Furthermore, edges between $H_{i^*}^{\{x_p,x_q\}}$ and $V\cup W$ are already dominated because their other endpoints are in $X\cup W(i^*)\subseteq V(F)$. For all other graphs $H_i^e$, i.e. with $i\neq i^*$ or with $i=i^*$ but $e\notin E(G[X])$ we can select an edge dominating set of size $\MEDS(H)$ that is incident with $C_i^e$. This dominates all edges in this $H$-graph and, crucially, dominates all edges between $H_i^e$ and $V\cup W$ because their endpoints in $H_i^e$ are all in $C_i^e$. (This is the only place where we need the third property of control pairs.) Thus we have selected an edge dominating set and it can be readily checked that we have picked exactly $k'=s+k+\sum_{i=1}^t |E(G_i)| \cdot \MEDS(H) + \binom{k}{2} \cdot \cost(B)$. Note that $\binom{k}{2} \cdot \cost(B)$ is exactly the additional cost of selecting $\binom{k}{2}$ times an edge dominating set for $H_{i^*}^{\{x_p,x_q\}}-B_{i^*}^{\{x_p,x_q\}}$ and $|B|$ edges between $B_{i^*}^{\{x_p,x_q\}}$ and $S_{p,q}$ rather than the optimum solution for $H_{i^*}^e$.
\end{proof}

Theorem \ref{theorem::LB} implies that whenever we have a family of connected graphs $\mathcal{H}$ which contains at least one graph $H \in \mathcal{H}$ that has a control pair, then \EDS parameterized by the size of a modulator to an $\mathcal{H}$-component graph does not have a polynomial kernel unless \containment. 
Now, one can ask which connected graphs have a control pair and whether \EDS parameterized by the size of a modulator to an $\mathcal{H}$-component graph has a polynomial kernel when no graph in \cH has a control pair.
First, we show which connected graphs have a control-pair. In a second step, we give a polynomial kernel for all remaining connected graphs of constant size.

\begin{lemma} \label{lemma::CP_Q}
Every connected graph $H$ that has at least one extendable vertex that is not free contains a control pair.
\end{lemma}
\begin{proof}
  Since $H$ has an extendable vertex that is not free, there exists a vertex $v \in Q(H) \setminus W(H)$ and a minimum edge dominating set $F$ in $H$ such that for every minimum edge dominating set $F'$ in $H-v$, which has size $|F|-1$, it holds that $V(F) \setminus Q(H) \nsubseteq V(F')$: If there would be no such vertex $v \in Q(H) \setminus W(H)$, then it would hold that for all $v \in Q(H) \setminus W(H)$ and for all minimum edge dominating sets $F$ in $H$ there exists a minimum edge dominating set $F'$ in $H-v$ of size $|F|-1$ and with $V(F) \setminus Q(H) \subseteq V(F')$. This holds also for all vertices in $W(H)$ and would imply that $Q(H)$ is free. Now, let $v \in Q(H) \setminus W(H)$ be a vertex that is extendable but not free, and let $F$ be a minimum edge dominating set in $H$ such that there exists no minimum edge dominating set $F'$ (of size $|F|-1$) in $H-v$ with $V(F) \setminus Q(H) \subseteq V(F')$.
  Let $C=V(F) \setminus Q(H)$ and $B=\{v\}$. Since $B=\{v\} \subseteq Q(H)$ it holds that $C \subseteq V \setminus (Q(H) \cup B)$. We will show that $(C,B)$ is a control pair. The set $B$ is strongly beneficial, because $v$ is extendable and not free (Proposition~\ref{proposition::properties}~(\ref{proposition::Q-W})). Furthermore, by construction, it holds that $F$ is a minimum edge dominating set with $C \subseteq V(F)$ and that for all minimum edge dominating set $F'$ in $H-B=H-v$ it holds that $C \nsubseteq V(F')$ (choice of $v$ and $C$). 
  
  It remains to show that no vertex in $C$ is extendable in $H-v$. Assume for contradiction that there exists a vertex $c \in C \cap Q(H-v)$. We will show that this implies that $c$ is also extendable in $H$, i.e., $c\in Q(H)$, which is a contradiction to the choice of $C$. Let $F_c$ be a minimum edge dominating set in $H-v-c$. Since $c$ is extendable in $H-v$, it holds that $|F_c|=\MEDS(H-v-c) = \MEDS(H-v)-1$. If $N_H(v) \setminus \{c\} = \emptyset$, then $F_c$ is also a minimum edge dominating set in $H-c$ and it would follow that $\MEDS(H-c)\leq |F_c|=\MEDS(H-v)-1 = \MEDS(H)-2$; but $\MEDS(H-c) \geq \MEDS(H)-1$ (Proposition~\ref{proposition::properties}~(\ref{proposition::vertex})), a contradiction. Thus, $N_H(v) \setminus \{c\} \neq \emptyset$ and we pick an arbitrary vertex $w\in N_H(v) \setminus \{c\}$. Now, $F_c \cup \{ \{v,w\} \}$ would be an edge dominating set in $H-c$ of size $\MEDS(H)-1$, hence $c$ would be extendable in $H$, which is a contradiction to the choice of $C$.
\end{proof}

\begin{lemma} \label{lemma::CP_U}
Every connected graph $H$ that has a strongly beneficial set that contains at least one uncovered vertex contains a control pair.
\end{lemma}
\begin{proof}
 Let $B$ be a strongly beneficial set in $H$ that contains at least one uncovered vertex and let $C=N_H(B \cap U(H)) \setminus B$ be the neighborhood of all uncovered vertices in $B$ without the vertices in $B$. Since $H$ is connected and every vertex in $U(H)$ has no neighbor in $U(H)$ or $Q(H)$ (Proposition~\ref{proposition::properties}~(\ref{proposition::neighborhood})), it holds that $C \subseteq V(H) \setminus Q(H)$. Furthermore, $C$ is not the empty set. Otherwise, if all neighbours of $B \cap U(H)$ are also contained in $B$ then $\MEDS(H-B)=\MEDS(H-(B \setminus U(H)))$. This implies that $B$ is not beneficial which is a contradiction. We will show that $(C,B)$ is a control pair. The set $B$ is strongly beneficial and the set $C$ is a subset of $V(H) \setminus (Q(H) \cup B)$ (by choice). It holds that every minimum edge dominating set in $H$ contains $C$, because $C \subseteq N(U(H))$ (Proposition~\ref{proposition::properties}~(\ref{proposition::N(U)})). Moreover, there exists no minimum edge dominating set $F'$ in $H-B$ such that $C \subseteq V(F')$: If not then such a set $F'$ would also be a minimum edge dominating set in $H - (B \setminus U(H))$ because $U(H) \cap B\subseteq U(H)$ is an independent set (Proposition~\ref{proposition::properties}~(\ref{proposition::neighborhood})) and all neighbors of $U(H) \cap B$ are contained in $C$. But, this implies that $B$ is not beneficial because $\MEDS(H - \widetilde{B}) \leq \MEDS(H-B)$ where $\widetilde{B}=B \setminus U(H)\subsetneq B$.
 
 Next, we show that no vertex in $C$ is extendable in $H-B$. Assume that there exists a vertex $c \in Q(H-B) \cap C$ that is extendable in $H-B$. Let $F'$ be a minimum edge dominating set in $H-B-c$; hence $|F'|=\MEDS(H-B)-1$. Since $c \in C = N_H(B \cap U(H))$, there exists a vertex $b \in B \cap U(H)$ with $\{c,b\} \in E(H)$. Now, $F' \cup \{\{c,b\}\}$ is an edge dominating set in $H-(B \setminus \{b\})$ of size $|F'|+1 = \MEDS(H-B)$. Thus, $\MEDS(H-(B \setminus \{b\})) \leq \MEDS(H-B)$, which implies that $B$ is not beneficial, which is a contradiction.
\end{proof}

\begin{lemma} \label{lemma::CP_R}
Every connected graph $H$ that contains at least one vertex that is not in $N[W(H)] \cup U(H)$ contains a control pair.
\end{lemma}
\begin{proof}
  We can assume that the graph $H$ neither contains an extendable vertex that is not free nor a strongly beneficial set which contains at least one uncovered vertex; otherwise we can apply Lemma \ref{lemma::CP_Q} resp.\ Lemma \ref{lemma::CP_U} to find a control pair.
  
  First, we prove that there exists a vertex $v \notin N[W(H)] \cup U(H)$ that is not contained in every minimum edge dominating set of $H$. 
  Assume for contradiction that every vertex in $R:= V(H) \setminus (N[W(H)] \cup U(H) )$ is contained in every minimum edge dominating set of $H$. Let $v \in R$ and let $F$ be a minimum edge dominating set in $H$ with $|N_H(v) \cap V(F)|$ maximal. Let $x \in V(H)$ such that $\{v,x\} \in F$; note that $x \notin W(H)$, because $v \notin N[W(H)]$. Consider the vertex set $X:=N_H(v) \setminus V(F)$. It holds that $X$ neither contains a vertex of $R$ (because every vertex in $R$ is contained in every minimum edge dominating set) nor a vertex of $W(H)$ (because $v \notin N[W(H)]$). In addition, $X$ contains no vertex of $N(W(H))$: If $X$ would contain a vertex in $N(W(H))$ then we know that there exists an edge dominating set $F'$ in $H$ such that $(V(F) \cup N(W(H)) \setminus W(H) \subseteq V(F')$ (Proposition~\ref{proposition::properties}~(\ref{proposition::N(W)IN})). This implies that $|N_H(v) \cap V(F)| < |N_H(v) \cap V(F')|$ because $N_H(v) \cap V(F')$ contains all vertices that are contained in $N_H(v) \cap V(F)$ and the vertex in $N(W(H))$ that is contained in $X$; note that $W(H) \cap N_H(v) = \emptyset$. Furthermore, $X$ is not the empty set because this would imply that $x$ is extendable but not free: The set $F - \{\{x,v\}\}$ would be a minimum edge dominating set of $H-x$, because $N_H(v) \subseteq V(F)$. Thus, the set $X$ is contained in $U(H)$ and not empty. Let $Y = X \cup \{x\}$. 
  The set $F - \{\{v,x\}\}$ is an edge dominating set in $H-Y$ because $F- \{v,x\}$ dominates all edges in $H-Y$ that are not adjacent to $v$ (the vertex $x$ is contained in $Y$), and vertex $v$ is only adjacent to the vertices in $Y$ and vertices that are contained in $V(F)$. Thus, it holds that $\MEDS(H-Y) < \MEDS(H)$.
  This implies that there exists a strongly beneficial set $B \subseteq Y$ (Proposition~\ref{proposition::properties}~(\ref{proposition::strongly})). Since $x \notin W(H)$ and every extendable vertex is also free, no vertex in $Y$ is extendable and, therefore, $|B| \geq 2$ (Proposition~\ref{proposition::properties}~(\ref{proposition::Q-W}) and (\ref{proposition::disjointQ})). This directly implies that $B \cap U(H) \neq\emptyset$ because $Y$ contains only one element, namely $x$, that is not in $U(H)$. This, however, is a contradiction since we assumed that $H$ contains no strongly beneficial set that contains at least one uncovered vertex. Thus, there exists a vertex in the set $R$ that is not contained in every minimum edge dominating set.
  
  Let $v \in R$ be such a vertex that is not contained in every minimum edge dominating set and let $F$ be a minimum edge dominating set in $H$ that contains $v$ (such a minimum edge dominating set exists; otherwise $v$ would be uncovered). Let $x \in V(H)$ such that $\{v,x\} \in F$. The vertex $x$ is not contained in $W(H)$ because vertex $x$ is adjacent to vertex $v$, and because vertex $v$ is in $R$. Hence, neither vertex $v$ nor vertex $x$ are extendable because every extendable vertex is also free.
  The fact that $v,x \notin Q(H)$ together with Proposition~\ref{proposition::properties}~(\ref{proposition::BFSedge}) implies that the set $B=\{v,x\}$ is strongly beneficial.
  Let $C:=N_H(v) - x$. It holds that $C \cap Q(H)$ is empty because $W(H)=Q(H)$ and $v \notin N[W(H)]$. Furthermore, $C \cap B$ is empty (by choice of $C$).
  Next, we will show that $(C, B)$ is a control pair. We already showed that $B$ is strongly beneficial. Recall that we chose a vertex $v \in R$ that is not contained in every minimum edge dominating set of $H$. Thus, there exists a minimum edge dominating set $F_v$ in $H$ that does not contain $v$. This edge dominating must contain all vertices in $N_H(v)$. Since $C$ is a subset of $N_H(v)$, there exists a minimum edge dominating $\hat{F}$ in $H$ with $C \subseteq V(\hat{F})$. Furthermore, there exists no minimum edge dominating set $F_B$ in $H-B$ with $C \subseteq V(F_B)$: Otherwise, $x$ would be extendable (and not free) because $F_B$ would also be an edge dominating set in $H-x$. To prove that no vertex in $C$ is extendable in $H-B$ assume for contradiction that there exists a vertex $c \in C$ that is extendable in $H-B$. Thus, there exists a minimum edge dominating set $F_c$ in $H-B-c$ of size $\MEDS(H-B)-1$. We can extend $F_c$ to an edge dominating set in $H-c$ of size $\MEDS(H-B)=\MEDS(H)-1$ by adding the edge $\{v,x\}$ to $F_c$. But now, $c$ would also be extendable in $H$, which contradicts the assumption.
\end{proof}

So far, we showed that every connected graph $H$ that contains an extendable vertex that is not free, a vertex in $V(H) \setminus (N[W(H)] \cup U(H) )$, or a strongly beneficial set that contains at least one vertex of $U(H)$ has a control pair. Thus, for all remaining connected graphs $H$ it holds that every extendable vertex is also free ($Q(H)=W(H)$) and that $V(H) \setminus (N[W(H)] \cup U(H) )$ is empty; hence $V(H) = N[W(H)] \cup U(H)$. Moreover, no strongly beneficial set contains any vertex of $U(H)$ and, hence, strongly beneficial sets must be subsets of $N(W(H))$; recall that (strongly) beneficial sets are subsets of $V(H)\setminus W(H)$.

\begin{lemma}\label{lemma::CP_other}
Every connected graph $H$ that has a strongly beneficial set $B\subseteq N(W(H))$ such that no minimum edge dominating set $F_B$ of $H-B$ covers all vertices of $N(W(H))\setminus B$ contains a control pair.
\end{lemma}
\begin{proof}
  We can assume that graph $H$ neither contains an extendable vertex that is not free, a strongly beneficial set which contains at least one uncovered vertex, nor a vertex that is not in $N[W(H)] \cup U(H)$; otherwise we can apply Lemma \ref{lemma::CP_Q} resp.\ Lemma \ref{lemma::CP_U} resp.\ Lemma \ref{lemma::CP_R}. Thus $H=N[W(H)] \cup U(H)$ and every strongly beneficial set is contained in $N(W(H))$.
  
  Note that every strongly beneficial set $B \subseteq N(W(H))$ with the property that there exists no minimum edge dominating set in $H-B$ that contains all vertices of $N(W(H)) \setminus B$ together with the set $C=N(W(H)) \setminus B$ fulfills all except one property of a control pair: The set $B$ is strongly beneficial, $C \subseteq V(H) \setminus (Q(H) \cup B)$, and for every minimum edge dominating set $F_B$ in $H-B$ is holds that $C \nsubseteq V(F_B)$ (choice of $B$). Furthermore, there exists a minimum edge dominating set $F$ in $H$ such that $C \subseteq V(F)$, because $C \subseteq N(W(H))$ (Proposition~\ref{proposition::properties}~(\ref{proposition::N(W)IN})). But, we do not know whether $C \cap Q(H-B) = \emptyset$. We will show that if there exists a strongly beneficial set $B$ in $H$ such that no minimum edge dominating set in $H-B$ contains $N(W(H)) \setminus B$, then there exists also a strongly beneficial set $B'$ in $H$ such that no minimum edge dominating set in $H-B'$ contains $C'=N(W(H)) \setminus B'$, and such that $Q(H-B') \cap C' = \emptyset$.
  
  Let $B \subseteq N(W(H))$ be a strongly beneficial set in $H$ such that no minimum edge dominating set of $H-B$ contains all vertices in $C:=N(W(H)) \setminus B$, and with $|Q(H-B) \cap C|$ minimal under these strongly beneficial sets. (Such a set $B$ exists by assumption.) If $Q(H-B) \cap C = \emptyset$, then $B$ fulfills the desired property and $(B,C)$ is a control pair. Thus, assume that $|Q(H-B) \cap C|>0$ and let $c \in Q(H-B) \cap C$.
  We show that $B'=B \cup \{c\}$ is a strongly beneficial set in $H$ that fulfills the same properties as $B$, and with $|Q(H-B) \cap C| > |Q(H-B') \cap C'|$, where $C'=N(W(H)) \setminus B'$, which is a contradiction to the choice of $B$.
  \begin{claim}
   The set $B'=B \cup \{c\}$ is a strongly beneficial set in $H$.
  \end{claim}
  \begin{claimproof}
   It holds that $\MEDS(H-B')+1 = \MEDS(H-B)$, because $c$ is extendable in $H-B$ (and $H-B'=(H-B)-c$). This implies that $\cost(B')=\cost(B)$, because $\cost(B')=\MEDS(H-B') + |B'|-\MEDS(H) = \MEDS(H-B) - 1 + |B|+1 - \MEDS(H) = \cost(B)$. 
   First, assume for contradiction that $B'$ is not beneficial. Hence, there exists a proper subset $\widetilde{B} \subsetneq B'$ of $B'$ such that $\MEDS(H-B') \geq \MEDS(H-\widetilde{B})$ (definition of beneficial). Note that $B \setminus \widetilde{B} \neq \emptyset$ is not the empty set, otherwise $\widetilde{B}=B$ and $\MEDS(H-\widetilde{B})=\MEDS(H-B) > \MEDS(H-B')$. Since $\widetilde{B} \setminus \{c\} \subseteq \widetilde{B}$ (they can be equal if $c \notin \widetilde{B}$) and $|\widetilde{B} \setminus (\widetilde{B} \setminus \{c\})|\leq 1$ it follows from Proposition~\ref{proposition::properties}~(\ref{proposition::monoton}) that $\MEDS(H - \widetilde{B}) \geq \MEDS(H- (\widetilde{B} \setminus \{c\})) - 1$. Combining the above inequations we obtain:
   \[
    \MEDS(H-B) = \MEDS(H-B')+1 \geq \MEDS(H-\widetilde{B}) +1 \geq \MEDS(H-(\widetilde{B} \setminus \{c\}))
   \]
   But, this implies that $B$ is not beneficial because $\widetilde{B} \setminus \{c\} \subsetneq B$ (we showed that $B \setminus \widetilde{B} \neq \emptyset$), and $\MEDS(H-B) \geq \MEDS(H-(\widetilde{B} \setminus \{c\}))$. This contradicts the choice of $B$, hence $B'$ is beneficial.
   
   Next, we show that $B'$ is strongly beneficial. Again, assume for contradiction that $B'$ is not strongly beneficial (but beneficial). Thus, there exists a non-trivial partition $B_1,B_2,\ldots, B_h$ of $B'$ such that $\cost(B') \geq \sum_{i=1}^h \cost(B_i)$ (Proposition~\ref{proposition::properties}~(\ref{proposition::definition})). We assume w.l.o.g.\ that $c \in B_1$ and it follows from Proposition~\ref{proposition::properties}~(\ref{proposition::monoton}) that $\cost(B_1 \setminus \{c\}) \leq \cost(B_1)$. It holds that $B_1 \setminus \{c\}, B_2, \ldots, B_h$ is a partition of $B$. Additionally, $B_1 \setminus \{c\}, B_2, \ldots, B_h$ is a non-trivial partition of $B$ if $B_1 \setminus \{c\} \neq \emptyset$, because $h \geq 2$ and no set is the empty set. Furthermore, $B_2,B_3, \ldots, B_h$ is a non-trivial partition of $B$ if $B_1 \setminus \{c\} = \emptyset$: Otherwise it holds that $B_1=\{c\}$ and $B_2=B$ which implies that $\cost(B_1)+\cost(B_2)=1+\cost(B) = 1+ \cost(B') > \cost(B')$ which is a contradiction. (Note that the first equality holds because $c \in N(W(H))$ is not extendable in $H$). Now, it holds that
   \[
    \cost(B_1 \setminus \{c\}) + \sum_{i=2}^h \cost(B_i) \leq \sum_{i=1}^h \cost(B_i) \leq \cost(B') = \cost(B).
   \]
  But this implies that $B$ is not strongly beneficial, because we showed that there exists a non-trivial partition $(B_1 \setminus \{c\},) B_2, \ldots B_h$ of $B$ with $\sum_{i=1}^h \cost(B_i) \leq \cost(B)$. This is a contradiction, thus $B'$ is strongly beneficial.
  \end{claimproof}
  Now, we will show (by contradiction) that there exists no minimum edge dominating set $F_{B'}$ in $H-B'$ such that $C':=N(W(H)) \setminus B' \subseteq V(F_{B'})$. Assume that there exists a minimum edge dominating set $F_{B'}$ in $H-B'$ such that $C' \subseteq V(F_{B'})$. The set $F_B=F_{B'} \cup \{\{c,v\}\}$ with $v \in N_{H-B}(c)$ is a minimum edge dominating set of size $|F|+1=\MEDS(H-B')+1=\MEDS(H-B)$ in $H-B$ because $F_B$ dominates all edges in $H-B'$ and all edges that are incident with $c$ (since $c \in V(F_B)$). (Note that $N_{H-B}(c)$ is not empty, because $c \in N(W(H))$ and $W(H) \cap B = \emptyset$ (definition of beneficial sets).) But, the minimum edge dominating set $F_B$ in $H-B$ contains all vertices in $C:=N(W(H)) \setminus B$ as an endpoint: it holds $C = C' \cup \{c\}$ and $C' \subseteq V(F_{B'}) \subseteq V(F_B)$ as well as $c \in V(F_B)$.
  This contradicts the choice of $B$, hence there exists no minimum edge dominating set $F_{B'}$ in $H-B'$ such that $C' \subseteq V(F_{B'})$.
  
  Finally, we show that $|Q(H-B) \cap C| > |Q(H-B') \cap C'|$ which contradicts the choice of $B$.
  \begin{claim}
   It holds that $Q(H-B') \cap C' \subsetneq Q(H-B) \cap C$
  \end{claim}
  \begin{claimproof}
   $(\subseteq :)$ Let $v \in Q(H-B') \cap C')$ be a vertex that is extendable in $H-B'$. (If $Q(H-B')$ is empty then $Q(H-B') \cap C' \subseteq Q(H-B) \cap C$.) Thus, there exists a minimum edge dominating set $F'_v$ in $H-B'-v$ of size $\MEDS(H-B')-1$. Consider $F_v:=F'_v \cup \{\{c,u\}\}$ with $u \in N_{H-B-v}(c)$. (Again, $N_{H-B-v}(c) \neq \emptyset$ because $c$ is adjacent to a vertex in $W(H)$ and $W(H) \cap (B' \cup \{v\}) = \emptyset$.)
   The set $F_v$ is an edge dominating set (of size $\MEDS(H-B')=\MEDS(H-B)-1$) in $H-B-v$ because $F_v$ dominates all edges in $H-B'-v$ and all edges that are incident with $c$.  Furthermore, $F_v$ is a minimum edge dominating set in $H-B-v$ because $|F_v| \geq \MEDS(H-B-v) \geq \MEDS(H-B) -1=|F_v|$ (Proposition~\ref{proposition::properties}~(\ref{proposition::vertex})). It follows that $\MEDS(H-B-v)=\MEDS(H-B)-1$; thus $v \in Q(H-B)$ (definition of extendable vertices). Hence, $Q(H-B') \cap C' \subseteq Q(H-B) \cap C$.
   
   $(\neq:)$ The vertex $c$ is extendable in $H-B$ (choice of $c$); hence $c \in Q(H-B) \cap C$. But $c$ is not  extendable in $H-B'$, because $c \in B'$ and thus not contained in $H-B'$. This implies $Q(H-B) \cap C \neq Q(H-B') \cap C'$ and concludes the proof.
  \end{claimproof}
  Now, $B'$ is a strongly beneficial set in $H$ such that no minimum edge dominating set $F_{B'}$ in $H-B'$ contains all vertices in $C'$. Furthermore, $Q(H-B') \cap C' \subsetneq Q(H-B) \cap C$ which contradicts the choice of $B$ because $|Q(H-B') \cap C'| < |Q(H-B) \cap C|$.
  
  Overall, we showed that there exists a strongly beneficial set $B \subseteq N(W(H))$ such that there exists no minimum edge dominating set in $H-B$ that covers all vertices of $C:=N(W(H)) \setminus B \subseteq V \setminus (Q(H) \cup B)$ and such that $Q(H-B) \cap C = \emptyset$. It holds that $(B,C)$ is a control pair (see argumentation above).
\end{proof}

\subsection[Generalizing the upper bound obtained for P5-component graphs]{Generalizing the upper bound obtained for $\boldsymbol{P_5}$-component graphs}

We showed for many finite sets \cH that \EDS parameterized by the size of a modulator to an $\mathcal{H}$-component graph has a no polynomial kernel, unless \containment.
Now, we will show that \EDS parameterized by the size of a modulator to \cH-component graphs admits a polynomial kernel for all remaining choice of finite sets $\cH$. Recall that all remaining sets \cH only contain graphs $H$ with $V(H)=N[W(H)] \cup U(H)$. Furthermore, it holds for every strongly beneficial set $B$ in $H$ that there exists a minimum edge dominating set $F_B$ in $H-B$ with $N(W(H)) \setminus B \subseteq V(F_B)$.

We first show that if no graph in \cH has any beneficial set then there is a kernel with $\Oh(|X|^2)$ vertices and $\Oh(|X|^3)$ edges. (Recall, for every kernel lower bound we showed that there exists a strongly beneficial set, thus all graphs that have no beneficial set can only have vertices that are uncovered, free, or neighbours of free vertices.) This will later be extended to a more involved kernelization that also handles $H$-components where $H$ does have beneficial sets but they always have a minimum edge dominating set $F_B$ in $H-B$ as above. 

\begin{lemma} \label{lemma::kernel1}
Let $\cH$ be a finite set of connected graphs that contain no beneficial sets and such that each graph $H\in\cH$ has $V(H)=N[W(H)]\cup U(H)$, i.e., each graph $H$ only has vertices that are uncovered, free, or neighbors of a free vertex. Then \EDS parameterized by the size of a modulator $X$ to an $\mathcal{H}$-component graph admits a kernel with $\Oh(|X|^2)$ vertices, $\Oh(|X|^3)$ edges, and size $\Oh(|X|^3\log|X|)$. 
\end{lemma}
\begin{proof}
  Let $(G,k,X)$ be an instance of \EDS parameterized by the size of a modulator to an $\mathcal{H}$-component graph. 
We can assume that $k - \MEDS(G-X) < |X|$. Otherwise, we can return a trivial solution consisting of a minimum edge dominating set in $G-X$ and one edge in $\delta_G(x)$ for each $x \in X$. Let $\mathcal{C}$ be the set of connected components in $G-X$, let $W$ be the set of all free vertices in $G-X$, and let $U$ be the set of all uncovered vertices in $G-X$, hence $W = \bigcup_{C \in \mathcal{C}} W(C)$ and $U = \bigcup_{C \in \mathcal{C}} U(C)$.
  Let $X_W \subseteq X$ be the set of vertices in $X$ that are adjacent to a vertex in $W$, hence $X_W= \{ x \in X \mid \exists w \in W \colon \{x,w\} \in E(G) \}$, and let $\mathcal{C}_W$ be the set of connected components $C$ in $\mathcal{C}$ where $W(C)$ is adjacent to a vertex in $X_W$, hence with $E(W(C),X_W) \neq \emptyset$. 
  To find vertices in $X$ that can be covered by every edge dominating set of size at most $k$ without spending extra budged we construct a bipartite graph $G_W$. One part of $G_W$ is the set $X_W$ and the other part consists of one vertex $s_C$ for each connected component $C$ in $\mathcal{C}_W$. We add an edge between a vertex $x \in X_W$ and a vertex $s_C$ with $C \in \mathcal{C}_W$ if and only if vertex $x$ is adjacent to a vertex in $W(C)$ in $G$. Now, we apply Theorem \ref{theorem::matching} to obtain either a maximum matching in $G_W$ that saturates $X_W$, or to find a set $Y \subseteq X_W$ such that $|N_{G_W}(Y)|<|Y|$ and such that there exists a maximum matching in $G_W - N_{G_W}[Y]$ that saturates $X_W \setminus Y$.
  If there exists a maximum matching in $G_W$ that saturates $X_W$ then let $X_W^h=X_W$, let $X_W^l=\emptyset$, and let $\mathcal{C}_W^l=\emptyset$. Otherwise, if there exists a set $Y$ with the above properties then let $X_W^l=Y$, let $X_W^h=X_W \setminus Y$, and let $\mathcal{C}_W^l$ be the connected components $C$ in $\mathcal{C}_W$ where $W(C)$ is adjacent to a vertex in $X_W^l$. It holds that $|\mathcal{C}_W^l|=|N_{G_W}(Y)|<|Y|=|X_W^l|$ because every connected component in $\mathcal{C}_W^l$ corresponds to a vertex in $N_{G_W}(Y)$; hence $|\mathcal{C}_W^l| \leq |X_W^l|$.
  \begin{redrule} \label{rule::W}
    Delete the set $X_W^h$ from $G$, i.e., let $G'=G-X_W^h$, $X'=X \setminus X_W^h$ and $k'=k$.
  \end{redrule}
  \begin{claim} \label{claim::X_W}
    Reduction Rule \ref{rule::W} is safe.
  \end{claim}
  \begin{claimproof}
    Let $F$ be an edge dominating set of size at most $k$ in $G$. We construct an edge dominating set $F'$ of size at most $k'=k$ in $G'$ as follows: First, we delete every edge $e \in F$, if both endpoints of $e$ are contained in $X_W^h$. Next, for every edge $e=\{x,y\} \in F$ that has exactly one endpoint in $X_W^h$ (w.l.o.g.\ $x \in X_W^h$) we either replace $e$ by one edge in $\delta_{G'}(y)$, if $\delta_{G'}(y) \neq \emptyset$ or delete $e$, if $\delta_{G'}(y) = \emptyset$. It holds that $F'$ has size at most $k=k'$ because we only delete edges or replace edges. Furthermore, $V(G') \setminus V(F')$ is an independent set because $V(F')$ contains every vertex in $V(F) \setminus X_W^h$ that is not isolated in $G'$, because $V(G) \setminus V(F)$ is an independent set, and because $V(G')=V(G) \setminus X_W^h$. Thus, $F'$ is an edge dominating set of size at most $k'$ in $G'$.
    
    For the other direction, let $F'$ be an edge dominating set of size at most $k'$ in $G'$ that is also a matching. 
    Let $M$ be a maximum matching in $G_W - N_{G_W}[X_W^l]$ that saturates $X_W^h$, and for each $x \in X_W^h$ let $C_x$ be the connected component in $\mathcal{C}_W$ with $\{s_C,x\} \in M$. Consider the connected component $C_x$ for a vertex $x \in X_W^h$. 
    Since $M$ is a maximum matching in $G_W - N_{G_W}[X_W^l]$ it holds that the connected component $C_x$ does not correspond to a vertex in $N_{G_W}(X_W^l)$. This implies that no free vertex in $C_x$ is adjacent to a vertex in $X_W^l$ (construction of $G_W$).
    Let $F'_x$ be the set of edges in $F'$ that have at least one endpoint in $C_x$; hence $F'_x=\{f \in F' \mid f \cap C_x \neq \emptyset\}$. We partition $F'_x$ in three sets: Let $F'_{x,1}$ be the set of edges in $F'_x$ that have one endpoint in $X$ and the other endpoint in $C_x$, let $F'_{x,2}$ be the set of edges in $F'_x$ that have one endpoint in $U(C_x)$ and the other endpoint in $N(W(C_x))$, and let $F'_{x,3}$ be the set of remaining edges in $F'_x$. Recall that uncovered vertices are only adjacent to vertices in $N(W(C_x)) \cup X$ (Proposition~\ref{proposition::properties}~(\ref{proposition::neighborhood})). Thus, all edges in $F'_x$ that have one endpoint in $U(C_x)$ (the set of uncovered vertices in $C_x$) are contained in $F'_{x,1} \cup F'_{x,2}$. Let $B_x$ be the set of vertices in $C_x$ that are incident with an edge in $F'_{x,1} \cup F'_{x,2}$, hence $B_x=V(C_x) \cap V(F'_{x,1} \cup F'_{x,2})$. 
    Recall that every vertex in $W(C_x)$ is only adjacent to vertices in $N(W(C_x))$ because no free vertex of $C_x$ is adjacent to a vertex in $X_W^l$ (and we delete set $X_W^h$ to obtain $G'$). Thus, $V(F'_{x,1} \cup F'_{x,2})$ contains no free vertex of $C_x$ because no free vertex in of $C_x$ is adjacent to a vertex in $X$ or an uncovered vertex of $C_x$; thus $B_x \subseteq N(W(C_x)) \cup U(C_x)$. 
    Since $C_x$ has no beneficial set, it holds that $\MEDS(C_x-B_x) = \MEDS(C_x)$ (Proposition~\ref{proposition::properties}~(\ref{proposition::BFS})); thus, $F'_x$ has size at least $\MEDS(C_x)$ plus the size of $|F'_{x,1} \cup F'_{x,2}|$. Furthermore, $F'_{x,3}$ is an edge dominating set in $C_x - B_x$ (by choice of $B_x$); thus $|F'_{x,3}| \geq \MEDS(C_x-B_x) = \MEDS(C_x)$. It follows that $|F'_x|=|F'_{x,1} \cup F'_{x,2}|+|F'_{x,3}| \geq |F'_{x,1} \cup F'_{x,2}| + \MEDS(C_x)$.
    
    Let $w_x \in N_G(x) \cap W(C_x)$ be a free vertex in $C_x$ that is adjacent to $x \in X_W^h$. We replace $F'_{x,3}$ by the set $F_{x,3}$ that consists of a minimum edge dominating set in $C_x - w_x$ that covers all vertices in $N(W(C_x))$ (which exists by Proposition~\ref{proposition::properties}~(\ref{proposition::N(W)IN})) and the edge $\{w_x,x\}$. The edge set $F_{x,3}$ has size $\MEDS(C_x)$ because $w_x \in W(C_x)$ (definition of free); hence $|F_{x,3}| \leq |F'_{x,3}|$. We do this for all $x \in X_W^h$ to obtain $F$.     
    It holds that $F$ has the same size as $F'$ because we only replace $F'_{x,3}$ by $F_{x,3}$ for each $x \in X_W^h$, and because $|F_{x,3}| \leq |F'_{x,3}|$. 
    
    It remains to prove that $F$ is indeed an edge dominating set in $G$. The set $V(F)$ contains all vertices in $V(F')$, except some free vertices in the connected components where we change the edge dominating set. But, all neighbours of these free vertices are contained in $V(F)$, because these free vertices are only adjacent to vertices in $X_W^h$, which are contained in $V(F)$, and to vertices in the connected component of $G-X$ they belong to. Thus, $F$ is an edge dominating set in $G$.
  \end{claimproof}
  
  Now, let $X_U \subseteq X$ be the set of vertices in $X$ that are adjacent to a vertex in $U$, hence $X_U= \{ x \in X \mid \exists u \in U \colon \{x,u\} \in E(G) \}$. We partition $X_U$ in two sets as follows: let $X_U^h \subseteq X_U$ be the set of vertices in $X_U$ that are adjacent to the set of uncovered vertices in at least $|X|+1$ connected components (hence $X_U^h = \{ x \in X_U \mid \exists C_1,C_2,\ldots,C_{|X|+1} \in \mathcal{C} \forall j \in [|X|+1] \colon x \in N_G(U(C_j)) \}$) and let $X_U^l = X_U \setminus X_U^h$ be the set of vertices in $X_U$ that are adjacent to the set of uncovered vertices in less than $|X|$ connected components in $\mathcal{C}$. By $\mathcal{C}_U^l$ we denote the connected components $C$ in $\mathcal{C}$ with $E(U(C),X_U^l) \neq \emptyset$. It holds that $|\mathcal{C}_U^l| \leq |X_U^l| \cdot |X|$.
  \begin{redrule} \label{rule::U}
    For all $x \in X_U^h$ add a vertex $x'$ and the edge $\{x,x'\}$ to $G$.
  \end{redrule}
  Let $\mathcal{C}_S$ be the set of new connected components in $G'$; these are the connected components consisting of a single vertex $x'$ that we add during Reduction Rule \ref{rule::U}.
  It is easy to verify that Reduction Rule \ref{rule::U} is safe, i.e.\ that there exists a solution for $(G,k,X)$ if and only if there exists a solution for $(G',k',X')$. This follows from the fact that every edge dominating set $F$ in $G$ of size at most $k$ must contain the vertices in $X_U^h$ as endpoints: If there would be a vertex $x \in X_U^h$ that is not contained in $V(F)$, then there exist at least $|X|+1$ connected components that contain a vertex in $U$ that is adjacent to $x$. Now, these at least $|X|+1$ vertices must be contained in $V(F)$. But, every connected component $C$ with this property is adjacent to $\MEDS(C) +1$ edges in $F$ because no minimum edge dominating set in $C$ covers an uncovered vertex. Thus, $|F| \geq \MEDS(G-X) + |X|+1 > k$ because two connected components in $G-X$ are not adjacent. This is a contradiction and shows that every vertex in $X_U^h$ must be contained in an edge dominating set of size at most $k$ in $G$. By adding the vertex $x'$ and the edge $\{x,x'\}$, with $x \in X_U^h$, to $G$, we encode that vertex $x$ must be in every solution.

  Let $\mathcal{C}_D \subseteq \mathcal{C}$ be the set of connected components in $G-X$ that are not contained in $\mathcal{C}_U^l \cup \mathcal{C}_W^l \cup \mathcal{C}_S$. Note, $\mathcal{C}_D$ contains all connected components where neither a free nor an uncovered vertex is adjacent to a vertex in $X$.
  \begin{redrule} \label{rule::delete}
    Delete all connected components in $\mathcal{C}_D$ and decrease $k$ by the size of a minimum edge dominating set in $\mathcal{C}_D$.
  \end{redrule}
  \begin{claim}
   Reduction Rule \ref{rule::delete} is safe.
  \end{claim}
  \begin{claimproof}
   Let $F$ be an edge dominating set of size at most $k$ in $G$. If no edge in $F$ has one endpoint in a connected component of $\mathcal{C}_D$ and the other endpoint in $X$, then $F' = F \setminus E(\mathcal{C}_D)$ is an edge dominating set of size at most $|F| - \MEDS(\mathcal{C}_D) \leq k'$ in $G'$ and we are done. Thus, assume that there exists at least one edge that has one endpoint in $X$ and one endpoint in a connected component of $\mathcal{C}_D$; hence $F \cap E(V(\mathcal{C}_D),X) \neq \emptyset$. We denote the set of edges in $F$ that are incident with a connected component in $\mathcal{C}_D$ by $F_D$, and let $F_{D,X} = F \cap E(V(\mathcal{C}_D),X)$.
   
   It holds that every vertex $v$ in $V(F_{D,X})$ is not a free vertex: Otherwise, there exists a vertex $x \in X_W$ such that $\{x,v\} \in F_{D,X}$. But, we delete every vertex in $X_W^h$ during Reduction Rule \ref{rule::W} and every connected component that contains a free vertex that is adjacent to a vertex in $X_W^l$ is a connected component in $\mathcal{C}_W^l$ and therefore not in $\mathcal{C}_D$.
   
   Let $\widetilde{X} \subseteq X$ be the set of vertices in $X$ that are contained in $V(F_D)$, hence $\widetilde{X}=X \cap V(F_D)$. Since no vertex in $\widetilde{X}$ is adjacent to a free vertex in $\mathcal{C}_D$, and since the connected components do not have beneficial sets it holds that the size of a minimum edge dominating set in $\mathcal{C}_D$ plus the size of the set $\widetilde{X}$ is smaller or equal to the number of edges in $F_D$.
   Now, to obtain an edge dominating set $F'$ of size at most $k'$ in $G'$, we delete the edge set $F_D$ from $F$ and we add for all vertices $x \in \widetilde{X}$ exactly one edge of the set $\delta_{G'}(x)$ to the edge dominating set (or none if the edge set $\delta_{G'}(x)$ is empty). By construction it follows that $F'$ has size at most $k'$ because we delete $|F_D| \geq \MEDS(\mathcal{C}_D) + |\widetilde{X}|$ edges from $F$ and add at most $|\widetilde{X}|$ edges to $F$. Furthermore, $F'$ is an edge dominating set in $G'$, because $V(F')$ contains all vertices of $V(F)$ that are contained in $V(G')$ and not isolated in $G'$.
   
   For the other direction, let $F'$ be an edge dominating set of size at most $k'$ in $G'$. To obtain an edge dominating set $F$ of size at most $k$ in $G$ we add for every connected component $C \in \mathcal{C}_D$ a minimum edge dominating set $F_C$ in $C$ with $N(W(C)) \subseteq V(F_C)$ to $F'$ (the existence of such a minimum edge dominating set follows from Proposition~\ref{proposition::properties}~(\ref{proposition::N(W)IN})). To show that $F$ is indeed an edge dominating set in $G$ we only have to show that every edge $e \in E(\mathcal{C}_D,X)$ is dominated by $F$. All other edges are dominated by $F$, because they are already dominated by $F' \subseteq F$ or by $F_C$ with $C \in \mathcal{C}_D$. 
   Assume for contradiction that there exists a connected component $C \in \mathcal{C}_D$ and an edge $e=\{v,x\}$ in $G$ with $v \in V(C)$ and $x \in X$ such that neither vertex $v$ nor vertex $x$ is an endpoint of an edge in $F$. Since all vertices in $N(W(C))$ are endpoints of an edge in $F$ (choice of $F_C$), and since $C$ contains only free vertices, neighbours of free vertices and uncovered vertices, it holds that $v$ is either a free or uncovered vertex in $C$; hence $x$ is contained in $X_W \cup X_U$. It holds that all vertices in $X_U^h$ are contained in $V(F') \subseteq V(F)$ because every vertex $x$ in $X_U^h$ is adjacent to a connected component in $\mathcal{C}_S$ that consists of a single vertex $x'$ which is only adjacent to vertex $x$ in $X_U^h$. Furthermore, during Reduction Rule \ref{rule::W} we delete the vertex set $X_W^h$. Thus, $x$ is contained in $X_U^l \cup X_W^l$. Now, if $v$ is a free vertex in $C$, then $x$ must be a vertex in $X_W^l$ which implies that $C$ is a connected component in $\mathcal{C}_W^l$ and not in $\mathcal{C}_D$, which is a contradiction. Similar, if $v$ is an uncovered vertex in $C$, then $x$ must be a vertex in $X_U^l$ which implies that $C$ is a connected component in $\mathcal{C}_U^l$ and not in $\mathcal{C}_D$, which is a contradiction. Hence, $F$ is an edge dominating set in $G$.
  \end{claimproof}
    We already showed that the reduction is safe. Next, we show that the reduced instance $(G',k',X')$ has at most $\Oh(|X|^2)$ vertices. The set of connected components in $G'-X'$ is $\mathcal{C}':=\mathcal{C}_U^l \cup \mathcal{C}_W^l \cup \mathcal{C}_S$, because we only add connected components to $G$ during Reduction Rule \ref{rule::U} (namely the components in $\mathcal{C}_S$) and we only delete connected components during Reduction Rule \ref{rule::delete}. We delete all connected components that are not contained in $\mathcal{C}_U^l \cup \mathcal{C}_W^l \cup \mathcal{C}_S$.
    It follows that $G'-X'$ has at most $2|X|^2$ connected components, because $|\mathcal{C}'| \leq |\mathcal{C}_U^l| + |\mathcal{C}_W^l| + |\mathcal{C}_S| \leq |X_U^l| \cdot |X| + |X_W^l|+ |X_U^h| \leq 2 |X|^2$. Since every connected component has constant size, and since $V(G')=V(\mathcal{C}') \cup X$ it holds that $G'$ has at most $\Oh(|X|^2)$ vertices. Next, we have to bound the number of edges. Every connected component has only constant size, thus it has only a constant number of edges (because our graph is simple); hence $|E[\mathcal{C}']| \in \Oh(|X|^2)$. The number of edges between vertices in $X$ is at most $|X|^2$. All remaining edges are between $X$ and $\mathcal{C}'$ and there are at most $|X| \cdot |V(\mathcal{C}')| \in  \Oh(|X|^3)$ edges between $X$ and $\mathcal{C}'$. This sums up to at most $\Oh(|X|^3)$ edges.
    
    It is easy to see that we can perform the reduction in polynomial time: We apply every Reduction Rule exactly once and we also compute every set exactly once. Furthermore, we can compute the sets $W$ and $U$ in polynomial time because we can compute a minimum edge dominating set in a connected component of constant size in constant time. Moreover, we can compute all remaining sets in polynomial time (by applying Theorem \ref{theorem::matching} or by simple counting). We can also apply each Reduction Rule in polynomial time because we only delete resp.\ add vertex or edge sets of size polynomial in $|G|$ which we can compute in polynomial time.
\end{proof}

\begin{remark}
 If $\cH$ is a finite set of connected graphs that contain neither beneficial sets nor uncovered vertices (and such that each graph $H \in \cH$ has $V(H)=N[W(H)]$ then \EDS parameterized by the size of a modulator $X$ to an \cH-component graph admits a kernel with $\Oh(|X|)$ vertices, $\Oh(|X|^2)$ edges, and size $\Oh(|X|^2 \log|X|)$.
 This holds because the only set that has size $\Oh(|X|^2)$ is the set of connected components in $\mathcal{C}_U^l$. But, if we have no uncovered vertex then $\mathcal{C}_U^l = \emptyset$. Hence, the reduced instance has only $\Oh(|X|)$ many connected components, and therefore, only $\Oh(|X|)$ many vertices.
\end{remark}

\begin{lemma}\label{lemma::kernel2}
Let $d\in \N$ and let $\cH$ be a finite set of connected graphs such that no graph $H\in\cH$ has a strongly beneficial set of size exceeding $d$, such that $V(H)=N[W(H)]\cup U(H)$ for all $H\in\cH$, and such that each strongly beneficial set $B$ of any graph $H\in\cH$ is contained in $N(W(H))$. Moreover, assume that for each strongly beneficial set $B$ of a graph $H\in\cH$ there exists a minimum edge dominating set in $H-B$ that covers all vertices in $N(W(H))\setminus B$.
Then \EDS parameterized by the size of a modulator $X$ to the class \cH-component graphs admits a kernel with $\Oh(|X|^d)$ vertices, $\Oh(|X|^{d+1})$ edges, and size $\Oh(|X|^{d+1} \log|X|)$.
\end{lemma}

\begin{proof}
  Let $(G,k,X)$ be an instance of \EDS parameterized by the size of a modulator to an \cH-component graph. Again, we can assume that $k-\MEDS(G-X) < |X|$ (see proof of Lemma \ref{lemma::kernel1}). The kernelization is similar to the previous kernelization. We construct a different graph $G_W$ to compute $X_W^h$ because we have to be a little bit more careful with connected components in $\mathcal{C}_U^l$ and because we have beneficial sets. Furthermore, we define the set $\mathcal{C}_D$ differently (for this purpose we compute another auxiliary graph).
  
  In the previous kernelization (Lemma \ref{lemma::kernel1}) the connected components in $G-X$ have no beneficial set. Thus, for every connected component $C$ that has at least one free vertex and whose set of free vertices is adjacent to at least one vertex in $X$, we could assume that an edge dominating set contains at least one of these edges between the free vertices and $X$. We can not assume this anymore because a connected component can also contain a beneficial set $B$. Now, it could be necessary that an edge dominating set of size at most $k$ contains a matching between $B$ and $X$. For example, assume that there exists a connected component $C$ that contains a free vertex that is adjacent to vertex $x_1$ in $X$, and that $C$ has also a beneficial set $B=\{b_1,b_2,b_3\}$ of size three with $\cost(B)=1$ such that $b_i$ is adjacent to a vertex $x_i$ in $X$ (where all $x_i$ are pairwise different). Now, we can either increase the $\cost$ locally by one to cover three vertices in $X$ or use the same local $\cost$ to cover one vertex, namely $x_1$, in $X$. Further, assume that $x_1,x_2,x_3$ are endpoints of edges in every edge dominating set of size at most $k$ and that there exist no way to cover $x_2$ and $x_3$ with an extra budget of one. Hence, the edge dominating set will probably contain the edges between $B$ and $\{x_1,x_2,x_3\}$ as well as a minimum edge dominating set in $C-B$, and no edge between a free vertex of $C$ and $X$.
  
  Besides this, in the previous kernelization we could assume that every edge that is incident with an uncovered vertex increases the cost locally by one. We cannot assume this anymore because a connected component $C$ with a beneficial set $B$ can cover $|B|$ vertices in $N(B) \cap U(C)$ while increasing the cost locally by $\cost(B) < |B|$. Thus, a beneficial set $B$ could be useful to cover some vertices in $U(C)$ and to cover vertices in $X$. But, if all vertices in $U(C)$ are only adjacent to vertices in $X$ that must be in every edge dominating set of size at most $k$ then we will never cover the vertices in $U(C)$ by edges inside a connected component. 
  
  For these reasons, we compute the set $X_W^h$ using a different auxiliary graph which leads to $\Oh(|X|^2)$ connected components in $\mathcal{C}_W^l$ (in the worst case). Note that it would be possible to bound the number of connected components in $\mathcal{C}_W^l$ by $\Oh(|X|)$ by defining an auxiliary graph that handles the connected components that have free vertices or beneficial sets at the same time. This would make the analysis more complicated. But, we are only able to bound the number of connected components that contain a beneficial set by $\Oh(|X|^d)$, with $d \geq 2$. Thus, even if no graph has uncovered vertices we will not be able to reduce to less than $\Oh(|X|^2)$ vertices because we have connected components that contain beneficial sets (of size $d \geq 2$). Recall, in the previous kernelization we can reduce to $\Oh(|X|)$ vertices if we have no uncovered vertices.
  
  As before, let $\mathcal{C}$ be the set of connected components in $G-X$, let $W$ be the set of all free vertices in $G-X$, and let $U$ be the set of all uncovered vertices in $G-X$. 
  Let $X_W \subseteq X$ be the set of vertices in $X$ that are adjacent to a vertex in $W$, hence $X_W= \{ x \in X \mid \exists w \in W \colon \{x,w\} \in E(G) \}$, and let $\mathcal{C}_W$ be the set of connected components $C$ in $\mathcal{C}$ where $W(C)$ is adjacent to a vertex in $X_W$, hence with $E(W(C),X_W) \neq \emptyset$. 
  Again, we compute a bipartite graph $G_W$: One part consists of $|X|$ vertices $x^1,x^2,\ldots,x^{|X|}$ for every vertex $x$ in $X_W$. We denote this set by $R$. The other part consists of one vertex $s_C$ for every connected component $C$ in $\mathcal{C}_W$. We add an edge between a copy $x^i$ of vertex in $x \in X_W$, with $i \in [|X|]$, and a vertex $s_C$ with $C \in \mathcal{C}_W$ if and only if $x$ is adjacent to a vertex in $W(C)$. Now, we apply Theorem \ref{theorem::matching} to obtain either a maximum matching in $G_W$ that saturates $R$, or to find a set $Y \subseteq R$ such that $|N_{G_W}(Y)|<|Y|$ and such that there exists a maximum matching in $G_W - N_{G_W}[Y]$ that saturates $R \setminus Y$. Observe, since every copy of a vertex $x \in X_W$ has the same neighborhood it holds that either all copies of $x$ are contained in $Y$ or none.
  If there exists a maximum matching in $G_W$ that saturates $R$ then let $X_W^h=X_W$, let $X_W^l=\emptyset$, and let $\mathcal{C}_W^l=\emptyset$. Otherwise, if there exists a set $Y$ with the above properties then let $X_W^l = \{ x \in X_W \mid x^1 \in Y\}$ be the vertices in $X$ whose copies are contained in $Y$, let $X_W^h = X_W \setminus X_W^l$, and let $\mathcal{C}_W^l$ be the set of connected components $C$ in $\mathcal{C}_W$ where $W(C)$ is adjacent to a vertex in $X_W^l$. Note that every connected component $C$ in $\mathcal{C}_W^l$ corresponds to a vertex in $N_{G_W}(Y)$. Thus, the set $\mathcal{C}_W^l$ contains at most $|N_{G_W}(Y)| < |Y|=|X| \cdot |X_W^l|$ connected components.
  Now, we apply Reduction Rule \ref{rule::W}. 
  
  \begin{claim}
   Reduction Rule \ref{rule::W} is safe.
  \end{claim}
  \begin{claimproof}
    Let $F$ be an edge dominating set of size at most $k$ in $G$. We can construct an edge dominating set of size at most $k'=k$ in $G'$ as in the proof of Claim \ref{claim::X_W}. We delete every edge $e=\{x,y\} \in F$ if $x,y \in X_W^h$ or if $x \in X_W^h$ and $y$ is isolated in $G'$. Furthermore, we replace each edge $e=\{x,v\} \in F$ with $x \in X_W^h$ and $v \notin X_W^h$ (not isolated in $G'$) by an edge in $\delta_{G'}(v)$. By construction, the resulting set $F'$ is an edge dominating set of size at most $k=k'$ in $G'$.
    
    For the other direction, let $F'$ be an edge dominating set of size at most $k'$ in $G'$. Recall that $k-\MEDS(G-X) < |X|$, which implies that $k'-\MEDS(G'-X') < |X|$ because $\MEDS(G'-X')=\MEDS(G-X)$ and $k=k'$. Thus, there are at most $|X|-1$ connected components $C$ of $G'-X'$ that are incident with more than $\MEDS(C)$ edges of $F'$. Recall, the graph $G_W - N_{G_W}[Y]$ (resp.\ the graph $G_W$ if $X_W^l = \emptyset$) contains a matching $M$ that saturates $R \setminus Y = \{x^i \mid x \in X_W^h \wedge i \in [|X|]\}$. For every vertex $x \in X_W^h$ let $C_x^1, C_x^2, \ldots, C_x^{|X|}$ be the connected components in $\mathcal{C}_W \setminus \mathcal{C}_W^l$ with $\{x^i,s_{C_x^i}\} \in M$. Note that the set of free vertices in these connected components is not adjacent to a vertex in $X_W^l$ because all connected components whose set of free vertices is adjacent to a vertex in $X_W^l$ correspond to a vertex in $N_{G_W}(Y)$. 
    
    Since at most $|X|-1$ connected components $C$ of $G'-X'$ are incident with more that $\MEDS(C)$ edges of $F'$ at least one of the connected components $C_x^1, C_x^2, \ldots, C_x^{|X|}$ is only incident with $\MEDS(C_x^i)$ edges of $F'$. Say, w.l.o.g., that for all $x \in X_W^h$ the connected component $C_x^1$ is only incident with $\MEDS(C_x^1)$ edges of $F'$. (Note that for two different vertices $x,y \in X_W^h$ the connected components $C_x^1$ and $C_y^1$ are different.) Furthermore, the set of free vertices in $C_x^1$ is only adjacent to vertices in $C_x^1$ because we delete $X_W^h$ to obtain $G'$ and every connected component whose set of free vertices is adjacent to a vertex in $X_W^l$ is contained in $\mathcal{C}_W^l$. Since we have only $\MEDS(C_x^1)$ edges to dominate all edges in $C_x^1$ it holds that no edge has an endpoint in $U(C_x^1)$ or one endpoint in $C_x^1$ and the other endpoint in $X$. Let $w_x \in W(C_x^1)$ be a free vertex in $C_x^1$ that is adjacent to vertex $x$ in $G$; hence $\{w_x,x\} \in E(G)$.
    Now, for every $x \in X_W^h$ we delete all edges that are incident with $C_x^1$ from $F'$ and add a minimum edge dominating set in $C_x^1 - w_x$ that covers $N(W(C_x^1))$ (Proposition \ref{proposition::properties} (\ref{proposition::N(W)IN})) as well as the edge $\{w_x,x\}$ to obtain $F$. It holds that $F$ has size $k$ because $\MEDS(C_x^1)=\MEDS(C_x^1 - w_x) +1$.
    It remains to prove that $F$ is an edge dominating set in $G$. The set $V(F)$ contains all vertices in $V(F')$ except some free vertices in the connected components where we change the edge dominating set. But, all neighbors of these vertices are contained in $V(F)$ because these free vertices are only adjacent to vertices in $X_W^h$ (which are contained in $V(F)$) and to vertices in the connected component of $G-X$ they belong to. Thus, $F$ is an edge dominating set of size at most $k$ in $G$.
  \end{claimproof}

  Let $X_U \subseteq U$ be the set of vertices in $X$ that are adjacent to a vertex in $U$, let $X_U^h \subseteq X_U$ be the set of vertices in $X_U$ that are adjacent to the set of uncovered vertices in at least $|X|+1$ connected components, and let $X_U^l = X_U \setminus X_U^h$. Again, by $\mathcal{C}_U^l$ we denote the connected components $C$ in $\mathcal{C}$ with $E(U(C),X_U^l) \neq \emptyset$, and it holds that $|\mathcal{C}_U^l| \leq |X| \cdot |X_U^l|$.
  
  Next, we apply Reduction Rule \ref{rule::U}. Let $\mathcal{C}_S$ be the set of new connected components in $G-X$ that we add during Reduction Rule \ref{rule::U}. 
  To prove that Reduction Rule \ref{rule::U} is safe in the previous kernelization (Lemma \ref{lemma::kernel1}), we only showed that every vertex in $X_U^h$ must be covered by every solution of size at most $k$. The same argumentation holds here; thus, Reduction Rule \ref{rule::U} is safe.
  
  Let $\mathcal{C}_B$ be the set of connected components $C$ in $\mathcal{C} \setminus (\mathcal{C}_W^l \cup \mathcal{C}_U^l \cup \mathcal{C}_S)$ that contain a strongly beneficial set. 
  To find connected components in $\mathcal{C}_B$ that can be safely removed from $G$ we construct an auxiliary graph $G_A$ as follows:
  \begin{itemize}
   \item Add for each set $Y \subseteq X$ with $2 \leq |Y| \leq d$ and for each $\beta \in [|Y|-1]$ a vertex $r_{Y,\beta}$ to $G_A$; denote the union of these vertices by $R$.
   \item Add for each connected component $C$ in $\mathcal{C}_B$ a vertex $s_{C}$ to $G_A$; denote the union of these vertices by $S$.
   \item For each connected component $C \in \mathcal{C}_B$ we add the edge $\{r_{Y,\beta},s_C\}$ to $G_A$ if and only if there exists a strongly beneficial set $B$ of size $|Y|$ in $C$ with $\cost(B)=\beta$ such that there exists a perfect matching $M$ in $(Y \cup B, E(G) \cap E(Y,B))$. Hence, we add an edge between a vertex $s_C$ that represents connected component $C$ and a vertex $r_{Y, \beta}$ if and only if a local solution for $C$ that contains a maximum matching between $Y$ and $B$ (that covers $Y$) increases the cost of a local solution for $C$ only by $\beta$.
  \end{itemize}
  
  Thus, the auxiliary graphs tells us which connected components $C \in \mathcal{C}_B$ can help us to cover a set $Y \subseteq X$ (of size at most $d$) without using $|Y|$ additional edges (or more precisely by using $\beta$ ``additional'' edges). 
  \begin{claim}
   We can construct graph $G_A$ in polynomial time
  \end{claim}
  \begin{claimproof}
   The set $R$ has $\sum_{i=2}^d \binom{|X|}{i} \cdot i \in \Oh(|X|^d)$ many vertices, one for each pair $(Y,\beta)$ where $Y \subseteq X$ with $2 \leq |Y| \leq d$ and $\beta \in [|Y|]$. Hence we can construct $R$ in polynomial time. Since the set $S$ contains one vertex for each connected component in $\mathcal{C}_B$ we can construct $S$ in polynomial time.
   
   For every connected component $C \in \mathcal{C}_B$ we can compute a minimum edge dominating set in $C$ as well as in $C-B$, for every $B \subseteq V(C)$, in polynomial time because $C$ is of constant size. Thus, we can compute all strongly beneficial sets $B$ in $C$ as well as the value $\cost(B)$ in polynomial time: There are only a constant number of possible sets and we can compute $\MEDS(C)$ and $\MEDS(H-B)$ in polynomial time. This is only possible, because $C$ is of constant size. Let $C$ be a connected component in $\mathcal{C}_B$, let $B$ be a strongly beneficial set in $C$, and let $Z=N(B) \cap X$ be the set of vertices in $X$ that are adjacent to a vertex in $B$. For each set $Y \subseteq Z$ of size $|B|$ we add the edge $\{r_{Y,\cost(B)},s_C\}$ to $G_A$, if there exists a perfect matching in $(B \cup Y, E(B,Y))$. We can do this in polynomial time, because we only have to compute $\binom{|Z|}{|B|} \leq |X|^d$ many maximum matchings. We can do this for every connected component and every beneficial set in a component because there are only a polynomial number of connected components and because every connected component has only a constant number of strongly beneficial sets.
  \end{claimproof}
  Now, we apply Theorem \ref{theorem::matching} to obtain either a maximum matching in $G_A$ that saturates $R$, or to find a set $Z \subseteq R$ such that $|N_{G_A}(Z)| < |Z|$ and such that there exists a maximum matching in $G_A - N_{G_A}[Z]$ that saturates $R \setminus Z$. Note that if there exists a matching in $G_A$ that saturates $R$ we can set $Z= \emptyset$. Thus, we can assume that we always find a set $Z$ that fulfills the above properties. 
  Let $M$ be a maximum matching in $G_A - N_{G_A}[Z]$. By choice of $Z$, it holds that the matching $M$ saturates $R \setminus Z$.  
  Now, let $\mathcal{C}_{B}^l$ be the set of connected components $C$ in $\mathcal{C}_B$ that correspond to a vertex $s_C$ in $N_{G_A}(Z)$, and let $\mathcal{C}_B^h=\{ C \in \mathcal{C}_B \mid \exists Y \subseteq X, \beta \in [|Y|-1] \colon \{r_{Y,\beta},s_C\} \in M\}$ be the set of connected components $C$ in $\mathcal{C}_B$ that correspond to a vertex $s_C$ in $S$ with the property that $\{r_{Y,\beta},s_C\}$ is an edge in $M$ for a set $Y \subseteq X$ and an integer $\beta \in [|Y|-1]$. Note that $r_{Y,\beta} \in R \setminus Z$. Now, we can bound the number of connected components in $\mathcal{C}_B^l$ and $\mathcal{C}_B^h$: It holds that $|\mathcal{C}_B^l| = |N_{G_A}(Z)| < |Z|$ (property of $Z$) and that $|\mathcal{C}_B^h| = |R \setminus Z|$ (because we add one vertex for every vertex in $R \setminus Z$ to $\mathcal{C}_B^h$); thus $|\mathcal{C}_B^l \cup \mathcal{C}_B^h| \leq |R| \in \Oh(|X|^d)$.
   
  So far, we know that $|\mathcal{C}_W^l| \leq |X| \cdot |X_W^l|$, that $|\mathcal{C}_U^l| \leq |X| \cdot |X_U^l|$, that $|\mathcal{C}_S| \leq |X_U^h|$, and that $|\mathcal{C}_B^l \cup \mathcal{C}_B^h| \in \Oh(|X|^d)$. We will show that the remaining connected components of $G-X$ can be safely removed by reducing the value $k$ accordingly.
  Let $\mathcal{C}_D$ be the set of connected components in $G-X$ that are not contained in $\mathcal{C}_U^l \cup \mathcal{C}_S \cup \mathcal{C}_W^l \cup \mathcal{C}_B^l \cup \mathcal{C}_B^h$.
  We apply Reduction Rule \ref{rule::delete} to delete all connected components in $\mathcal{C}_D$ and to obtain our reduced instance. 

  \begin{claim}
    Reduction Rule \ref{rule::delete} is safe.
  \end{claim}
  \begin{claimproof}
   Let $(G,k,X)$ be the instance before applying Reduction Rule \ref{rule::delete}, and let $(G',k',X')$ be the instance after applying Reduction Rule \ref{rule::delete}. Note that $X=X'$, because we only delete connected components of $G-X$ and decrease $k$.
   
   Let $F'$ be an edge dominating set of size at most $k'$ in $G'$. To obtain an edge dominating set of size at most $k$ in $G$, we add for each connected component $C$ in $\mathcal{C}_D$ a minimum edge dominating set $F_C$ with $N(W(C)) \subseteq V(F_C)$ to $F'$. The existence of such a minimum edge dominating set follow from Proposition~\ref{proposition::properties}~(\ref{proposition::N(W)IN}). We denote the resulting set by $F$. By construction, the set $F$ has size at most $k$ because $k'=k-\MEDS(\mathcal{C}_D)$ and we add exactly a minimum edge dominating set of $\mathcal{C}_D$ to $F'$ to obtain $F$. 
   To show that $F$ is an edge dominating set of $G$, we have to show that $F$ dominates every edge between $X$ and a connected component $C$ in $\mathcal{C}_D$: All other edges are either dominated by $F'$ or by the added minimum edge dominating set $F_C$ with $C \in \mathcal{C}_D$. 
   Assume for contradiction that there exists a connected component $C \in \mathcal{C}_D$, a vertex $v \in V(C)$, and a vertex $x \in X$ such that $\{v,x\} \in E(G)$ is not dominated by $F$. 
   All vertices in $N(W(C))$ are incident with an edge in $F_C$ (choice of $F_C$) and therefore incident with an edge in $F$. Thus, $v$ must be a free or uncovered vertex because $V(C)=N[W(C)] \cup U(C)$. If $v$ is a free vertex, then $x$ must be contained in $X_W$. Since we delete $X_W^h$ during Reduction Rule \ref{rule::W}, it holds that $x \in X_W^l$. But this implies that $C$ is a connected component in $\mathcal{C}_W^l$ because $\mathcal{C}_W^l$ contains all connected components whose set of free vertices is adjacent to a vertex in $X_W^l$. Thus, $C$ is contained in $\mathcal{C}_W^l$ and not in $\mathcal{C}_D$ which is a contradiction. If $v$ is an uncovered vertex then $x$ must be a vertex in $X_U$. All vertices in $X_U^h$ must be incident with an edge in $F'$ because every vertex $x$ in $X_U^h$ is adjacent to a connected component in $\mathcal{C}_S$ that consists of a single vertex $x'$ which is only adjacent to this vertex $x$ in $X$; hence $x \in X_U^l$. But, this implies that $C$ is contained in $\mathcal{C}_U^l$, and not in $\mathcal{C}_D$, because every connected component whose set of uncovered vertices is adjacent to a vertex in $X_U^l$ is contained in $\mathcal{C}_W^l$. This is a contradiction. Thus, $F$ is an edge dominating set of size at most $k$ in $G$.
   
   For the other direction, let $F$ be an edge dominating set of size at most $k$ in $G$ that is also a matching. Recall, every vertex $v \in W(C)$ that is adjacent to a vertex $x \in X$ is contained $x \in X_W^l$ because we delete $X_W^h$ during Reduction Rule \ref{rule::W}; hence $C \in \mathcal{C}_W^l$. Thus, every connected component $C$ with $E(W(C),X)) \neq \emptyset$ is contained in $\mathcal{C}_W^l$, and therefore, not contained in $\mathcal{C}_D \cup \mathcal{C}_B$.
   Let $C_1,C_2,\ldots, C_p$ be all connected components in $\mathcal{C}_D \cup \mathcal{C}_B^h$ that are incident with an edge in $F$ that has its other endpoint in $X$; thus for all $i \in [p]$ it holds $E(C_i,X) \cap F \neq \emptyset$. For each connected component $C_i$, with $i \in [p]$, let $B_i=\{ c \in V(C_i) \mid \exists x \in X \colon \{x,c\} \in F\}$ be the set of vertices in $V(C_i)$ that are incident with an edge in $F$ that has its other endpoint in $X$. Note that $B_i \subseteq V(C_i) \setminus W(C_i)$ because $E(W(C_i),X) = \emptyset$ for all connected components in $\mathcal{C}_D \cup \mathcal{C}_B$.
   Since $B_i \subseteq V(C) \setminus W(C)$ we can apply Proposition \ref{proposition::properties} (\ref{proposition::decomposeB}): For all $i \in [p]$ let $B_i^1, B_i^2,\ldots, B_i^{q_i} \subseteq B_i$ be a partition of $B_i$ where $B_i^j$ is either strongly beneficial or has $\cost(B_i^j)=|B_i^j|$, for all $j \in [q_i]$, such that $\cost(B_i) \geq \sum_{j=1}^{q_i} \cost(B_i^j)$.
   Let $Y_i^j = \{x \in X \mid \exists v \in V(B_i^j) \colon \{x,v\} \in F \}$, with $i \in [p]$ and $j \in [q_i]$, be the set of vertices in $X$ that are incident with an edge in $F$ whose other endpoint is contained in $B_i^j$, and let $\beta_i^j=\cost(B_i^j)$. Note that the sets $B_i^j$ and $Y_i^j$ have the same size because $F$ is a matching in $G$. Therefore, if $B_i^j$ is a strongly beneficial set then the size of $Y_i^j$ is at most $d$ because every strongly beneficial set has size at most $d$.
   
   Now, if $B_i^j$, with $i \in [p]$ and $j \in [q_i]$, is a strongly beneficial set in $C_i$ then it holds that $2 \leq |B_i^j| = |Y_i^j| \leq d$ and that $\beta_i^j=\cost(B_i^j)<|B_i^j|=|Y_i^j|$. 
   Thus, the vertex $r_{Y_i^j,\beta_i^j}$ is contained in $R$. Note that $C_i$ is a connected component in $\mathcal{C}_B^h$ if $B_i^j$ is strongly beneficial for an index $j \in [q_i]$ because every connected component in $\mathcal{C} \setminus (\mathcal{C}_W^l \cup \mathcal{C}_U^l \cup \mathcal{C}_S)$ that contains a strongly beneficial set is contained in $\mathcal{C}_B$, and because $C_i \in \mathcal{C}_D \cup \mathcal{C}_B^h$. Recall that $\mathcal{C}_D$ contains all connected components that are not contained in $\mathcal{C}_W^l \cup \mathcal{C}_U^l \cup \mathcal{C}_S \cup \mathcal{C}_B^l \cup \mathcal{C}_B^h$.
   Furthermore, the vertex $r_{Y_i^j,\beta_i^j}$ must be contained in $R \setminus Z$: If $r_{Y_i^j,\beta_i^j}$ is a vertex in $Z$ then $N_{G_A}(r_{Y_i^j,\beta_i^j}) \subseteq N_{G_A}(Z)$. But, every vertex in $N_{G_A}(Z)$ corresponds to a connected component in $\mathcal{C}_B^l$ which would imply that $C_i$ is contained in $\mathcal{C}_B^l$ because $r_{C_i} \in N_{G_A}(r_{Y_i^j,\beta_i^j})$ which is a contradiction; thus, $r_{Y_i^j,\beta_i^j} \in R \setminus Z$.
   
   Let $C_i^j$ be the connected component in $\mathcal{C}_B^h$ whose corresponding vertex in $S$ is matched to $r_{Y_i^j,\beta_i^j}$; hence with $\{ r_{Y_i^j,\beta_i^j}, s_{C_i^j}\} \in M$. Since $\{ r_{Y_i^j,\beta_i^j}, s_{C_i^j}\} \in E(G_A)$ there exists a strongly beneficial set $\bar{B}_i^j$ of size $|Y_i^j|$ in $C_i^j$ with $\cost(\bar{B}_i^j)=\beta_i^j(=\cost(B_i^j))$ such that there exists a perfect matching between $Y_i^j$ and $\bar{B}_i^j$ in $(Y_i^j \cup \bar{B}_i^j, E(G) \cap E(Y_i^j, \bar{B}_i^j))$.
   Hence, for every strongly beneficial set $B_i^j$, with $i \in [p]$ and $j \in [q_i]$, the set $Y_i^j$ is associated with a different connected component $C_i^j$ in $\mathcal{C}_B^h$ and a beneficial set $\bar{B}_i^j$ that has the same advantage as $B_i^j$.
   
   In the case that the size of $B_i^j$, with $i \in [p]$ and $j \in [q_i]$, is equal to $\cost(B_i^j)$ it holds that $Y_i^j$ is equal to $\beta_i^j$. Thus, every edge that is incident with a vertex in $Y_i^j$ increases the cost locally by one. 
   
   To construct an edge dominating set of size at most $k'$ in $G'$ we delete all edges in $F$ that are incident with a vertex in a connected component of $\mathcal{C}_D \cup \mathcal{C}_B^h$; denote the resulting set by $\widetilde{F}$. Recall that $C_1,C_2,\ldots, C_p$ are the connected components in $\mathcal{C}_D \cup \mathcal{C}_W^h$ that are incident with an edge in $F$ whose other endpoint is contained in $X$. 
   Next, for each $i \in [p]$ and $j \in [q_i]$ with $B_i^j \subseteq B_i$ strongly beneficial we add a minimum edge dominating set in $C_i^j - \bar{B}_i^j$ with $N(W(C_i^j)) \setminus \bar{B}_i^j \subseteq V(F_C)$ to $\widetilde{F}$ as well as a maximum matching between $\bar{B}_i^j$ and $Y_i^j$ that saturates both sets: Such a minimum edge dominating set exists by assumption and such a matching exists because $\{ r_{Y_i^j,\beta_i^j}, s_{C_i^j}\}$ is an edge in $G_A$. Recall that every connected component $C_i^j$ is contained in $\mathcal{C}_B^h$. Thus, all added edges are contained in $G'$. For all remaining connected components $C$ in $\mathcal{C}_B^h$ we add a minimum edge dominating set in $C$ to $\widetilde{F}$. Finally, we add for each vertex $y$ that is contained in a set $Y_i^j$ with $|Y_i^j|=|B_i^j|=\cost(B_i^j)=\beta_i^j$, where $i \in [p]$ and $j \in [q_i]$, an arbitrary edge in $\delta_{G'}(y)$ to $\widetilde{F}$ if $\delta_{G'}(y)$ is not the empty set. Otherwise, if $v$ is isolated in $G'$ we add no edge to $\widetilde{F}$.  
   We denote the resulting set by $F'$.
   
   First, we show that $F'$ is indeed an edge dominating set of $G'$. Every vertex in $X$ that is incident with an edge in $F$ that has its other endpoint in a connected component of $\mathcal{C}_D \cup \mathcal{C}_B^h$ is contained in a set $Y_i^j$, with $i \in [p]$ and $j \in [q_i]$. During the reduction we delete only edges that are incident with a connected component in $\mathcal{C}_D \cup \mathcal{C}_B^h$, but we also add for each vertex $v$ in $Y_i^j$ that is not isolated in $G'$ an edge in $\delta_{G'}(y)$ to $F'$. Thus, every vertex in $X$ that is covered by $V(F)$ and not isolated in $G'$ is also covered by $V(F')$.
   Furthermore, every edge in a connected component is dominated: We delete only edges that are incident with connected components in $\mathcal{C}_D \cup \mathcal{C}_B^h$ and we add for all connected components in $\mathcal{C}_B^h$ an edge dominating set to $F'$. Since the connected components in $\mathcal{C}_D$ are not contained in $G'$ it holds that $F'$ dominates all edges in $G'-X'$. Thus, the only edges that are possibly not dominated by $F'$ have one endpoint in $X$ and the other endpoint in a connected component of $\mathcal{C}_B^h$ because these are the only connected components where we change the edge dominating set (and because $V(F')$ covers all vertices in $X \cap V(F)$ that are not isolated in $G'$). Assume that there exists a connected component $C \in \mathcal{C}_B^h$, a vertex $v \in V(C)$, and a vertex $x \in X$ such that edge $\{v,x\} \in E(G')$ is not dominated by $F'$; hence $v,x \notin V(F')$. No vertex in $W(C)$ is adjacent to a vertex in $X$ (otherwise $C \in \mathcal{C}_W^l$ and not in $\mathcal{C}_B^h$). Furthermore, every vertex in $N(W(C))$ is incident with an edge in $F'$ (by construction). Thus, $v$ must be a vertex in $U(C)$, and $x$ must be a vertex in $X_U$. Every vertex in $X_U^h$ must be incident with an edge in $F'$ because $X_U^h \subseteq V(F)$, and because every vertex in $X$ that is covered by the edge dominating set $F$ and is not isolated in $G'$ is contained in $V(F')$; hence $x \in X_U^l$. But, if $x \in X_U^l$ then it follows that $C$ is a connected component in $\mathcal{C}_U^l$ and not in $\mathcal{C}_B^h$, which is a contradiction; thus $F'$ is an edge dominating set in $G'$.
   
   It remains to show that $F'$ contains at most $k'$ edges. 
   The connected component $C_i$, with $i \in [p]$, is incident with at least $\MEDS(C_i) + \cost(B_i)$ many edges (definition of $\cost$). For all remaining connected components $C \in \mathcal{C}_D \cup \mathcal{C}_B^h$ we need at least $\MEDS(C)$ many edges. Thus, it holds that $\widetilde{F} \leq |F|-\MEDS(\mathcal{C}_D \cup \mathcal{C}_B^h) - \sum_{i=1}^p \sum_{j=1}^{q_i} \beta_i^j$ because
   \begin{align*}
    \widetilde{F} &\leq |F|-\left(\MEDS(\mathcal{C}_D \cup \mathcal{C}_B^h) + \sum_{i=1}^p \cost(B_i) \right) 
    \leq |F|-\MEDS(\mathcal{C}_D \cup \mathcal{C}_B^h) - \sum_{i=1}^p \sum_{j=1}^{q_i} \cost(B_i^j)
    \\&\leq |F|-\MEDS(\mathcal{C}_D \cup \mathcal{C}_B^h) - \sum_{i=1}^p \sum_{j=1}^{q_i} \beta_i^j
   \end{align*}
    
    To obtain $F'$ we add $\MEDS(C_i^j) + \cost(\bar{B}_i^j)=\MEDS(C_i^j) + \beta_i^j$ edges to $\widetilde{F}$, with $i \in [p]$, and $j \in [q_i]$, if $B_i^j \subseteq B_i$ is a strongly beneficial set in $C_i$: We add a minimum edge dominating set in $C_i^j-\bar{B}_i^j$ as well as a matching between $Y_i^j$ and $\bar{B}_i^j$ that saturates both sets to $\widetilde{F}$, and $\cost(\bar{B}_i^j)=\beta_i^j$. All these connected components $C_i^j$ are contained in $\mathcal{C}_B^h$. 
    For all remaining connected components $C$ in $\mathcal{C}_B^h$ we add a minimum edge dominating set of $C$ to $\widetilde{F}$. 
    Furthermore, for all vertices $y$ that are not isolated in $G'$ and that are contained in a set $Y_i^j$, with $i \in [p]$, and $j \in [q_i]$, where $|Y_i^j|=\beta_i^j=\cost(B_i^j)$, we add an arbitrary edge in $\delta_{G'}(y)$ to $\widetilde{F}$. Thus,
    \[
     |F'|\leq|\widetilde{F}| + \MEDS(\mathcal{C}_B^h) + \sum_{i=1}^p \sum_{j=1}^{q_i} \beta_i^j \leq |F| - \MEDS(\mathcal{C}_D) = k'
    \]
  This completes the proof.
  \end{claimproof}
  We showed that all reduction rules are safe.
  To show that the reduced $(G',k',X')$ instance has only $\Oh(|X|^d)$ vertices, we only have to bound the number of connected components in $G'-X'$ because every connected component has constant size. During the reduction rules we delete all connected components that are not contained in $\mathcal{C}':=\mathcal{C}_W^l \cup \mathcal{C}_U^l \cup \mathcal{C}_S \cup \mathcal{C}_B^h \cup \mathcal{C}_B^l$. We already showed that $|\mathcal{C}_W^l \cup \mathcal{C}_U^l \cup \mathcal{C}_S| \leq 2 \cdot |X|^2$ (see above). Furthermore, we showed that $|\mathcal{C}_B^l \cup \mathcal{C}_B^h| \in \Oh(|X|^d)$. This implies that $G'$ has at most $\Oh(|X|^d)$ connected components, and thus, at most $\Oh(|X|^d)$ vertices. (We assumed that there exists at least one graph $H$ in \cH that has a beneficial set and these beneficial set has at least size two; thus $d \geq 2$.) 
  Next, we have to bound the number of edges. Every connected component has only constant size, thus it has only a constant number of edges; hence $|E[\mathcal{C}']| \in \Oh(|X|^d)$. The number of edges between vertices in $X$ is at most $|X|^2$. All remaining edges are between $X$ and $\mathcal{C}'$. This are at most $|X| \cdot |V(\mathcal{C}')| \in \Oh(|X|^{d+1})$ many edges. This sums up to at most $\Oh(|X|^{d+1})$ edges.
  
  It remains to show that we can perform the reduction in polynomial time. We can compute the sets $W$ and $U$ in polynomial time because every connected component is of constant size, and therefore, we can compute minimum edge dominating sets in every connected component as well as in every subgraph of a connected component in polynomial time. Furthermore, we can construct the auxiliary graphs $G_W$ and $G_A$ in polynomial time. Hence, by applying Theorem \ref{theorem::matching} we can compute the set $X_W^h$, the set $X_W^l$, the set $\mathcal{C}_W^l$, the set $\mathcal{C}_B^l$, and the set $\mathcal{C}_B^h$ in polynomial time. The set $X_U^l$, the set $X_U^h$, as well as the set $\mathcal{C}_U^l$ can be computed in polynomial time by simple counting. Since we can compute all sets in polynomial time, we can apply the reduction rules in polynomial time because we only delete the set $X_W^h$ (Reduction Rule \ref{rule::W}) as well as all connected components that are not contained in $\mathcal{C}_U^l \cup \mathcal{C}_S \cup \mathcal{C}_W^l \cup \mathcal{C}_B^l \cup \mathcal{C}_B^h$ (Reduction Rule \ref{rule::delete}), and we only add one vertex for every vertex in $X_U^h$ (Reduction Rule \ref{rule::U}).
\end{proof}

\begin{theorem}\label{theorem::polynomiallowerbound}
Let $d\in\N$ and let $\cH$ be a finite set of connected graphs such that some $H\in\cH$ has a strongly beneficial set of size $d$. Then \EDS parameterized by the size of a modulator to $\mathcal{H}$-component graphs does not have a kernelization of size $\Oh(|X|^{d-\varepsilon})$, for any $\varepsilon>0$, unless \containment.
\end{theorem}

At a first glance, $d$-dimension $d$-\prob{Set Cover} seems to be a suitable problem to prove Theorem \ref{theorem::polynomiallowerbound} by giving a polynomial parameter transformation from $d$-dimension $d$-\prob{Set Cover} parameterized by the size of the universe to \EDS parameterized by the size of a modulator to $\mathcal{H}$-component graphs, where $\cH$ contains a graph $H$ that has a strongly beneficial set $B$ of size $d$. Indeed, if the beneficial set $B=\{b_1,b_2,\ldots,b_d\}$ has for example $\cost(B)= 1$ then there exists an easy polynomial parameter transformation from $d$-dimension $d$-\prob{Set Cover} parameterized by the size of the universe to \EDS parameterized by the size of a modulator to $\mathcal{H}$-component graphs, where $\cH$ contains a graph $H$ that has a strongly beneficial $B$ set of size $d$: For an instance $((U:=U_1 \dot\cup U_2 \dot\cup \ldots \dot\cup U_d,\mathcal{F}),k)$ of $d$-dimension $d$-\prob{Set Cover} we construct an instance $(G,k',X)$ of \EDS by adding for each set $F=\{u_1,u_2,\ldots,u_d\} \in \mathcal{F}$, where $u_i \in U_i$ for all $i \in [d]$, a copy $H_F$ of $H$ to $G$ as well as two vertices $u$, $u'$ for each element $u \in U$. We add an edge between $u$ and $u'$ for each $u \in U$ as well as an edge between $u$ and the copy of $b_i$ in $H_F$ if $u \in U_i \cap F$. Let $k'=\MEDS(G-X)+k$ and $X$ be the set of all vertices that are not contained in a copy of $H$. In general, if $\cost(B) > 1$ then one would set $k'=\MEDS(G-X)+k \cdot \cost(B)$.

But, there could be cases where it seems unlikely that such a polynomial parameter transformation exists, and where the above construction is not correct. For example, assume that there exists a graph $H$ that has a strongly beneficial set $B=\{b_1,b_2,\ldots,b_{15}\}$ of size 15 with $\cost(B)=5$. It could be possible that also the sets $B_1=\{b_1,b_2,\ldots,b_{10}\}$, $B_2=\{b_5,b_6,\ldots,b_{15}\}$, and $B_3=\{b_1,b_2,\ldots,b_5,b_{10},b_{11},\ldots,b_{15}\}$ are strongly beneficial and that $\cost(B_i)=3$ for all $i \in [3]$. Note that this does not violate the definition of strongly beneficial sets. Now, instead of covering 30 vertices in the modulator by using only edges between the modulator of $X$ and two different copies of $H$ and using 10 edges more than one need to cover these two copies, one can cover 30 vertices in the modulator by using edges between copies of $B_1$, $B_2$ and $B_3$ in different copies of $H$ and using only 9 edges more than one need to cover these three copies. Hence, an edge dominating set in $G$ of size at most $k'$ would not lead to a set cover in $(U,\mathcal{F})$ of size at most $k$, because an edge dominating set in $G$ could contain edges between $X$ and more than $k$ connected components of $G-X$. Thus, this would lead to a set cover with more that $k$ sets, where some sets are only subsets of sets in $\mathcal{F}$.
We can handle this problem by giving a cross-composition of cost $t^{1/d}$ from the \NP-hard \MCC problem.

\begin{proof}
Fix a graph $H\in\cH$ that contains a strongly beneficial set of size $d$, and fix a strongly beneficial set $B=\{b_1,\ldots,b_d\}$ of size $d$ in $H$. If any of the Items \ref{item:lowerbound:CP_Q} through \ref{item:lowerbound:CP_other} of Theorem~\ref{theorem:detailedclassification} applies to $H$ then we already ruled out \emph{any} polynomial kernelization (unless \containment). Thus, it suffices to prove the theorem in the remaining case where we know that $V(H)=N[W(H)]\cup U(W)$, that $B\subseteq N(W(H))$, and that there is a minimum edge dominating set $F_B$ of $H-B$ that covers $N(W(H))\setminus B$.

To prove the theorem, we give a cross-composition of cost $f(t)=t^{1/d}$ from the \NP-hard \MCC problem to \EDS parameterized by the size of a modulator to an \cH-component graph, where $t$ is the number of \MCC instances. We will construct an instance $(G',k',X')$ where all components of $G'-X'$ are isomorphic to $H$, implying that the result holds for all sets $\cH$ containing $H$ (though a stronger lower bound may follow using another $H'\in\cH$).
 
  We choose the same equivalence relation $\mathcal{R}$ as in the proof of Theorem \ref{theorem::P3}. Let a sequence of instances $I_i=(G_i,k)_{i=1}^t$ of \MCC be given that are in the same equivalence class of $\mathcal{R}$. As before, since all color classes have the same size we can identify for each color class the vertex sets. Let $V$ be the vertex set (of size $k \cdot n$) of the $t$ instances and let $V_1,V_2,\ldots, V_k$ be the different color classes (of size $n$). 
  We assume, w.l.o.g.\ that every instance has at least one edge in $E(V_p,V_q)$ for all $1 \leq p < q \leq k$; otherwise, this instance would be a trivial no instance and we can delete it.
  We copy some instance until we have $\widetilde{t}=s^d$ instances, where $s$ is the least odd integer with $t\leq s^d$. It holds that $s=\lceil t^{1/d} \rceil$ or $s=\lceil t^{1/d} \rceil + 1$; hence $s \leq t^{1/d}+2$. Clearly, this does not affect whether at least one instance is yes for MCC.
  
  In the proof of Theorem \ref{theorem::P3} resp.\ Theorem \ref{theorem::LB} we add a set $W$ of size $2 \cdot \log(t)$ to the modulator to encode for each path $P_3$ resp.\ graph $H$ which instance it corresponds to. We cannot apply this construction here because we do not have the ``control set'' $C$ of a control pair; we only have a (large) strongly beneficial set $B$. Therefore, we have to find a different approach to encode to which instance a copy of the graph $H$ corresponds. Like Dell and Marx \cite{DBLP:conf/soda/DellM12} we add $d \cdot s$ vertex sets to the graph $G'$ (more precisely the modulator $X'$), which form $d$ groups of size $s$ each. The goal is to associate each instance with a different choice of $d$ out of the $d\cdot s$ vertex sets, picking one from each group; there are $\widetilde{t}=s^d$ choices. 
  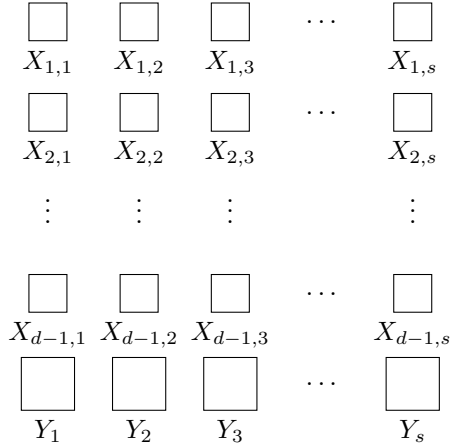
\begin{figure}
   \centering
   \begin{tikzpicture}[scale=0.6]
      \foreach \s/\ls in {1,2,3,5/s}{
	\node at (2*\s -2, -4) {$\vdots$};
	\foreach \t/\lt in {1,2,4/d-1}{
	  \node[draw, rectangle, minimum width=0.5cm, minimum height=0.5cm] (X_\t_\s) at (2*\s - 2, -2 *\t +2) [label=below:$X_{\lt,\ls}$] {};
	 }
	 \node[draw, rectangle, minimum width=0.7cm, minimum height=0.7cm] (Y_\s) at (2*\s -2, -8) [label=below:$Y_\ls$] {};
      }
      \node at (6, -8) {$\ldots$};
      \foreach \t in {1,2,4}{
	\node at (6,-2*\t +2) {$\ldots$};}
   \end{tikzpicture}
   \caption{The $d \cdot s$ sets that encode the $\widetilde{t}=s^d$ instances in the construction of $G'$.}
   \label{figure::instance}
  \end{figure}
  
  We construct an instance $(G',k',X')$ of \EDS parameterized by the size of a modulator to an $H$-component graph; of course, this is also an instance of \EDS parameterized by the size of modulator to an \cH-component graph. (See Figures \ref{figure::instance} and \ref{figure::LBgadgets} for an illustration.)
  
  We add $(d-1) \cdot s$ vertex sets, each of size $\binom{k}{2}$, to $G'$; we denote these sets by $X_{i,j}$ where $i \in [d-1]$ and $j \in [s]$. Every vertex in $X_{i,j}$, with $i \in [d-1]$ and $j \in [s]$, represents a different edge set $E(V_p,V_q)$ for $1\leq p<q \leq k$. By $x_{i,j}^{p,q}$ we denote the vertex in $X_{i,j}$ that represents the edge set $E(V_p,V_q)$.
  Next, we add $s$ sets, each of size $\binom{k}{2} \cdot n^2$, to $G'$. We denote these sets by $Y_j$ with $j \in [s]$. Every vertex in $Y_j$, with $j \in [s]$, represents a possible edge (of a \MCC instance) between two vertices in different color classes $V_p$ and $V_q$, with $1\leq p<q\leq k$. By $y_j^{\{u,v\}}$ we denote the vertex in $Y_j$ that represents the possible edge $\{u,v\}$ with $u \in V_p$, $v \in V_q$.
  
  We modify the indexing of the input instances from using $i$ with $ i \in [\widetilde{t}]$ to using index vectors $h=(h_1,h_2,\ldots, h_d) \in [s]^d$; there are $s^d=\widetilde{t}$ different index vectors. Henceforth, we refer to instances and their graphs through their index $h$. In the rest of the construction, every instance $(G_h,k)$ of \MCC with $h=(h_1,h_2,\ldots, h_d) \in [s]^d$ only interacts with the vertex sets $X_{i,h_i}$ for $i \in [d-1]$ and the vertex set $Y_{h_d}$.
  For every instance $G_h$, with $h \in [s]^d$, we add $|E(G_h)|$ copies of the graph $H$ to $G'$ and denote the copy of $H$ that represents edge $e \in E(G_h)$ by $H_h^e$. Let $B_h^e=\{b_{h,1}^e,b_{h,2}^e,\ldots, b_{h,d}^e\}$ be the copy of the beneficial set $B$ in $H_h^e$, i.e. vertex $b_{h,i}^e$ corresponds to vertex $b_i$ in $B$, for all $i \in [d]$.
  We add all edges $\{b_{h,i}^e,x_{i,h_i}^{p,q} \}$, with $i \in [d-1]$, as well as the edge $\{b_{h,d}^e,y_{h_d}^e\}$ to $G'$, with $h=(h_1,h_2,\ldots,h_d) \in [s]^d$, $e \in E(G_h) \cap E(V_p,V_q)$ and $1 \leq p < q \leq k$. That is, an edge $e$ between color classes $V_p$ and $V_q$ in $G_h$ is (in part) represented by connecting its corresponding graph $H^e_h$ to the sets $X_{i,h_i}$ corresponding to $h$: The $i$th vertex $b^e_{h,i}$ of the beneficial set $B_h^e$ in $H^e_h$ is made adjacent to $x^{p,q}_{i,h_i}\in X_{i,h_i}$, for $i \in [d-1]$. 
  These edges between $H_h^e$ and $X_{i,h_i}$ represent only the colors of the endpoints of $e$. Whereas, the edges between $H_h^e$ and $Y_{h_d}$ represent the endpoints of $e$: The vertex $b_{h,d}^e$ of the beneficial set $B_h^e$ in $H_h^e$ is made adjacent to $y_{h_d}^e$.
  Thus, every vertex $b^e_{h,i}$, with $i \in [d]$, is adjacent to exactly one vertex that is not contained in $V(H_h^e)$. 
  
  We need the sets $X_{i,j}$ with $i \in [d-1]$, and $j \in [s]$ only to encode the $\widetilde{t}$ instances; to make sure that there exists a clique in at least one instance we primarily use the sets $Y_j$ with $j \in [s]$. 
  Our goal is that for every $i \in [d-1]$ exactly one of the sets $X_{i,1}, X_{i,2}, \ldots X_{i,s}$ is contained in the set of endpoints of an edge dominating set of size at most $k'$ in $G'$. We obtain this by means of the following gadget:
  \begin{figure}
   \centering
  \begin{minipage}[t]{0.35\textwidth}
    \centering
    \begin{tikzpicture}[scale=1]
        \foreach \s in {1,2,...,5}{
	  \node[draw, rectangle, inner sep=0.35cm, rotate=360/5 * (\s) + 18] (c_\s) at ({360/5 * (\s) + 18}:1cm) {};
	  \node at ({360/5 * (\s) + 18}:1cm) {{\scriptsize{$X'_\s$}}};
	}
        \foreach \s in {1,2,...,5}{
	  \node[draw, rectangle, inner sep=0.25cm, rotate=360/5 * (\s) + 18] (v_\s) at ({360/5 * (\s) + 18}:2.1cm) {};
	  \node at ({360/5 * (\s) + 18}:2.1cm) {{\scriptsize{$X_\s$}}};
	}
        \foreach \s in {1,2,...,5}{
	  \draw[line width=3] (c_\s) -- (v_\s);
	  \foreach \t in {1,...,5}{
	    \draw[line width=3] (c_\s) -- (c_\t);
        }}
    \end{tikzpicture}
         \subcaption{selection gadget with $s=5$}
         \label{figure::selection_gadget}
    \end{minipage}
    \hspace{3em}
    \begin{minipage}[t]{0.45\textwidth}
    \centering
    \begin{tikzpicture}[scale=0.75]
    \foreach \s in {1,2,...,9}{
      \node[draw, rectangle, minimum width=0.25cm, minimum height=0.25cm] (zp_\s) at (\s-5,1) {};
      \node[draw, rectangle, minimum width=0.25cm, minimum height=0.25cm] (z_\s) at (\s-5,2) {};
      \draw[line width=3] (z_\s) -- (zp_\s);
      }
    \foreach \s in {1,2,3}{
      \node[draw, rectangle, minimum width=0.25cm, minimum height=0.25cm] (t_\s) at (3*\s-6,0) [label=below:$\widetilde{T}_\s$] {};
      }
    \foreach \s in {1,2,3}{
      \draw[line width=3] (t_1) -- (zp_\s);
    }
    \foreach \s in {4,5,6}{
      \draw[line width=3] (t_2) -- (zp_\s);
    }
    \foreach \s in {7,8,9}{
      \draw[line width=3] (t_3) -- (zp_\s);
    }
    \foreach \s in {1,2,...,27}{
      \node[draw, circle, inner sep=1pt, fill] (y_\s) at (0.33*\s-4.62,3) {};}
    \node at (5,1) {$Z'$};
    \node at (5,2) {$Z$};
    \node at (5,3)  {$Y$};
    \foreach \i/\s in {1/4,2/5,3/6}{
      \draw[fill] ($(z_1.north) + (-0.6mm,0)$) -- ($(z_1.north) + (0.6mm,0)$) -- (y_\i.south) -- ($(z_1.north) + (-0.5mm,0)$) -- cycle ;
      \draw[fill] ($(z_\s.north) + (-0.6mm,0)$) -- ($(z_\s.north) + (0.6mm,0)$) -- (y_\i.south) -- ($(z_\s.north) + (-0.5mm,0)$) -- cycle ;
    }
    \foreach \i/\s in {4,5,6}{
      \draw[fill] ($(z_2.north) + (-0.6mm,0)$) -- ($(z_2.north) + (0.6mm,0)$) -- (y_\i.south) -- ($(z_2.north) + (-0.5mm,0)$) -- cycle ;
      \draw[fill] ($(z_\s.north) + (-0.6mm,0)$) -- ($(z_\s.north) + (0.6mm,0)$) -- (y_\i.south) -- ($(z_\s.north) + (-0.5mm,0)$) -- cycle ;
    }
    \foreach \i/\s in {7/4,8/5,9/6}{
      \draw[fill] ($(z_3.north) + (-0.6mm,0)$) -- ($(z_3.north) + (0.6mm,0)$) -- (y_\i.south) -- ($(z_3.north) + (-0.5mm,0)$) -- cycle ;
      \draw[fill] ($(z_\s.north) + (-0.6mm,0)$) -- ($(z_\s.north) + (0.6mm,0)$) -- (y_\i.south) -- ($(z_\s.north) + (-0.5mm,0)$) -- cycle ;
    }
    \foreach \i/\s in {10/7,13/8,16/9}{
      \draw[fill] ($(z_1.north) + (-0.6mm,0)$) -- ($(z_1.north) + (0.6mm,0)$) -- (y_\i.south) -- ($(z_1.north) + (-0.5mm,0)$) -- cycle ;
      \draw[fill] ($(z_\s.north) + (-0.6mm,0)$) -- ($(z_\s.north) + (0.6mm,0)$) -- (y_\i.south) -- ($(z_\s.north) + (-0.5mm,0)$) -- cycle ;
    }
    \foreach \i/\s in {11/7,14/8,17/9}{
      \draw[fill] ($(z_2.north) + (-0.6mm,0)$) -- ($(z_2.north) + (0.6mm,0)$) -- (y_\i.south) -- ($(z_2.north) + (-0.5mm,0)$) -- cycle ;
      \draw[fill] ($(z_\s.north) + (-0.6mm,0)$) -- ($(z_\s.north) + (0.6mm,0)$) -- (y_\i.south) -- ($(z_\s.north) + (-0.5mm,0)$) -- cycle ;
    }
    \foreach \i/\s in {12/7,15/8,18/9}{
      \draw[fill] ($(z_3.north) + (-0.6mm,0)$) -- ($(z_3.north) + (0.6mm,0)$) -- (y_\i.south) -- ($(z_3.north) + (-0.5mm,0)$) -- cycle ;
      \draw[fill] ($(z_\s.north) + (-0.6mm,0)$) -- ($(z_\s.north) + (0.6mm,0)$) -- (y_\i.south) -- ($(z_\s.north) + (-0.5mm,0)$) -- cycle ;
    }
    \foreach \i/\s in {19/4,20/5,21/6}{
      \draw[fill] ($(z_7.north) + (-0.6mm,0)$) -- ($(z_7.north) + (0.6mm,0)$) -- (y_\i.south) -- ($(z_7.north) + (-0.5mm,0)$) -- cycle ;
      \draw[fill] ($(z_\s.north) + (-0.6mm,0)$) -- ($(z_\s.north) + (0.6mm,0)$) -- (y_\i.south) -- ($(z_\s.north) + (-0.5mm,0)$) -- cycle ;
    }
    \foreach \i/\s in {22/4,23/5,24/6}{
      \draw[fill] ($(z_8.north) + (-0.6mm,0)$) -- ($(z_8.north) + (0.6mm,0)$) -- (y_\i.south) -- ($(z_8.north) + (-0.5mm,0)$) -- cycle ;
      \draw[fill] ($(z_\s.north) + (-0.6mm,0)$) -- ($(z_\s.north) + (0.6mm,0)$) -- (y_\i.south) -- ($(z_\s.north) + (-0.5mm,0)$) -- cycle ;
    }
    \foreach \i/\s in {25/4,26/5,27/6}{
      \draw[fill] ($(z_9.north) + (-0.6mm,0)$) -- ($(z_9.north) + (0.6mm,0)$) -- (y_\i.south) -- ($(z_9.north) + (-0.5mm,0)$) -- cycle ;
      \draw[fill] ($(z_\s.north) + (-0.6mm,0)$) -- ($(z_\s.north) + (0.6mm,0)$) -- (y_\i.south) -- ($(z_\s.north) + (-0.5mm,0)$) -- cycle ;
    }
    \end{tikzpicture}
         \subcaption{clique gadget with $k=3$ and $|V_i|=3$}
         \label{figure::clique_gadget}
    \end{minipage}
  \caption{Gadgets with notation as in the definition}
  \label{figure::LBgadgets}
  \end{figure}
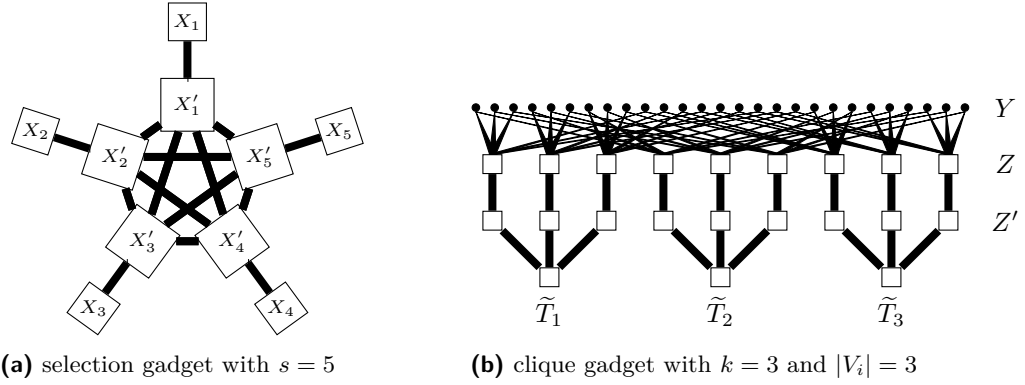
  
  A \emph{selection gadget} of size $\alpha$ consists of $s$ sets, say $X'_1,X'_2,\ldots,X'_s$, each of size $2 \cdot \binom{k}{2} \cdot d$ and $s$ sets, say $X_1,X_2,\ldots, X_s$, each of size $\alpha$. Each vertex in $X'_j$, with $j \in [s]$, is connected to all vertices in $X_j$ and to all vertices in $X'_{\bar{j}}$ with $\bar{j} \in [s]$ and $\bar{j} \neq j$ (see Figure \ref{figure::selection_gadget} for an illustration).
  Intuitively, with a local budget of $(s-1)/2\cdot 2 \cdot \binom{k}{2} \cdot d$ edges, one could cover exactly all $(s-1)\cdot 2\cdot\binom{k}{2}$ vertices of all but one set $X'_j$ by picking appropriate edges with endpoints in the other sets; here we use that $s$ is odd. This would force us to cover edges between $X'_j$ and $X_j$ by making all vertices in $X_j$ endpoints of solution edges with other endpoint outside of th selection gadget.
  We add $d-1$ selection gadgets of size $\binom{k}{2}$ to the modulator $X'$ and identify for each $j \in [s]$ the set $X_{i,j}$ with the set $X_j$ of one selection gadget. The $i$th selection gadget has the vertex sets $X_{i,1}, X_{i,2},\ldots, X_{i,s}$ and $X'_{i,1}, X'_{i,2},\ldots, X'_{i,s}$ where $i \in [d-1]$.
  
  We still have to make sure that we pick $\binom{k}{2}$ edges that have their endpoints in a vertex set of size $k$, so that they must form a $k$-clique. To guarantee this, we add for each set $Y_j$, with $j \in [s]$, a gadget that we call clique gadget (see Figure \ref{figure::clique_gadget}) to the graph $G'$:
  A \emph{clique gadget} consists of $k$ sets $\widetilde{T}_1,\widetilde{T}_2,\ldots,\widetilde{T}_k$ each of size $k-1$; every set represents one color class. Additionally, we add for every vertex $v \in V$ a set $Z_v$ of size $k-1$ as well as a copy of $Z_v$, named $Z'_v$, to the gadget. The final set $Y$ of the gadget contains $\binom{k}{2} \cdot n^2$ vertices, one for each possible edge in an instance. (Each edge has its endpoints in different color classes and we have $k$ color classes of size $n$.) We denote the vertex in $Y$ that represents the edge between vertex $v$ and $u$ in different color classes by $y^{\{v,u\}}$. (Later we will identify $Y$ with one set $Y_j$ for $j \in [s]$.)
  We connect every vertex in $Z_v$ for a vertex $v \in V_i$, with $i \in [k]$, to all vertices in $Z'_v$, and to every vertex $y^{\{v,u\}}$ with $u \in V \setminus V_i$. Furthermore, we connect every vertex in $\widetilde{T}_i$, for $i \in [k]$, to all vertices in $Z'_v$, if $v \in V_i$, i.e., if $v$ has color $i$.
  
  This gadget is perhaps the most vital part of our construction (apart from understanding strongly beneficial sets in $H$-components). There are two different cases for its behavior, which we will trigger by another selection gadget. 
  If there are no other constraints then it can be covered entirely by picking edges connecting sets $Z_v$ to sets $Z'_v$ (using $(k-1) \cdot k \cdot n$ edges). Else, as we will ensure for exactly one of these gadgets, the vertices in all sets $\widetilde{T}_i$ must be endpoints of solution edges because they have neighbors outside the gadget that are not contained in the solution. Nevertheless, we only want (have) $(k-1) \cdot k \cdot n$ solution edges that can be inside a clique gadget and we have to cover the vertices in all sets $\widetilde{T}_i$ only with this budget. To cover the vertices in $\widetilde{T}_i$ we add a matching between $\widetilde{T}_i$ and the vertices in one set $Z'_v$ for a single $v\in V_i$ to the solution. These $k$ vertices will be the vertices of a clique in one instance. Since, all sets $\widetilde{T}_i$, with $i \in [k]$, are covered, we can select $k-1$ edges each between the vertices of a set $Z_v$ (where $Z'_v$ is not covered by solution edges to $\widetilde{T}_1, \cup \ldots \cup \widetilde{T}_k$). We can pick these edges such that all, except the edges between $Z_v$ (where $v$ is a vertex in the ``clique'') and vertices in $y$ that represent an edge between $v$ and another vertex in the ``clique'', are dominated. We will dominate these $\binom{k}{2}$ edges in $Y$ that are incident with a non dominated edge via edges that have one endpoint in a copy of $H$, more precisely the vertex $b_d$ in a copy of $H$. This will guarantee that the edge between two clique vertices is an edge in the instance.
  
  As mentioned above, we add $s$ clique gadgets to the modulator $X'$ and identify each $Y_j$, here $j \in [s]$, with a different set $Y$ of a clique gadget. To distinguish between sets in the different clique gadgets, we denote the other sets for the clique gadget containing set $Y_j$, with $j \in [s]$, by $T_{j,i}$, with $i \in [k]$, and by $Z_{j,v}$ and $Z'_{j,v}$, with $v \in V$; let $Z_j=\bigcup_{v \in V} Z_{j,v}$, let $Z'_j=\bigcup_{v \in V} Z'_{j,v}$ and let $T_j:=\bigcup_{i=1}^k T_{j,i}$.
  Now, we have $s$ different clique gadgets and want to choose exactly one clique gadget where the set $T_j$ must be covered by the solution $F$. To this end, we add one last selection gadget of size $k \cdot (k-1)$ to $X'$ and identify the set $T_j$ with the set $X_j$ of the selection gadget, with $j \in [s]$. We denote the sets $X_j'$ of the selection gadget, with $j \in [s]$, by $T'_j$.
  
  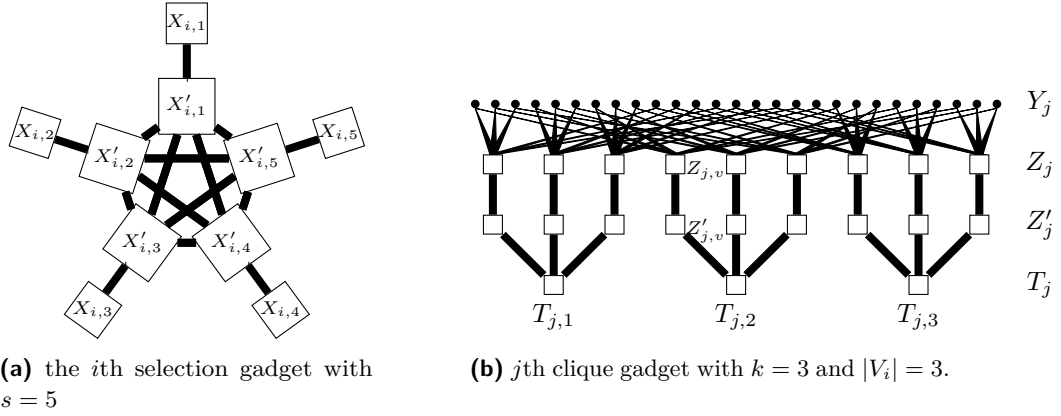
\begin{figure}
   \centering
  \begin{minipage}[t]{0.35\textwidth}
    \centering
    \begin{tikzpicture}[scale=1]
        \foreach \s in {1,2,...,5}{
	  \node[draw, rectangle, inner sep=0.37cm, rotate=360/5 * (\s) + 18] (c_\s) at ({360/5 * (\s) + 18}:1cm) {};
	  \node at ({360/5 * (\s) + 18}:1cm) {{\scriptsize{$X'_{i,\s}$}}};
	}
        \foreach \s in {1,2,...,5}{
	  \node[draw, rectangle, inner sep=0.27cm, rotate=360/5 * (\s) + 18] (v_\s) at ({360/5 * (\s) + 18}:2.1cm) {};
	  \node at ({360/5 * (\s) + 18}:2.1cm) {{\scriptsize{$X_{i,\s}$}}};
	}
        \foreach \s in {1,2,...,5}{
	  \draw[line width=3] (c_\s) -- (v_\s);
	  \foreach \t in {1,...,5}{
	    \draw[line width=3] (c_\s) -- (c_\t);
        }}
    \end{tikzpicture}
         \subcaption{the $i$th selection gadget with $s=5$}
    \end{minipage}
    \hspace{3em}
    \begin{minipage}[t]{0.45\textwidth}
    \centering
    \begin{tikzpicture}[scale=0.8]
    \foreach \s in {1,2,...,9}{
      \node[draw, rectangle, minimum width=0.25cm, minimum height=0.25cm] (zp_\s) at (\s-5,1) {};
      \node[draw, rectangle, minimum width=0.25cm, minimum height=0.25cm] (z_\s) at (\s-5,2) {};
      \draw[line width=3] (z_\s) -- (zp_\s);
      }
    \foreach \s in {1,2,3}{
      \node[draw, rectangle, minimum width=0.25cm, minimum height=0.25cm] (t_\s) at (3*\s-6,0) [label=below:$T_{j,\s}$] {};
      }
    \foreach \s in {1,2,3}{
      \draw[line width=3] (t_1) -- (zp_\s);
    }
    \foreach \s in {4,5,6}{
      \draw[line width=3] (t_2) -- (zp_\s);
    }
    \foreach \s in {7,8,9}{
      \draw[line width=3] (t_3) -- (zp_\s);
    }
    \foreach \s in {1,2,...,27}{
      \node[draw, circle, inner sep=1pt, fill] (y_\s) at (0.33*\s-4.62,3) {};}
    \node at (5,1) {$Z_j'$};
    \node at (5,2) {$Z_j$};
    \node at (5,3)  {$Y_j$};
    \node at (5,0) {$T_j$};
    \node at (-0.5,1.9) {{\scriptsize$Z_{j,v}$}};
    \node at (-0.5,0.9) {{\scriptsize$Z'_{j,v}$}};
    \foreach \i/\s in {1/4,2/5,3/6}{
      \draw[fill] ($(z_1.north) + (-0.6mm,0)$) -- ($(z_1.north) + (0.6mm,0)$) -- (y_\i.south) -- ($(z_1.north) + (-0.5mm,0)$) -- cycle ;
      \draw[fill] ($(z_\s.north) + (-0.6mm,0)$) -- ($(z_\s.north) + (0.6mm,0)$) -- (y_\i.south) -- ($(z_\s.north) + (-0.5mm,0)$) -- cycle ;
    }
    \foreach \i/\s in {4,5,6}{
      \draw[fill] ($(z_2.north) + (-0.6mm,0)$) -- ($(z_2.north) + (0.6mm,0)$) -- (y_\i.south) -- ($(z_2.north) + (-0.5mm,0)$) -- cycle ;
      \draw[fill] ($(z_\s.north) + (-0.6mm,0)$) -- ($(z_\s.north) + (0.6mm,0)$) -- (y_\i.south) -- ($(z_\s.north) + (-0.5mm,0)$) -- cycle ;
    }
    \foreach \i/\s in {7/4,8/5,9/6}{
      \draw[fill] ($(z_3.north) + (-0.6mm,0)$) -- ($(z_3.north) + (0.6mm,0)$) -- (y_\i.south) -- ($(z_3.north) + (-0.5mm,0)$) -- cycle ;
      \draw[fill] ($(z_\s.north) + (-0.6mm,0)$) -- ($(z_\s.north) + (0.6mm,0)$) -- (y_\i.south) -- ($(z_\s.north) + (-0.5mm,0)$) -- cycle ;
    }
    \foreach \i/\s in {10/7,13/8,16/9}{
      \draw[fill] ($(z_1.north) + (-0.6mm,0)$) -- ($(z_1.north) + (0.6mm,0)$) -- (y_\i.south) -- ($(z_1.north) + (-0.5mm,0)$) -- cycle ;
      \draw[fill] ($(z_\s.north) + (-0.6mm,0)$) -- ($(z_\s.north) + (0.6mm,0)$) -- (y_\i.south) -- ($(z_\s.north) + (-0.5mm,0)$) -- cycle ;
    }
    \foreach \i/\s in {11/7,14/8,17/9}{
      \draw[fill] ($(z_2.north) + (-0.6mm,0)$) -- ($(z_2.north) + (0.6mm,0)$) -- (y_\i.south) -- ($(z_2.north) + (-0.5mm,0)$) -- cycle ;
      \draw[fill] ($(z_\s.north) + (-0.6mm,0)$) -- ($(z_\s.north) + (0.6mm,0)$) -- (y_\i.south) -- ($(z_\s.north) + (-0.5mm,0)$) -- cycle ;
    }
    \foreach \i/\s in {12/7,15/8,18/9}{
      \draw[fill] ($(z_3.north) + (-0.6mm,0)$) -- ($(z_3.north) + (0.6mm,0)$) -- (y_\i.south) -- ($(z_3.north) + (-0.5mm,0)$) -- cycle ;
      \draw[fill] ($(z_\s.north) + (-0.6mm,0)$) -- ($(z_\s.north) + (0.6mm,0)$) -- (y_\i.south) -- ($(z_\s.north) + (-0.5mm,0)$) -- cycle ;
    }
    \foreach \i/\s in {19/4,20/5,21/6}{
      \draw[fill] ($(z_7.north) + (-0.6mm,0)$) -- ($(z_7.north) + (0.6mm,0)$) -- (y_\i.south) -- ($(z_7.north) + (-0.5mm,0)$) -- cycle ;
      \draw[fill] ($(z_\s.north) + (-0.6mm,0)$) -- ($(z_\s.north) + (0.6mm,0)$) -- (y_\i.south) -- ($(z_\s.north) + (-0.5mm,0)$) -- cycle ;
    }
    \foreach \i/\s in {22/4,23/5,24/6}{
      \draw[fill] ($(z_8.north) + (-0.6mm,0)$) -- ($(z_8.north) + (0.6mm,0)$) -- (y_\i.south) -- ($(z_8.north) + (-0.5mm,0)$) -- cycle ;
      \draw[fill] ($(z_\s.north) + (-0.6mm,0)$) -- ($(z_\s.north) + (0.6mm,0)$) -- (y_\i.south) -- ($(z_\s.north) + (-0.5mm,0)$) -- cycle ;
    }
    \foreach \i/\s in {25/4,26/5,27/6}{
      \draw[fill] ($(z_9.north) + (-0.6mm,0)$) -- ($(z_9.north) + (0.6mm,0)$) -- (y_\i.south) -- ($(z_9.north) + (-0.5mm,0)$) -- cycle ;
      \draw[fill] ($(z_\s.north) + (-0.6mm,0)$) -- ($(z_\s.north) + (0.6mm,0)$) -- (y_\i.south) -- ($(z_\s.north) + (-0.5mm,0)$) -- cycle ;
    }
    \end{tikzpicture}
         \subcaption{$j$th clique gadget with $k=3$ and $|V_i|=3$.}
    \end{minipage}
  \caption{Gadgets with notation as in the construction of $G'$}
  \label{figure::LBgadgets_new}
  \end{figure}
  
  The set $X'$ contains all vertices that are not contained in a copy of graph $H$; in total these are $(d-1) \cdot s \cdot \binom{k}{2} \cdot (1 + 2 \cdot d)$ vertices for the $d-1$ selection gadgets of size $\binom{k}{2}$, $s \cdot ( (k + 2 \cdot k \cdot n) \cdot (k-1) + \binom{k}{2} \cdot n^2)$ vertices for the $s$ clique gadgets and $s \cdot 2 \cdot \binom{k}{2} \cdot d$ vertices for the last selection gadget (only those that we did not already count, because they are also contained in a clique gadget); hence $|X'| \in \Oh(s \cdot n^2 \cdot k^2 \cdot d^2)= \Oh(t^{1/d} \cdot poly(n))$, because $t^{1/d} \geq s-2$ and $d,k \leq n$. 
  Let 
  \[
  k'= d \cdot 2 \cdot \binom{k}{2} \cdot d \cdot \frac{s-1}{2} + s \cdot k \cdot n \cdot (k-1) + \sum_{h \in [s]^d} |E(G_h)| \cdot \MEDS(H) + \binom{k}{2} \cost(B):
  \]
  Intuitively, we have a local budget of $2 \cdot \binom{k}{2} \cdot d \cdot \frac{s-1}{2}$ to dominate all edges of the complete $s$-partite graph that is contained in each of the $d$ selection gadgets, a local budget of $k \cdot n \cdot (k-1)$ to dominate all edges between $Z_j$ and $Z'_j$ in each of the $s$ clique gadgets, a local budget of $\sum_{h \in [s]^d} |E(G_h)| \cdot \MEDS(H)$ to dominate all edges of $G-X$, and an extra budget of $\binom{k}{2} \cost(B)$ edges to dominate all remaining edges.
  
  We have to show that there exists an index vector $h^* \in [s]^d$ such that $(G_{h^*},k)$ is a $\mathrm{YES}$-instance if and only if $(G',k',X')$ is a $\mathrm{YES}$-instance. 
  
   $(\Rightarrow:)$ Assume that for some index vector $h^* \in [s]^d$ the \MCC instance $(G_{h^*}, k)$ is a $\mathrm{YES}$-instance. Let $X=\{x_1,x_2,\ldots,x_k\} \subseteq V$ be a multicolored clique of size  $k$ in $G_{h^*}$ with $x_i \in V_i$ for $i \in [k]$ and let $E'$ be the set of edges of the clique $X$. Let $h^*=(h^*_1,h^*_2,\ldots,h^*_d) \in [s]^d$. We construct an edge dominating set $F$ of $G'$ as follows:
   
   For each $i \in [d-1]$ we add a minimum edge dominating set in $G'[X'_{i,1} \cup X'_{i,2} \cup \ldots \cup X'_{i,s}]$ of size $2 \cdot \binom{k}{2} \cdot d \cdot \frac{s-1}{2}$ to $F$ such that each set, except the set $X'_{i,h^*_i}$, is covered by $F$. Such a minimum edge dominating set exists, because $G'[X'_{i,1} \cup X'_{i,2} \cup \ldots \cup X'_{i,s}]$ is a complete $s$-partite graph and $s$ is odd. 
   Thus, we dominate all edges, except the edges between the vertex sets $X'_{i,h^*_i}$ and $X_{i,h^*_i}$ in these $d-1$ selection gadgets.
   
    Next, we add a minimum edge dominating set in $G'[T'_1 \cup T'_2 \cup \ldots \cup T'_s]$ of size $2 \cdot \binom{k}{2} \cdot d \cdot \frac{s-1}{2}$ to $F$ such that each vertex set, except the set $T'_{h^*_d}$, is covered by $F$. (Such a minimum edge dominating set exists for the same reasons as above.)
    Consider the $s$ clique gadgets: For each $j \in \{1,2,\ldots,s\} \setminus \{h^*_d\}$ we add a perfect matching between the vertex sets $Z_j$ and $Z'_j$ to $F$; such a matching of size $|V| \cdot (k-1) = n \cdot k \cdot (k-1)$ exists by construction (for each $v \in V$ it holds that every vertex in $Z_{j,v} \subseteq Z_j$ is connected to every vertex in $Z'_{j,v} \subseteq Z'_j$ and both sets have the same size). Thus, all edges in these clique gadgets and between these clique gadgets and the selection gadget are dominated: The only uncovered vertices in a clique gadget are the vertices $Y_j$ and $T_j$, with $j \in [s]$ and $j \neq h^*_d$. These sets are independent sets and only the set $T_j$ is adjacent to a selection gadget, more precisely, to the vertex set $T'_j$ which is covered by $F$. 
    
    Since the edges between the vertex sets $T'_{h^*_d}$ and $T_{h^*_d}$ are not dominated so far, we add a perfect matching between the vertex sets $T_{h^*_d}$ and $\{Z'_{h^*_d,x} \mid x \in X\}$ to $F$; such a matching of size $|X| \cdot (k-1) = k \cdot (k-1)$ exists by construction: the set $T_{h^*_d,i}$, with $i \in [k]$, has size $k-1$ and every vertex in $T_{h^*_d,i}$ is connected to all $k-1$ vertices in $Z'_{h^*_d,x_i}$ with $x_i \in X \cap V_i$. Thus, we covered all edges inside the selection gadget of size $k \cdot (k-1)$ and between this selection gadget and the clique gadgets. 
    Next, we add for all $v \in V_i \setminus X$, with $i \in [k]$, a perfect matching between $Z_{h^*_d,v}$ and $\{y_{h^*_d}^{\{v,x\}} \mid x \in X - x_i\}$: Both sets have size $k-1$ and every vertex in $Z_{h^*_d,v}$ is adjacent to all vertices in $\{y_{h^*_d}^{\{v,x\}} \mid x \in X - x_i\}$. In total, these are $|V \setminus X| \cdot (k-1) = (n-1) \cdot k \cdot (k-1)$ edges. 
    
    So far, we dominate all edges, except the edges between vertices in $Z_{h^*_d,x}$, with $x \in X$, and the vertices in $\{y_{h^*_d}^{\{x,y\}} \mid x,y \in X, x \neq y\}$: The sets $T_{h^*_d}$, $Z'_{h^*_d,x}$, with $x \in X$, and $Z_{h^*_d,v}$, with $v \in V \setminus X$, are covered by $F$. Thus, the only edges that are not dominated in this clique gadget are those between the vertex sets $Z_{h^*_d,x}$ with $x \in X$ and $Y_{h^*_d}$. A vertex in $Z_{h^*_d,x_i}$, with $i \in [k]$, is adjacent to a vertex $y_{h^*_d}^{\{u,v\}}$ if $u=x_i$ and $v \in V \setminus V_i$. But, for all $v \in V \setminus (V_i \cup X)$, the vertex $y_{h^*_d}^{\{x_i,v\}}$ is already covered by $F$ (see above).
    
    Finally, we add an edge dominating set for the copies of $H$ to $F$. For all graphs $H_h^e$ with $h \in [s]^d$, and $e \in E(G_h)$, and either $h \neq h^*$ or $e \notin E'$ we add a minimum edge dominating set in $H_h^e$ that covers all vertices in $B_h^e$ to $F$; such a minimum edge dominating set exists by assumption. (Recall, $E'$ is the set of edges between vertices in $X$.)
    For all graphs $H_{h^*}^e$ with $e=\{x_p,x_q\} \in E'$, with $1 \leq p < q \leq k$, we add a minimum edge dominating set in $H_{h^*}^e -B_{h^*}^e$ to $F$ as well as the edges $\{b_{h^*,i}^e, x_{i,h^*_i}^{p,q}\}$, with $i \in [d-1]$, and the edge $\{b_{h^*,d}^e, y_{h^*_d}^e\}$. These edges exist by construction, because $E' \subseteq E(G_{h^*})$. Thus, the set $V(F)$ contains the vertex set $X_{i,h^*_i}$, with $i \in [d-1]$, and the vertex set $\{y_{h^*_d}^{\{x,y\}} \mid x,y \in X, x \neq y\}$, which implies that $F$ dominates all edges that are contained in a clique gadget and in a selection gadget. 
    
    Since all vertices in $B_h^i$, with $h \in [s]^d$ and $i \in [k]$, are dominated by $F$ and these are the only vertices in the connected component of $H_h^e$ that are adjacent to a vertex in $X'$, and $F$ dominates all clique gadgets, selection gadgets, and connected components of $G'-X'$, the set $F$ is an edge dominating set of $G'$.
    
    The set $F$ contains $d \cdot 2 \cdot \binom{k}{2} \cdot d \cdot \frac{s-1}{2}$ edges inside the selection gadgets, $(s-1) \cdot n \cdot k \cdot (k-1)$ edges inside the clique gadgets that do not contain $Y_{h^*_d}$, $k \cdot (k-1) + (n-1) \cdot k \cdot (k-1)$ edges inside the clique gadget that contains $Y_{h^*_d}$, $\MEDS(H)$ edges for all graphs $H_h^e$ with $h \in [s]^d$, $e \in E(G_h)$ and either $h \neq h^*$ or $e \notin E'$, and $\MEDS(H-B)+|B|=\MEDS(H)+\cost(B)$ edges for all graphs $H_{h^*}^e$ with $e \in E'$. This sums up to $k'$, implying that $(G',k',X')$ is a $\mathrm{YES}$-instance.
    
    $(\Leftarrow:)$ Assume that $(G',k',X')$ is a $\mathrm{YES}$-instance of $\EDS$ and let $F$ be an edge dominating set of size at most $k'$ in $G'$. 
    First, we consider how the edge dominating set $F$ interacts with the graph $G'$: 
    \begin{itemize}
     \item We need at least $k \cdot n \cdot (k-1)$ edges to dominate all edges between $Z_j=\bigcup_{v \in V} Z_{j,v}$, and $Z'_j=\bigcup_{v \in V} Z'_{j,v}$ in one clique gadget, because $G'[Z_{j,v} \cup Z'_{j,v}]$ is a complete bipartite graph whose bipartition has the parts $Z_{j,v}$ and $Z'_{j,v}$ for all $v \in V$ and $|Z_{j,v}|=|Z'_{j,v}|=k-1$. Thus, at least $k \cdot n \cdot (k-1)$ edges of $F$ must be contained inside a clique gadget because we need at least $k \cdot n \cdot (k-1)$ edges to dominate all edges between $Z_j$ and $Z'_j$, with $j \in [s]$, and because these sets are only adjacent to vertices inside the clique gadget they belong to.
      \item Furthermore, $F$ contains at least $2 \cdot \binom{k}{2} \cdot d \cdot \frac{s-1}{2}$ edges inside each selection gadget, because the $s$ sets $X'_{i,1}, X'_{i,2}, \ldots, X'_{i,s}$, with $i \in [d-1]$, resp.\ the $s$ sets $T'_1, T'_2, \ldots, T'_s$ of each selection gadget form a complete $s$-partite graph where each partition has size $2 \cdot \binom{k}{2} \cdot d$ and these sets are only adjacent to vertices in their selection gadget. Note that, the sets $T_1, T_2,\ldots, T_s$ are contained in one selection gadget and in the clique gadgets; but our counting is still correct, because each of the $k \cdot n \cdot (k-1)$ edges that are contained in the clique gadgets must have at least one endpoint in the vertex set $Z_j \cup Z'_j$, with $j \in [s]$, and each of the $2 \cdot \binom{k}{2} \cdot d \cdot \frac{s-1}{2}$ edges in the selection gadgets must have at least one endpoint in the vertex set $T'_1 \cup T'_2 \cup \ldots \cup T'_s$, and because the sets $Z_j \cup Z'_j$ and $T'_1 \cup T'_2 \cup \ldots \cup T'_s$ are not adjacent.
      \item To dominate all edges in $H_h^e$, with $h \in [s]^d$ and $e \in E(G_h)$, we need at least $\MEDS(H)$ edges that are adjacent to $V(H_h^e)$. Thus, we need at least $\sum_{h \in [s]^d} |E(G_h)| \cdot \MEDS(H)$ edges to dominate all edges of $G'-X'$.
    \end{itemize}
    
    Summarizing, for all, except $\binom{k}{2} \cdot \cost(B)$ edges of $F$ we know at least one endpoint and that these edges are either contained in a selection gadget, a clique gadget, or adjacent to a copy of $H$. 
    During the proof, we will show that we can make some assumptions about the edge dominating set $F$. To achieve these assumptions, we replace some edges in $F$ such that the resulting graph is still an edge dominating set of size $|F|$ in $G'$. But, a replacement of an edge will always preserve the previous assumptions.

    \begin{claim} \label{claim::selection_gadget}
     There exists an edge dominating set $F'$ of size $k'$ in $G'$ such that for each $i \in [d-1]$ there exists exactly one $j \in [s]$ such that no vertex in $X'_{i,j}$ is covered by $F'$ (hence $X'_{i,j} \cap V(F')=\emptyset$), and there exists exactly one $j \in [s]$ such that no vertex in $T'_{j}$ is covered by $F'$ (hence $T'_j \cap V(F') = \emptyset$).
    \end{claim}
    \begin{claimproof}
     Since the $s$ sets $X'_{i,1}, X'_{i,2}, \ldots, X'_{i,s}$, with $i \in [d-1]$, resp.\ the $s$ sets $T'_1, T'_2, \ldots, T'_s$ of each selection gadget form a complete $s$-partite graph, it holds that $V(F)$ contains at least $s-1$ sets of the $s$ sets $X'_{i,1}, X'_{i,2}, \ldots, X'_{i,s}$, with $i \in [d-1]$, resp.\ at least $s-1$ of the $s$ sets $T'_1, T'_2, \ldots, T'_s$. 
     First, we show that not all sets $X'_{i,1}, X'_{i,2}, \ldots, X'_{i,s}$, with $i \in [d-1]$, resp.\ not all sets $T'_1, T'_2, \ldots, T'_s$ can be covered by $F$:
     
     Assume that all vertices in $X'_{i,1}\cup X'_{i,2}\cup \ldots \cup X'_{i,s}$, for some $i \in [d-1]$, resp.\ all vertices in $T'_1 \cup T'_2 \cup \ldots \cup T'_s$ are contained in $V(F)$. This would imply that at least 
     \[
     \frac{1}{2} \cdot \left|\bigcup_{j=1}^s X'_{i,j}\right| = \frac{1}{2} \cdot \left|\bigcup_{j=1}^s T'_{j}\right| =\frac{1}{2} \cdot s \cdot 2 \cdot \binom{k}{2} \cdot d = 2 \cdot \binom{k}{2} \cdot d \cdot \frac{s-1}{2} + \binom{k}{2} \cdot d
     \] 
     edges of $F$ must be contained in this selection gadget. These are at least $\binom{k}{2} \cdot d$ edges more than the minimum number of edges in $F$ that must be contained in a selection gadget. But, we showed above that $F$ has at most $\binom{k}{2} \cdot \cost(B) < \binom{k}{2} \cdot |B| = \binom{k}{2} \cdot d$ additional edges. Thus, if $V(F)$ contains the entire set $X'_{i,1}\cup X'_{i,2}\cup \ldots \cup X'_{i,s}$, for some $i \in [d-1]$, resp.\ the entire set $T'_1 \cup T'_2 \cup \ldots \cup T'_s$, then $F$ contains more than $k'$ edges which is a contradiction.
     
     Thus, for each $i \in [d-1]$ there exists an $h_i \in [s]$ such that not all vertices in $X'_{i,h_i}$ are covered by $F$ and there exists an $h_d \in [s]$ such that not all vertices in $T'_{h_d}$ are covered by $F$. Let $h=(h_1,h_2,\ldots,h_d)$.
     Since all vertices in $X'_{i,h_i}$, with $i \in [d-1]$, resp.\ all vertices in $T'_{h_d}$ have the same neighborhood and at least one vertex in these sets in not contained in $V(F)$, it holds that all vertices in the neighborhood of $X'_{i,h_i}$, with $i \in [d-1]$, resp.\ in the neighborhood of $T'_{h_d}$ must be contained in $V(F)$; otherwise $F$ would not be an edge dominating set in $G'$.
     
     We replace every edge in $F$ that is incident with a vertex in $X'_{i,h_i}$, with $i \in [d-1]$, resp.\ to a vertex in $T'_{h_d}$. (Note, that the set $\bigcup_{i=1}^{d-1} X'_{i,h_i} \cup T'_{h_d}$ is an independent set, thus every edge in $F$ that is incident with this set must have its other endpoint outside this set.) Let $f=\{u,v\} \in F$ be an edge in $F$ with $u \in \bigcup_{i=1}^{d-1} X'_{i,h_i} \cup T'_{h_d}$. By construction, the vertex $v$ is a vertex in $X'_{i,j}$, with $i \in [d-1]$, $j \in [s]$ and $j \neq h_i$, or $T_j'$, with $j \in [s]$ and $j \neq h_d$, or $X_{i,h_i}$, with $i \in [d-1]$, or $T_{h_d}$. Hence, $v$ has a neighbor $v'$ that is not contained in $\bigcup_{i=1}^{d-1} X'_{i,h_i} \cup T'_{h_d}$.
     Thus, we can replace edge $f$ in $F$ with edge $f'=\{v,v'\}$ to obtain $F'$. The set $F'$ is still an edge dominating set: the only vertices that are not covered by the set $F'$ any more are contained in $\bigcup_{i=1}^{d-1} X'_{i,h_i} \cup T'_{h_d}$, but this set is an independent set and the neighborhood of this set is still covered by $F'$. This proves the claim.
    \end{claimproof}
    Assume that the edge dominating set $F$ fulfills the properties of Claim \ref{claim::selection_gadget} (if this is not the case we can replace $F$ by $F'$). Let $h^*=(h^*_1,h^*_2,\ldots, h^*_d) \in [s]^d$ such that no vertex in $X'_{i,h^*_i}$, for $i \in [d-1]$, is covered by $F$ and no vertex in $T'_{h^*_d}$ is covered by $F$. It follows, that the sets $X_{i,h^*_i}$, for $i \in [d-1]$, must be covered by $F$ because all vertices in $X'_{i,h^*_i}$ are adjacent to all vertices in $X_{i,h^*_i}$.
    The vertex sets $X_{i,h^*_i}$, with $i \in [d-1]$, are only adjacent to the sets $X'_{i,h^*_i}$ and copies of $H$; hence the edges of $F$ that cover $X_{i,h^*_i}$ have their other endpoint in a copy of $H$. 

    Let $F_{h^*} = \{ f \in F \mid \exists i \in [d-1] \colon f \cap X_{i,h^*_i} \neq \emptyset \}$ be the set of edges in $F$ that are incident with a vertex in $X_{i,h^*_i}$, with $i \in [d-1]$, and let $F_h^e = \{ f \in F \mid f \cap V(H_h^e) \neq \emptyset \}$ be the set of edges in $F$ that are incident with a vertex in $H_h^e$ with $h \in [s]^d$. 
    Let  $B_h^e(F_{h^*}) = \{ b \in B_h^e \mid \exists f \in F_{h^*} \colon b \in f \}$ be the set of vertices in $B_h^e$ that are incident with an edge in $F_{h^*}$.
    It holds that $F_h^e$, with $h \in [s]^d$ and $e \in E(G_h)$, has at least the size of a minimum edge dominating set in $H_h^e - B_h^e(F_{h^*})$ plus the size of $B_h^e(F_{h^*})$, because the edges in $F_h^e$ that have one endpoint in $B_h^e(F_{h^*})$ have their other endpoint not in $H_h^e$ and to dominate all remaining edges in $H_h^e$ we need at least $\MEDS(H_h^e - B_h^e(F_{h^*}))$ many edges. Since no two copies of $H$ are adjacent, this implies that at least 
    \[
    \sum_{h \in [s]^d} |E(G_h)| \cdot \MEDS(H) + \sum_{h \in [s]^d} \sum_{ e \in E(G_h)}  \cost(B_h^e(F_{h^*}))
    \]
    edges of $F$ are incident with a copy of $H$ because
     \begin{align} 
      \sum_{h \in [s]^d} \sum_{ e \in E(G_h)} |F_h^e| 
      &\geq \sum_{h \in [s]^d} \sum_{ e \in E(G_h)} \left( \MEDS(H-B_h^e(F_{h^*})) + |B_h^e(F_{h^*})| \right)  \nonumber \\
      &= \sum_{h \in [s]^d} \sum_{ e \in E(G_h)} \left( \MEDS(H) + \cost(B_h^e(F_{h^*})) \right) \nonumber \\
      &=\sum_{h \in [s]^d} |E(G_h)| \cdot \MEDS(H) + \sum_{h \in [s]^d} \sum_{ e \in E(G_h)}  \cost(B_h^e(F_{h^*}))\text{.} \label{align::bound_F_h^e}
    \end{align}
    Now, we have $\sum_{h \in [s]^d} \sum_{ e \in E(G_h)}  \cost(B_h^e(F_{h^*}))$ edges more in $F$ that are incident with a copy of $H$ than the lower bound of $\sum_{h \in [s]^d} |E(G_h)| \cdot \MEDS(H)$ edges. These edges belong neither to the $2 \cdot \binom{k}{2} \cdot d \cdot \frac{s-1}{2}$ edges that we need to dominate all edges in the complete $s$-partite graph that is a subgraph of every selection gadget nor to the $k \cdot n \cdot (k-1)$ edges that we need to dominate one clique gadget. Thus, 
    \[
     \sum_{h \in [s]^d} \sum_{ e \in E(G_h)}  \cost(B_h^e(F_{h^*})) \leq \binom{k}{2} \cdot \cost(B).
    \]
    \begin{claim} \label{claim::extra_budget} 
      \[
      \sum_{h \in [s]^d} \sum_{ e \in E(G_h)}  \cost(B_h^e(F_{h^*})) = \binom{k}{2} \cdot \cost(B).
      \]
    \end{claim}
    \begin{claimproof}
      First, we rewrite the left-hand side as follows:
      \[
      \sum_{h \in [s]^d} \sum_{ e \in E(G_h)}  \cost(B_h^e(F_{h^*})) = \sum_{1 \leq p<q \leq k} \sum_{h \in [s]^d} \sum_{e \in E(G_h) \cap E(V_p,V_q)}  \cost(B_h^e(F_{h^*})).
      \]
      We assume for contradiction that 
      \[
      \sum_{1 \leq p<q \leq k} \sum_{h \in [s]^d} \sum_{e \in E(G_h) \cap E(V_p,V_q)}  \cost(B_h^e(F_{h^*})) < \binom{k}{2} \cdot \cost(B).
      \]
      This implies that there exist $1 \leq \bar{p} < \bar{q} \leq k$ such that 
      \[
       \sum_{h \in [s]^d} \sum_{e \in E(G_h) \cap E(V_{\bar{p}},V_{\bar{q}})}  \cost(B_h^e(F_{h^*})) < \cost(B).
      \]
      We will show that $B':=\bigcup_{h \in [s]^d} \bigcup_{e \in E(G_h) \cap E(V_{\bar{p}},V_{\bar{q}})} B_h^e(F_{h^*})$ contains at least one copy of every vertex in $B \setminus \{b_d\}$. Let $i \in [d-1]$. Consider vertex $x_{i,h^*_i}^{\bar{p},\bar{q}}$ in $X_{i,h^*_i}$, which is covered by an edge in $F_{h^*}$ (by definition of $F_{h^*}$) and let $f=\{x_{i,h^*_i}^{\bar{p},\bar{q}}, v\}$ be an edge in $F_{h^*}$ that has $x_{i,h^*_i}^{\bar{p},\bar{q}}$ as one endpoint. The vertex $v$ must be a vertex in a copy of $H$ because $X_{i,h^*_i}$ is only adjacent to vertices in $X'_{i,h^*_i}$ (which are not covered by $F$), and adjacent to copies of $H$. More precisely, since $f$ is an edge in $E(G')$, the vertex $v$ must be contained in $\{b_{h,i}^e \mid h \in [s]^d, h_i=h^*_i, e \in E(G_h) \cap E(V_{\bar{p}},V_{\bar{q}}) \}$ (construction of $G'$: we add edges $\{x_{i,h_i}^e,b_{h,i}^e\}$ to $G'$ with $h=(h_1,h_2,\ldots,h_d) \in [s]^d$ and $e \in E(G_h)$). But, every vertex in this set is a copy of $b_i$ and it follows that $B'$ contains at least one copy of $b_i$. Since, this holds for all $i \in [d-1]$ it follows that $B'$ contains at least one copy of each vertex in $B \setminus \{b_d\}$.
      
      Let $B_1,B_2,\ldots,B_l$ be the subsets of $B$ that correspond to the nonempty sets in $\{ B_h^e(F_{h^*}) \mid h \in [s]^d, e \in E(G_h) \cap E(V_{\bar{p}},V_{\bar{q}})\}$. Now, the sets $B_1,B_2,\ldots,B_l$ together with the set $\{b_d\}$ cover the set $B$. Since the vertex $b_d$ is not extendable in $H$ (Proposition \ref{proposition::properties} (\ref{proposition::disjointQ})) it holds that $\cost(\{b_d\})=1$. Thus, it holds that $\sum_{i=1}^l \cost(B_i) + \cost(\{b_d\}) < \cost(B) + 1$, because we assumed that $\sum_{i=1}^l \cost(B_i) < \cost(B)$ and $\cost(\{b_d\})=1$; hence $\sum_{i=1}^l \cost(B_i) + \cost(\{b_d\}) \leq \cost(B)$. Note that every set $B_i$, with $i \in [l]$, must be a proper subset of $B$; otherwise $\cost(B)=\cost(B_i)$ which contradicts the assumption that $\sum_{i=1}^l \cost(B_i) < \cost(B)$.
      Summarized, the sets $B_1,B_2,\ldots, B_l, \{b_d\} \subsetneq B$ cover $B$ and it holds that $\sum_{i=1}^l \cost(B_i) + \cost(\{b_d\}) \leq \cost(B)$. This implies that $B$ is not strongly beneficial (see definition), which is a contradiction and proves the claim.
    \end{claimproof}
    So far, we know that $F$ contains $d \cdot 2 \cdot \binom{k}{2} \cdot d \cdot \frac{s-1}{2}$ edges that cover the $d$ different complete $s$-partite graphs that are subgraphs of the $d$ selection gadgets, and that (at least) $\sum_{h \in [s]^d} |E(G_h)| \cdot \MEDS(H) + \binom{k}{2} \cost(B)$ edges are incident with copies of $H$ (Claim \ref{claim::extra_budget}).
    Thus, the remaining $s \cdot k \cdot n \cdot (k-1)$ edges must cover the $s$ clique gadgets. Since this is the number of edges we need (at least) to dominate the edges between the vertex sets $Z_j$ and $Z'_j$, with $j \in [s]$, in a clique gadget, every remaining edge must either be incident with a vertex of $Z_j$ or with a vertex of $Z'_j$.
    
    Consider the clique gadget that contains the vertex set $Y_{h^*_d}$. Since no vertex in $T'_{h^*_d}$ is covered by $F$, it holds that every vertex in $T_{h^*_d}$ must be covered by $F$ (because every vertex in $T'_{h^*_d}$ is adjacent to every vertex in $T_{h^*_d}$). Each vertex in $T_{h^*_d}$ is only adjacent to vertices in $T'_{h^*_d}$ and $Z'_{h^*_d}$; thus, every edge in $F$ that is incident with a vertex in $T_{h^*_d}$ has its other endpoint in $Z'_{h^*_d}$.
    Furthermore, for each vertex $v \in V$ the entire set $Z_{h^*_d,v}$ or the entire set $Z'_{h^*_d,v}$ is contained in $V(F)$ (both is also okay) because $G'[Z_{h^*_d,v} \cup Z'_{h^*_d,v}]$ is a complete bipartite graph whose partition has the parts $Z_{h^*_d,v}$ and $Z'_{h^*_d,v}$. Additionally, this implies that the vertex set $Z_{h^*_d,v} \cup Z'_{h^*_d,v}$ is incident with exactly $|Z_{h^*_d,v}| = |Z'_{h^*_d,v}|=k-1$ edges of $F$ because $|V|=k \cdot n$, and there are $k \cdot n \cdot (k-1)$ edges in $F$ that dominate all edges between $Z_{h^*_d}$ and  $Z'_{h^*_d}$.
    
    \begin{claim}\label{claim::clique_gadget}
     There exists an edge dominating set $F'$ of size $k'$ in $G'$ such that for all vertices $v \in V$ either no vertex in $Z'_{h^*_d,v}$ or no vertex in $Z_{h^*_d,v}$ is covered by $F'$. Furthermore, for each color class $V_i$, with $i \in [k]$, there exists exactly one vertex $v_i \in V_i$ such that $F'$ covers no vertex in $Z_{h^*_d,v}$.
    \end{claim}
    \begin{claimproof}
     Let $i \in [k]$ and let $v_i$ be a vertex in $V_i$ such that $F$ contains an edge that has one endpoint in $T_{h^*_d,i}$ and the other endpoint in $Z'_{h^*_d,v_i}$.
     Since the set $Z'_{h^*_d,v_i} \cup Z_{h^*_d,v_i}$ is only incident with $k-1$ edges of $F$, and since either $Z'_{h^*_d,v_i}$ or $Z_{h^*_d,v_i}$ must be entirely contained in $V(F)$ it holds that $Z'_{h^*_d,v_i}$ is entirely covered by $F$, and that every edge of $F$ that is incident with a vertex in $Z_{h^*_d,v_i}$ has its other endpoint in $Z'_{h^*_d,v_i}$. Furthermore, there exists a vertex in $Z_{h^*_d,v_i}$ that is not covered by $F$ because the set $Z'_{h^*_d,v_i} \cup Z_{h^*_d,v_i}$ is only incident with $k-1$ edges of $F$ and at least one of these $k-1$ edge has no endpoint in $Z_{h^*_d,v_i}$.
     Thus, the entire neighborhood of $Z_{h^*_d,v_i}$ must be covered by $F$ because all vertices in $Z_{h^*_d,v_i}$ have the same neighborhood.
     Hence, we can delete every edge in $F$ that has one endpoint in $Z'_{h^*_d,v_i}$ and add a maximum matching of the complete bipartite graph $G'[Z'_{h^*_d,v_i} \cup T_{h^*_d,i}]$ to $F$. This maximum matching has size $|Z'_{h^*_d,v_i}|=|T_{h^*_d,i}|$ and covers all vertices in $Z'_{h^*_d,v_i} \cup T_{h^*_d,i}$. The resulting edge set, which we denote by $\widetilde{F}$ is still an edge dominating set of size $|F|$ in $G'$ because the only vertices that are not covered anymore are contained in $Z_{h^*_d,v_i}$, but all neighbors are still covered: The vertices in $Z'_{h^*_d,v_i}$ are only adjacent to the vertices in $Z_{h^*_d,v_i}$ and $T_{h^*_d,i}$, and the vertices in $Z_{h^*_d,v_i}$ are adjacent to vertices in $Z'_{h^*_d,v_i}$ and $Y_{h^*_d}$. Thus, the only edges that we replace and that are incident with $Z_{h^*_d,v_i}$ or a neighbor of $Z_{h^*_d,v_i}$ are incident with $Z'_{h^*_d,v_i}$ and these vertices are still covered by $\widetilde{F}$. It holds that the edge dominating set $\widetilde{F}$ covers all vertices in $Z'_{h^*_d,v_i}$ and no vertex in $Z_{h^*_d,v_i}$.
     
     Now, consider a vertex $v \in V_i \setminus \{v_i\}$. Every vertex in $Z'_{h^*_d,v}$ is only adjacent to the vertices in $T_{h^*_d,i} \cup Z_{h^*_d,v}$. Since every vertex in $T_{h^*_d,i}$ is covered by $\widetilde{F}$, we can replace the edges in $\widetilde{F}$ that are incident with a vertex in $Z_{h^*_d,v}$.
     Recall that that either the set vertex $Z'_{h^*_d,v}$ or the vertex set $Z_{h^*_d,v}$ is entirely covered by $\widetilde{F}$, and that $Z'_{h^*_d,v} \cup Z_{h^*_d,v}$ is only incident with $|Z'_{h^*_d,v}|=|Z_{h^*_d,v}|=k-1$ edges of $\widetilde{F}$.
     If $\widetilde{F}$ covers all vertices in $Z'_{h^*_d,v}$ then we replace every edge $e=\{z,z'\} \in F$ with $z' \in Z'_{h^*_d,v}$ and $z \in Z_{h^*_d,v}$ by an edge in $E(z,Y_{h^*_d})$. 
     Otherwise, if $\widetilde{F}$ covers all vertices in $Z_{h^*_d,v}$ then we delete the $|Z'_{h^*_d,v}|=|Z_{h^*_d,v}|=k-1$ edges in $F$ that are incident with a vertex in $Z'_{h^*_d,v} \cup Z_{h^*_d,v}$ and add for each vertex $z$ in $Z_{h^*_d,v}$ exactly one edge in $E(z,Y_{h^*_d})$ to $\widetilde{F}$. We denote the resulting set by $F'$. 
     Note that in both cases the set $Z_{h^*_d,v}$ is not covered by $F'$ and the set $Z'_{h^*_d,v}$ is covered by $F'$.
     Clearly, $F'$ has the same size as $\widetilde{F}$ (by construction) and hence as $F$. The fact that $F'$ is still an edge dominating set holds because the only vertices that are possibly covered by $\widetilde{F}$ and not by $F'$ are contained in $Z'_{h^*_d,v}$, but all neighbors of $Z'_{h^*_d,v}$ are still contained in $V(F')$.
     We can do this for all $i \in [k]$ and all $v \in V_i \setminus \{v_i\}$ independently; this proves the claim.
    \end{claimproof}
    Let $F$ be the edge dominating set that we construct during the proof of Claim \ref{claim::clique_gadget}. 
    Now, for each $i \in [k]$ there exists exactly one vertex $v$ in $V_i$ such that no vertex in $Z_{h^*_d,v}$ is incident with an edge in $F$; denote this vertex by $x_i$. 
    Let $X=\{x_1,x_2,\ldots,x_k\}$. We will show that $X$ is a clique in $G_{h^*}$.
    
    Every vertex in the set $Z_{h^*_d,x_i}$, with $i \in [k]$, is adjacent to all vertices in $Z'_{h^*_d,x_i}$ and to all vertices in $\{ y_{h^*_d}^{\{x_i,v\}} \in Y_{h^*_d} \mid v \in V \setminus V_i \}$; thus both sets must be covered by $F$. The second set contains $(k-1) \cdot n$ vertices, one vertex for every vertex in $V \setminus V_i$. All vertices $y_{h_d}^{\{x_i,x_j\}}$, with $ j \in [k]$ and $j \neq i$, can only be covered by edges in $F$ that have one endpoint in a copy of $H$ because these vertices are only adjacent to copies of $H$, and adjacent to the set $Z_{h^*_d,x_i} \cup Z_{h^*_d,x_j}$, which is not incident with an edge in $F$.
    Let $Y=\{y_{h^*_d}^{\{x_i,x_j\}} \in Y_{h^*_d} \mid 1\leq i < j\leq k\}$ be the set of vertices in $Y_{h^*_d}$ that must be covered by $F$ via edges that have one endpoint in a copy of $H$. This set has size $\binom{k}{2}$.
    
    Recall that $F_h^e$ is the set of edges in $F$ that are incident with a vertex in $H_h^e$. The sets $F_h^e$ are pairwise disjoint because no two copies of $H$ are adjacent. Let $B_h^e(F) = \{ b \in B_h^e \mid \exists f \in F \colon b \in f \text{ and } f \nsubseteq V(H_h^e)\}$ be the set of vertices in $B_h^e$ that are an endpoint of an edge $f$ in $F$ whose other endpoint is not a vertex in a copy of $H$; thus, by construction, this other endpoint is contained in a set $X_{i,j}$ or a set $Y_j$, with $i \in [d-1]$ and $j \in [s]$. Note that $B_h^e(F_{h^*}) \subseteq B_h^e(F)$ because $B_h^e(F_{h^*})$ contains all edges that have one endpoint in $H_h^e$ (or more precisely $B_h^e$) and the other endpoint in a set $X_{i,h^*_i}$ ,with $i \in [d-1]$, whereas, $B_h^e(F_h)$ contains all edges that have one endpoint in $H_h^e$ (or more precisely $B_h^e$) and the other endpoint in $X'$.
    As before, the edge set $F_h^e$, with $h \in [s]^d$ and $e \in E(G_h)$, has at least the size of a minimum edge dominating set in $H_h^e - B_h^e(F)$ plus the size of $B_h^e(F)$ because $F_h^e$ contains the $|B_h^e(F)|$ edges between the vertices in $B_h^e(F)$ and a vertex that is not in $V(H_h^e)$, and because $F_h^e$ must also dominate all remaining edges in $H_h^e - B_h^e(F)$; thus,
    \begin{align}
     \sum_{h \in [s]^d} \sum_{ e \in E(G_h)} |F_h^e| &\geq \sum_{h \in [s]^d} \sum_{ e \in E(G_h)} \left( \MEDS(H_h^e-B_h^e(F)) + |B_h^e(F)| \right) \nonumber \\
     &= \sum_{h \in [s]^d} \sum_{ e \in E(G_h)} \left( \MEDS(H) + \cost(B_h^e(F)) \right). \label{align::F_h^e}
    \end{align}
    Since every edge in a set $F_h^e$, for $h \in [s]^d$ and $e \in E(G_h)$, is incident with a vertex in $H_h^e$ (and therefore neither incident with a vertex in $Z_j \cup Z_j'$, for $j \in [s]$, nor incident with a vertex in a selection gadget) we know that
    \[
     \sum_{h \in [s]^d} \sum_{ e \in E(G_h)} |F_h^e| \leq \sum_{h \in [s]^d} |E(G_h)| \cdot \MEDS(H) + \binom{k}{2} \cdot \cost(B).
    \]
    It follows from Claim \ref{claim::extra_budget} together with inequality (\ref{align::bound_F_h^e}) that 
    \[ 
      \sum_{h \in [s]^d} \sum_{ e \in E(G_h)} |F_h^e| \geq \sum_{h \in [s]^d} |E(G_h)| \cdot \MEDS(H) + \binom{k}{2} \cdot \cost(B).
    \]
    Combining the two last inequalities we obtain that
    \begin{align*}
      \sum_{h \in [s]^d} \sum_{ e \in E(G_h)} |F_h^e| = \sum_{h \in [s]^d} |E(G_h)| \cdot \MEDS(H) + \binom{k}{2} \cdot \cost(B). \label{align::equality}
    \end{align*}
   Summarized, this implies:
   \begin{align*}
      &\sum_{h \in [s]^d} |E(G_h)| \cdot \MEDS(H) + \binom{k}{2} \cdot \cost(B) \\
      &\qquad\overset{(\ref{align::equality})}{=}\sum_{h \in [s]^d} \sum_{ e \in E(G_h)} |F_h^e| \\
      &\qquad \overset{(\ref{align::F_h^e})}{\geq} \sum_{h \in [s]^d} \sum_{ e \in E(G_h)} \left( \MEDS(H) + \cost(B_h^e(F)) \right) \\
      &\qquad=\sum_{h \in [s]^d} |E(G_h)| \cdot \MEDS(H) + \sum_{h \in [s]^d} \sum_{ e \in E(G_h)}  \cost(B_h^e(F)) \\
      &\qquad \geq \sum_{h \in [s]^d} |E(G_h)| \cdot \MEDS(H) + \sum_{h \in [s]^d} \sum_{ e \in E(G_h)}  \cost(B_h^e(F_{h^*})) \\
      &\quad \overset{Claim \ref{claim::extra_budget}}{=}\sum_{h \in [s]^d} |E(G_h)| \cdot \MEDS(H) + \binom{k}{2} \cdot \cost(B)
    \end{align*}
    The last inequality holds because $B_h^e(F_{h^*}) \subseteq B_h^e(F)$, which implies that $\cost(B_h^e(F_{h^*})) \leq \cost(B_h^e(F))$ (Proposition \ref{proposition::properties} (\ref{proposition::monoton})). It follows that all terms are equal and, hence, 
    \[
     \sum_{1 \leq p<q \leq k} \sum_{h \in [s]^d} \sum_{e \in E(G_h) \cap E(V_p,V_q)}  \cost(B_h^e(F)) = \sum_{h \in [s]^d} \sum_{ e \in E(G_h)}  \cost(B_h^e(F_{h^*})) = \binom{k}{2} \cdot \cost(B).
    \]
    This implies that either for all $1 \leq p < q \leq k$ it holds that 
    \[
     \sum_{h \in [s]^d} \sum_{e \in E(G_h) \cap E(V_p,V_q)}  \cost(B_h^e(F)) =  \cost(B),
    \]
    or that there exist $1 \leq \bar{p} < \bar{q} \leq k$ such that 
    \[
     \sum_{h \in [s]^d} \sum_{e \in E(G_h) \cap E(V_{\bar{p}},V_{\bar{q}})}  \cost(B_h^e(F)) < \cost(B).
    \]
    We will show that for all $1 \leq p < q \leq k$ it holds that
    \begin{align}\label{align::equality}
     \sum_{h \in [s]^d} \sum_{e \in E(G_h) \cap E(V_p,V_q)}  \cost(B_h^e(F)) = \cost(B),
    \end{align}
    and that $B_h^e(F)=B_h^e$ if and only if $h = h^*$ and $e$ has both endpoints in $X=\{x_1,x_2,\ldots,x_k\}$.
    
    First, as in the proof of Theorem \ref{theorem::LB} we will show that we always have equality, hence that (\ref{align::equality}) holds.
    Let $1 \leq \bar{p} < \bar{q} \leq k$.
    We showed in the proof of Claim \ref{claim::extra_budget} that 
    $B':=\bigcup_{h \in [s]^d} \bigcup_{e \in E(G_h) \cap E(V_{\bar{p}},V_{\bar{q}})} B_h^e(F_{h^*})$ contains at least one copy of every vertex in $B \setminus\{b_d\}$. 
    Thus, $B'':=\bigcup_{h \in [s]^d} \bigcup_{e \in E(G_h) \cap E(V_{\bar{p}},V_{\bar{q}})} B_h^e(F)$ contains at least one copy of every vertex in $B \setminus\{b_d\}$, because $B_h^e(F_{h^*}) \subseteq B_h^e(F)$.
    Furthermore, $B''$ also contains a copy of $b_d$, because vertex $y_{h^*_d}^{\{x_{\bar{p}},x_{\bar{q}}\}}$ must be covered by an edge $f$ in $F$ that has its other endpoint in a copy of $H$. By construction, the vertices in a copy of $H$ that are adjacent to vertex $y_{h^*_d}^{\{x_{\bar{p}},x_{\bar{q}}\}}$ are the vertices $b_{h,d}^{\{x_{\bar{p}},x_{\bar{q}}\}}$ with $h\in [s]^d$ s.t.\ $h_d=h^*_d$, and $\{x_{\bar{p}},x_{\bar{q}}\} \in E(G_{h}) \cap E(V_{\bar{p}},V_{\bar{q}})$; this implies that $B''$ contains a copy of $b_d$.
    
    Now, let $B_1,B_2,\ldots, B_l$ be the subsets of $B$ that correspond to nonempty sets in $\{B_h^e(F) \mid h \in [s]^d, e \in E(G_h) \cap E(V_{\bar{p}},V_{\bar{q}})\}$. Note that every set $B_i$ is either a proper subset of $B$ or $l=1$: If there exists a set, say w.l.o.g. $B_1$, such that $B_1=B$ then $\cost(B_1)=\cost(B)$. Since no vertex in $B$ is extendable (Proposition \ref{proposition::properties} (\ref{proposition::disjointQ})) it holds that $\cost(B_i) \geq 1$ for all $i \in [l]$. Now, if $l > 1$ and $B_1=B$ then $\sum_{i=1}^l \cost(B_i) \geq \cost(B) + (l-1) > \cost(B)$, which contradicts the assumption. Thus, we have either $l=1$ and $B_1=B$ (since $B_1$ covers $B$) or that the sets $B_1,B_2,\ldots,B_l \subsetneq B$ cover $B$. Since $B$ is strongly beneficial, and the sets $B_1,B_2,\ldots,B_l$ cover $B$, it must hold that $\sum_{i=1}^l \cost(B_i) > \cost(B)$ or that $l=1$. Therefore, $\sum_{i=1}^l \cost(B_i) > \cost(B)$ if $l > 1$ and $\sum_{i=1}^l \cost(B_i) = \cost(B)$ if $l=1$, but, this must hold for all $1 \leq p < p \leq k$. Hence, we must always have the latter case where $l=1$, and there exist no $1 \leq \bar{p} < \bar{q} \leq k$ such that $\sum_{h \in [s]^d} \sum_{e \in E(G_h) \cap E(V_{\bar{p}},V_{\bar{q}})}  \cost(B_h^e(F)) < \cost(B)$. 
    Thus, it holds for all $1 \leq p < q \leq k$ that there exists an index $\hat{h} \in [s]^d$ and an edge $\hat{e} \in E(G_{\hat{h}}) \cap E(V_p,V_q)$ such that $B_{\hat{h}}^{\hat{e}}(F)=B_{\hat{h}}^{\hat{e}}$; such a set exists, because $B''$ contains at least one copy of each vertex in $B$ and $l=1$. Furthermore, all other sets $B_h^e(F)$, with $h \in [s]^d$, $e \in E(G_h) \cap E(V_p,V_q)$, and $h \neq \hat{h}$ or $e \neq \hat{e}$, are empty since $l=1$.
    We will show that $\hat{h}=h^*$ and that $\hat{e} = \{x_p,x_q\}$:
    
    Consider vertex $x_{i,h^*_i}^{p,q}$, with $i \in [d-1]$, which is contained in $X_{i,h^*_i}$. This vertex must be covered by an edge $f$ in $F$ and the other endpoint of this edge $f$ must be contained in $\{b_{h,i}^e \mid h \in [s]^d, h_i=h^*_i, e \in E(G_h) \cap E(V_{p},V_{q})\}$ (see proof of Claim \ref{claim::extra_budget}). 
    This vertex must also be contained in $B_{\hat{h}}^{\hat{e}}$ because all other sets $B_h^e(F)$ with $e \in E(V_p,V_q)$ are empty; thus $\hat{h}_i = h^*_i$. This holds for all $i \in [d-1]$ which implies that $(\hat{h}_1,\hat{h}_2,\ldots,\hat{h}_{d-1}) = (h_1,h_2,\ldots,h_{d-1})$.
    Furthermore, vertex $y_{h^*_d}^{\{x_{p},x_{q}\}}$ must be covered by an edge $f$ in $F$ whose other endpoint is contained in $\{b_{h,d}^{\{x_{p},x_{q}\}} \mid h \in [s]^d, h_d=h^*_d, \{x_{p},x_{q}\} \in E(G_h) \cap E(V_{p},V_{q})\}$. For the same reasons, this vertex must be contained in $B_{\hat{h}}^{\hat{e}}$. The only vertex in $\{b_{h,d}^{\{x_{p},x_{q}\}} \mid h \in [s]^d, h_d=h^*_d, \{x_{p},x_{q}\} \in E(G_h) \cap E(V_{p},V_{q})\}$ that is also contained in $B_{\hat{h}}^{\hat{e}}$ is vertex $b_{h^*,d}^{\{x_{p},x_{q}\}}$. Thus, $\hat{h}_d=h^*_d$, $\hat{e}=\{x_{p},x_{q}\}$, and $\{y_{h^*_d}^{\{x_{p},x_{q}\}}, b_{h^*,d}^{\{x_{p},x_{q}\}}\}$ is an edge in $G'$. Summarized, we showed that $B_h^e(F)=B_h^e$ if and only if $h=h^*$ and $e \in E(X,X)$, and that all other sets $B_h^e(F)$ are empty.
    
    We will show that the vertex set $X=\{x_1,x_2,\ldots,x_k\}$ is a clique in $G_{h^*}$. Recall that the vertex $x_i$ is contained in $V_i$, thus every vertex of $X$ is contained in a different color class. Consider two vertices $x_p$, $x_q$ with $1 \leq p < q \leq k$. We have to show that $\{x_p,x_q\}$ is an edge in $E(G_{h^*})$. Since vertex $y_{h^*_d}^{\{x_p,x_q\}}$ is adjacent to vertex $b_{h^*,d}^{\{x_p,x_q\}}$ it holds that $\{x_p,x_q\}$ is an edge in $G_{h^*}$, which proves that $X$ is a clique in $G_{h^*}$.
\end{proof}

Concluding the section, we observe a simple quadratic lower bound for the size of kernels for \EDS parameterized by the size of a modulator to \cH-component graphs that holds for all sets $\cH$.

\begin{lemma}\label{lemma::quadraticlowerbound}
The \EDS problem has no kernelization to size $\Oh(n^{2-\varepsilon})$ where $n$ is the number of vertices, for any $\varepsilon>0$, unless \containment. Therefore, for any set $\cH$ of (connected) graphs, the \EDS problem parameterized by the size of a modulator to $\cH$-component graphs admits no kernelization to size $\Oh(|X|^{2-\varepsilon})$, for any $\varepsilon>0$, unless \containment.
\end{lemma}

\begin{proof}
It is known that \VC admits no kernelization to size $\Oh(n^{2-\varepsilon})$, for any $\varepsilon>0$, unless \containment~\cite{DellM14}. By a straightforward reduction to an EDS instance $(G',k)$ with $n'=\Oh(n)$ vertices the same is true for EDS. This in turn yields an equivalent instance $(G',k,X')$ with a trivial modulator $X'=V(G')$ such that $G'-X'$ is the empty graph; since this is a feasible instance for the \EDS problem parameterized by the size of a modulator to $\cH$-component graphs for all sets $\cH$ and since $|X'|=n'=\Oh(n)$, the lemma follows.

The lower bound of $\Oh(n^{2-\varepsilon})$ for EDS is not surprising (and may well be known), as the same is known for a number of similar graph problems (like \VC). Thus, we only sketch a simple reduction (which surely has been rediscovered several times already).

Let $(G,k)$ be an instance of \VC with $G=(V,E)$ and, w.l.o.g., $k\leq |V|$. Construct a graph $G'$, starting from a copy of $G$ by adding $2k$ vertices $u_1,\ldots,u_k,u'_1,\ldots,u'_k$ and adding the edges $\{u_1,u'_1\},\ldots,\{u_k,u'_k\}$. Finally, make each vertex $u_i$ adjacent to all vertices in the copy of $V$ in $G'$. Return the instance $(G',k)$. Clearly, $G'$ has $n+2k=\Oh(n)$ vertices and the construction can be done in polynomial time.

It is easy to see that $(G,k)$ is yes for \VC if and only if $(G',k)$ is yes for \EDS: If $(G,k)$ is yes for VC then we can pick a vertex cover $S=\{v_1,\ldots,v_k\}$ of size exactly $k$. Clearly, $F=\{\{v_1,u_1\},\ldots,\{v_k,u_k\}\}$ is an edge dominating set of size $k$ for $G'$; all additional edges in comparison to $G$ have an endpoint in $\{u_1,\ldots,u_k\}$. For the converse, assume that $(G',k)$ is yes and let $F$ be an edge dominating set of size $k$. Let $S$ contain all vertices of $V$ in $G'$ that are endpoints of $F$. Observe that, because $F$ needs to contain at least one vertex per edge $\{u_1,u'_1\},\ldots,\{u_k,u'_k\}$, which are disjoint from $V$, and because it has at most $2k$ endpoints, the set $S$ has size at most $k$. Clearly, the set $S$ alone covers all edges of $G'[V]\cong G$, so $(G,k)$ is yes for \VC. This completes the proof.
\end{proof}

\section{EDS parameterized above half a maximum matching} \label{section::aboveLB}

A natural lower bound for the size of a minimum edge dominating set is $\frac{1}{2} MM$, where $MM$ denotes the size of a maximum matching.
We show that \EDS is \NP-hard even for the special case where the input graph has a perfect matching and we need to determine whether there is an edge dominating set of at most half the size of that matching.
This implies that \EDS parameterized by $l=k-\frac{1}{2} MM$, where $k$ is the solution size, is para-\NP-hard. 

\begin{theorem} \label{theorem::MM}
    \EDS parameterized by $k-\frac{1}{2} MM$ is para-\NP-hard.
\end{theorem}

\begin{proof}
    To prove the theorem we show that it is \NP-hard to decide whether a given graph that has a perfect matching has an edge dominating set of size equal to half the size of this matching.
    We will show that this problem is \NP-hard by giving a reduction from  3-\SAT (which is known to be \NP-complete \cite{Cook71}). Let $(X,\mathcal{C})$ be an instance of 3-\SAT with $n$ variables and $m$ clauses. We construct a graph $G$ as follows:
    
    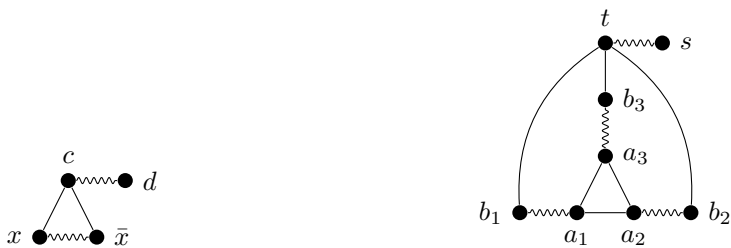
\begin{figure}[t]
    \centering
    \begin{minipage}[t]{0.35\textwidth}
    \centering
    \begin{tikzpicture}[scale=0.75]
        \node[fill, circle, inner sep = 2pt] (x) at (0cm, 0cm) [label=left:$x$] {};
        \node[fill, circle, inner sep = 2pt] (nx) at (1cm,0cm) [label=right:$\bar{x}$] {};
        \node[fill, circle, inner sep = 2pt] (c) at (0.5cm,1cm) [label=above:$c$] {};
        \node[fill, circle, inner sep = 2pt] (d) at (1.5cm,1cm) [label=right:$d$] {};
        \draw[decorate,decoration={snake,amplitude=.4mm,segment length=1mm}] (x) -- (nx);
        \draw (x) -- (c);
        \draw (c) -- (nx);
        \draw[decorate,decoration={snake,amplitude=.4mm,segment length=1mm}] (c) -- (d);
    \end{tikzpicture}
         \subcaption{Gadget for a variable $x \in X$}
         \label{figure::variable_gadget}
    \end{minipage}
    \hspace{3em}
    \begin{minipage}[t]{0.45\textwidth}
    \centering
    \begin{tikzpicture}[scale=0.75]
        \node[fill, circle, inner sep = 2pt] (a1) at (0cm, 0cm) [label=below:$a_1$] {};
        \node[fill, circle, inner sep = 2pt] (a2) at (1cm,0cm) [label=below:$a_2$] {};
        \node[fill, circle, inner sep = 2pt] (a3) at (0.5cm,1cm) [label=right:$a_3$] {};
        \node[fill, circle, inner sep = 2pt] (b1) at (-1cm,0cm) [label=left:$b_1$] {};
        \node[fill, circle, inner sep = 2pt] (b2) at (2cm,0cm) [label=right:$b_2$] {};
        \node[fill, circle, inner sep = 2pt] (b3) at (0.5cm,2cm) [label=right:$b_3$] {};
        \node[fill, circle, inner sep = 2pt] (t) at (0.5cm,3cm) [label=above:$t$] {};
        \node[fill, circle, inner sep = 2pt] (s) at (1.5cm,3cm) [label=right:$s$] {};
        \draw (a1) -- (a2);
        \draw (a1) -- (a3);
        \draw (a2) -- (a3);
        \draw[decorate,decoration={snake,amplitude=.4mm,segment length=1mm}] (a1) -- (b1);
        \draw[decorate,decoration={snake,amplitude=.4mm,segment length=1mm}] (a2) -- (b2);
        \draw[decorate,decoration={snake,amplitude=.4mm,segment length=1mm}] (a3) -- (b3);
        \draw (b1) to [bend left] (t);
        \draw (b2) to [bend right] (t);
        \draw (b3) -- (t);
        \draw[decorate,decoration={snake,amplitude=.4mm,segment length=1mm}] (s) -- (t);
    \end{tikzpicture}
         \subcaption{Gadget for a clause $C=\{ \lambda_1,\lambda_2,\lambda_3\} \in \mathcal{C}$}
         \label{figure::clause_gadget}
    \end{minipage}
    \caption{The wavy edges are the edges that are contained in the perfect matching $M$.}
    \label{figure::gadgets}
    \end{figure}
    
    For each variable $x \in X$ we construct a variable gadget (see Figure \ref{figure::variable_gadget}) consisting of four vertices $x$, $\bar{x}$, $c$, $d$, where the vertices $x$, $\bar{x}$, $c$ form a clique and the vertex $d$ is only adjacent to vertex $c$. (Here $x$ and $\bar{x}$ are the literals of variable $x$.) For every clause $C=\{ \lambda_1,\lambda_2,\lambda_3\} \in \mathcal{C}$ we construct a clause gadget (see Figure \ref{figure::clause_gadget}) consisting of eight vertices $a_1$, $a_2$, $a_3$, $b_1$, $b_2$, $b_3$, $s$, and $t$. The vertices $a_1$, $a_2$, and $a_3$ form a clique, each vertex $a_i$, with $i=1,2,3$, is also adjacent to $b_i$, and vertex $t$ is adjacent to $b_1$, $b_2$, $b_3$, and $s$. For every clause $C=\{ \lambda_1,\lambda_2,\lambda_3\} \in \mathcal{C}$ we make vertex $\lambda_i$ (which is contained in a variable gadget) for $i \in \{1,2,3\}$ adjacent to vertex $a_i$ in the clause gadget $C$.
    
    It is easy to verify that $G$ has a perfect matching, e.g. the edges $\{x, \bar{x}\}$, $\{c,d\}$ in every variable gadget together with the edges $\{a_i, b_i\}$ for $i=1,2,3$, and the edge $\{s,t\}$ in every clause gadget are a perfect matching in $G$. (In Figure \ref{figure::gadgets} the matching edges are the wavy edges.) We denote this maximum matching by $M$. Note that the matching $M$ has size $2n+4m$; two edges in every variable gadget and four edges in every clause gadget.
    
    We will show that $G$ has an edge dominating set of size $n+2m$ if and only if the 3-\SAT instance has a satisfying assignment. There cannot be an edge dominating set of smaller size, because every edge dominating set has at least half the size of any (maximum) matching.

    Suppose first that the 3-\SAT instance $(X,\mathcal{C})$ has a satisfying assignment $s \colon X \rightarrow \{\mathrm{true}, \mathrm{false}\}$. We construct an edge dominating set $F$ of $G$ by selecting edge $\{x,c\}$ in the variable gadget for $x \in X$ to $F$ if $x$ is set to true and by selecting edge $\{\bar{x},c\}$ in the variable gadget $x \in X$ to $F$ if $x$ is set to false. For every clause $C=\{\lambda_1, \lambda_2, \lambda_3\}$ we choose one true literal, w.l.o.g.\ say $\lambda_1$ is true and we add the edges $\{t,b_1\}$ and $\{a_2,a_3\}$ to $F$. The set $F$ has exactly $n+2m$. We have to show that $F$ is an edge dominating set.
    
    Assume for contradiction that there exists an edge $e$ in $G$ that is not dominated by $F$. By construction of $F$, this edge cannot be in a clause or a variable gadget. Hence, this edge must have one endpoint in a clause gadget and one endpoint in a variable gadget. Let $C=\{\lambda_1, \lambda_2,\lambda_2\}$ be the clause that corresponds to the clause gadget that contains one endpoint of $e$. By construction of $G$ and $F$, the endpoint of $e$ is exactly the vertex $a_i$ that is not contained in $V(F)$. This implies that literal $\lambda_i$ is true. Since $\lambda_i$ is the only neighbor of $a_i$ outside the clause gadget, $\lambda_i$ is the other endpoint of $e$. But $\lambda_i$ is contained in $V(F)$ (by construction), hence $e$ is dominated.
    
    Now, suppose $G$ has an edge dominating set $F$ of size $n+2m$. Since the matching $M$ has twice the size of the edge dominating set $F$, it must hold that every edge in $F$ dominates two matching edges and different edges in $F$ dominate different matching edges; otherwise there would be an edge that is not dominated by $F$. The matching edge $\{c,d\}$ in a variable gadget for variable $x \in X$ has as neighbors only the vertices $x$ and $\bar{x}$, therefore either edge $\{c,x\}$ or edge  $\{c,\bar{x}\}$ is contained in the edge dominating set $F$.
    
    To satisfy the instance $(X,\mathcal{C})$ of 3-\SAT let $x$ be set to $\mathrm{true}$ if $\{x,c\} \in F$ and $x$ be set to $\mathrm{false}$ if $\{\bar{x},c\} \in F$. Exactly one of these edges is contained in $F$ (see above). To show that this is a satisfying assignment consider an arbitrary clause $C=\{\lambda_1, \lambda_2, \lambda_3\} \in \mathcal{C}$ and the clause gadget for clause $C$. 
    
    The matching edge $\{s,t\}$ in the clause gadget for clause $C$ has as neighbors only the vertices $b_1$, $b_2$, and $b_3$, hence exactly one of the edges $\{t,b_1\}$, $\{t,b_2\}$, $\{t,b_3\}$ is contained in $F$. Assume w.l.o.g.\ that $\{t,b_1\} \in F$, which dominates the matching edge $\{a_1,b_1\}$. Thus, no other edge of $F$ can also dominate this edge $\{a_1,b_1\}$ and, hence, $a_1\notin V(F)$. Now, however, the edge $\{a_1,\lambda_1\}$ is only dominated by $F$ if $\lambda_1\in V(F)$, which holds only if $\{c,\lambda_1\}$ is contained in $F$. By construction of our assignment in the previous paragraph, this implies that $\lambda_1$ is $\mathrm{true}$ and that clause $C$ is satisfied.
    This completes the proof.
\end{proof}

The graph we construct in the proof of Theorem \ref{theorem::MM} is also a K\H{o}nig graph, i.e., it has minimum vertex cover size equal to maximum matching size. This implies that \EDS for K\H{o}nig graphs is also \NP-hard (even if $k=\frac12MM$).

\section{Proof of Proposition~\ref{proposition::properties}} \label{section::proposition}

Let $H=(V,E)$ be a connected graph, let $W=W(H)$ be the set of free vertices, let $Q=Q(H)$ be the set of extendable vertices, and let $U=U(H)$ be the set of uncovered vertices. 

\begin{lemma}[Proposition \ref{proposition::properties} (\ref{proposition::free})]
 The set $W$ is well defined 
\end{lemma}
\begin{proof}
  To show that the maximum free set is unique, we show that the union of two free sets is free. This implies that $W$ is the union of all free sets in $H$. Let $Y_1, Y_2 \subseteq Q$ be free sets in $H$, and let $Y= Y_1 \cup Y_2$. Since $Y_1$ and $Y_2$ are free, it hold that for all $y \in Y$ and for all minimum edge dominating set $F$ in $H$ there exists a minimum edge dominating set $F'$ in $H-y$ (of size $|F|-1$) with either $V(F) \subseteq V(F') \setminus Y_1 \subseteq V(F') \setminus Y$ (if $y \in Y_1$) or $V(F) \subseteq V(F') \setminus Y_2 \subseteq V(F') \setminus Y$ (if $y \in Y_2$); thus $Y$ is free.
\end{proof}

\begin{lemma}[Proposition \ref{proposition::properties} (\ref{proposition::neighborhood})]
 The set $U$ is an independent set and no vertex in $Q$ is adjacent to a vertex in $U$; hence $N_H(U) \cap (Q \cap U) = \emptyset$.
\end{lemma}
\begin{proof}
 If there was an edge $\{u,u'\}$ with $u,u'\in U$ then no feasible edge dominating set could avoid being incident with either vertex. If $u\in U$ had a neighbor $q\in Q$ then $\MEDS(H-q)+1=\MEDS(H)$ but then combining a minimum solution for $H-q$ and adding edge $\{u,q\}$ would be a minimum solution for $H$ and be incident with $u$; a contradiction.
\end{proof}

\begin{lemma}[Proposition \ref{proposition::properties} (\ref{proposition::N(U)})]
 If $v \in N_H(U)$ is a vertex that is adjacent to a vertex in $U$, then $v$ is contained in every minimum edge dominating set of $H$
\end{lemma}
\begin{proof}
 Vertices in $U$ are never endpoints of edges in minimum solutions. Thus, to dominate the incident edges all their neighbors must be endpoints of solution edges.
\end{proof}

\begin{lemma}[Proposition \ref{proposition::properties} (\ref{proposition::vertex})]
 It holds for all vertices $v \in V$ that $\MEDS(H)-1 \leq \MEDS(H-v) \leq \MEDS(H)$.
\end{lemma}
\begin{proof}
 Let $F_v$ be an edge dominating set of $H-v$ and let $u \in V$ such that $\{u,v\} \in E$. Clearly, $F_v \cup \{\{u,v\}\}$ is an edge dominating set of $H$; hence $\MEDS(H) - 1 \leq |F_v| = \MEDS(H-v)$. Now, let $F$ be a minimum edge dominating set of $H$. If $v$ is not incident with an edge in $F$ then $F$ is also an edge dominating set in $H$; hence $\MEDS(H-v) \leq |F| = \MEDS(H)$. If $v$ is incident with an edge $f=\{v,u\}$ in $F$ then replace $f$ either by an edge in $\delta_H(v) \setminus \{f\}$ or delete $f$ when $\delta_H(v) \setminus \{f\}$ is empty; hence $\MEDS(H-v) \leq |F| = \MEDS(H)$.
\end{proof}

\begin{lemma}[Proposition \ref{proposition::properties} (\ref{proposition::monoton})]
  Let $Y \subseteq V$. It holds for all subsets $X \subseteq Y$ that $\MEDS(H-X) - |Y \setminus X| \leq\MEDS(H-Y) \leq \MEDS(H-X)$, and that $\cost(X) \leq \cost(Y)$.
\end{lemma}
\begin{proof}
  It follows from (\ref{proposition::vertex}) that $\MEDS(H-X) -|Y \setminus X| \leq \MEDS(H-Y) \leq \MEDS(H-X)$: Let $Y=\{y_1,y_2,\ldots, y_p\}$ and let $X=\{y_1,y_2,\ldots,y_q\}$ with $q \leq p$. It holds that $\MEDS(H-Y) \leq \MEDS(H-(Y \setminus \{y_p\})) \leq \MEDS(H-(Y \setminus \{y_{p-1},y_p\})) \leq \ldots \leq \MEDS(H-X)$. Furthermore, $\MEDS(H-Y) \geq \MEDS(H-(Y \setminus \{y_p\})) - 1 \geq \MEDS(H-(Y \setminus \{y_{p-1},y_p\})) - 2 \geq \ldots \geq \MEDS(H-X) - |Y \setminus X|$. Now, $\cost(Y) = \MEDS(H-Y) + |Y| - \MEDS(H) \geq \MEDS(H-X) - |Y \setminus X| + |Y| - \MEDS(H) = \MEDS(H-X) + |X| - \MEDS(H)= \cost(X)$.
\end{proof}

\begin{lemma}[Proposition \ref{proposition::properties} (\ref{proposition::N(W)IN})]
 Let $F$ be a minimum edge dominating set in $H$. There exists a minimum edge dominating set $F'$ in $H$ with $(V(F) \cup N_H(W) ) \setminus W \subseteq V(F')$.
\end{lemma}
\begin{proof}
 Let $F$ be an arbitrary minimum edge dominating set in $H$, and let $F'$ be a minimum edge dominating set in $H$ with $V(F) \setminus W \subseteq V(F')$ and $|V(F') \cap N_H(W)|$ maximal (under the minimum edge dominating sets that fulfill $V(F) \setminus W \subseteq V(F')$). Assume for contradiction that $V(F') \cap N_H(W) \neq N_H(W)$. Let $v \in N_H(W) \setminus V(F')$ be a vertex in the neighborhood of $W$ that is not incident with an edge in $F'$, and let $w \in N_H(v) \cap W$ be a free vertex that is adjacent to $v$. Since vertex $w$ is free, there exists a minimum edge dominating set $\widetilde{F}$ in $H-w$ (of size $|F'|-1$) with $V(F') \setminus W \subseteq V(\widetilde{F})$. Now, we can add the edge $\{v,w\}$ to the minimum edge dominating set $\widetilde{F}$ to obtain a minimum edge dominating set $\hat{F} = \widetilde{F} \cup \{\{v,w\}\}$ of $H$. It holds that $V(F) \setminus W \subseteq V(F') \setminus W \subseteq V(\widetilde{F}) \setminus W \subseteq V(\hat{F}) \setminus W$. Furthermore, the set $V(F') \cap N_H(W)$ is a proper subset of $V(\hat{F}) \cap N_H(W)$, because $V(\hat{F})$ contains all vertices in $V(F') \setminus W$ and the vertex $v \notin W$ that is not contained in $V(F')$. This contradicts the choice of $F'$ and proves the statement.
\end{proof}

\begin{lemma}[Proposition \ref{proposition::properties} (\ref{proposition::Q-W})]
 Every set that consists of one vertex $v \in Q \setminus W$ is strongly beneficial. Furthermore, these are the only beneficial sets of size one.
\end{lemma}
\begin{proof}
 Let $v \in Q \setminus W$ be a vertex that is extendable and not free, and let $B=\{v\}$. We show that $B$ is strongly beneficial. Since $\MEDS(H-v)+1 = \MEDS(H)$, and since for every set $\widetilde{B} \subsetneq B$ it holds that $\MEDS(H-\widetilde{B})=\MEDS(H)$ it holds that $B$ is beneficial. ($\widetilde{B}=\emptyset$ is the only proper subset of $B$.) Assume for contradiction that $B$ is not strongly beneficial. This would imply that there exists a cover $B_1, B_2, \ldots, B_h \subsetneq B$, but the only proper subset of $B=\{v\}$ is the empty set. Thus, $B$ is strongly beneficial.
 
 Assume there exists a beneficial set $B=\{v\}$ with $v \notin Q \setminus W$. Note that beneficial sets are disjoint from $W$ (definition), thus $v \in V \setminus Q$. Since $B$ is beneficial it holds that $\MEDS(H-B) < \MEDS(H)$. Together with Proposition~\ref{proposition::properties}~(\ref{proposition::vertex}) it follows that $\MEDS(H-B)=\MEDS(H)-1$; thus $v \in Q$ which is a contradiction.
\end{proof}
 
\begin{lemma}[Proposition \ref{proposition::properties} (\ref{proposition::disjointQ})]
 If $B$ is a strongly beneficial set of size at least two then $B$ contains no extendable vertex; hence $B \cap Q=\emptyset$. 
\end{lemma}
\begin{proof}
 If there would exists a vertex $v \in B$ that is extendable, but not free, then $\cost(\{v\})=\MEDS(H-v) + |\{v\}| - \MEDS(H)=0$. Furthermore, it follows from Proposition~\ref{proposition::properties}~(\ref{proposition::monoton}) that $\cost(B\setminus \{v\}) \leq \cost(B)$. Now, $\{v\}, B \setminus \{v\} \subsetneq B$ is a cover of $B$ and it holds that $\cost(\{v\}) + \cost(B \setminus \{v\}) \leq 0 + \cost(B) = \cost(B)$. This implies that $B$ is not strongly beneficial which is a contradiction.
\end{proof}

\begin{lemma}[Proposition \ref{proposition::properties} (\ref{proposition::BFS})]
  If there exists a set $Y \subseteq V \setminus W$ with $\MEDS(H-Y) < \MEDS(H)$, then there exists a beneficial set $B \subseteq Y$ with $\MEDS(H-B) = \MEDS(H-Y)$.
\end{lemma}
\begin{proof}
 If $Y$ is beneficial then $B=Y$ is a beneficial set with $\MEDS(H-Y)=\MEDS(H-B)$. Thus, assume that $Y$ is not beneficial. Hence, there exists a set $Y' \subsetneq Y$ with $\MEDS(H-Y)\geq \MEDS(H-Y')$ (definition of beneficial). Pick $B \subsetneq Y$ minimal with $\MEDS(H-B)\leq \MEDS(H-Y)$. This implies that $B$ is beneficial: Otherwise there would exists a set $Y' \subsetneq B$ with $\MEDS(H-Y')\leq\MEDS(H-B)$ which contradicts the choice of $B$. 
\end{proof}
 
\begin{lemma}[Proposition \ref{proposition::properties} (\ref{proposition::BFScost1})]
 If there exists a set $Y \subseteq V \setminus W$ with $\MEDS(H-Y) < \MEDS(H)$, then there exists a beneficial set $B \subseteq Y$ with $\MEDS(H-B) +1 = \MEDS(H)$. Furthermore, $B$ is strongly beneficial. 
\end{lemma}
\begin{proof}
 Let $B \subseteq Y$ be minimal such that $\MEDS(H- B) < \MEDS(H)$. Thus, for every $Z \subsetneq B$ it holds $\MEDS(H-Z)=\MEDS(H)$. First, we proof that $B$ is beneficial. If $B$ is not beneficial, then there exists a set $\widetilde{B} \subsetneq B$ such that $\MEDS(H-B) \geq \MEDS(H-\widetilde{B})$. This contradicts the choice of $B$, because $\widetilde{B} \subsetneq B$ and $\MEDS(H-\widetilde{B}) < \MEDS(H)$; hence $B$ is beneficial. 
 
 Next, we show that $\MEDS(H-B)+1=\MEDS(H)$. Let $b \in B$ and $B'=B \setminus \{b\}$. It holds that $\MEDS(H-B')=\MEDS(H)$ (choice of $B$). It follows from Proposition~\ref{proposition::properties}~(\ref{proposition::vertex}) that $\MEDS(H) -1 = \MEDS(H-B') -1 \leq \MEDS(H-B' -b) = \MEDS(H-B) \leq \MEDS(H-B') = \MEDS(H)$. Since $\MEDS(H-B) < \MEDS(H)$ it follows that $\MEDS(H-B)+1 = \MEDS(H)$. 
 
 Finally, we show that $B$ is strongly beneficial. If $B$ has size one, then $B=\{v\}$ with $v \in Q \setminus W$, and it follows that $B$ is strongly beneficial Proposition~\ref{proposition::properties}~(\ref{proposition::Q-W}). 
 Now, assume for contradiction that $B$ is not strongly beneficial. Hence, there exists a cover $B_1,B_2,\ldots,B_h \subsetneq B$ of $B$ with $\cost(B) \geq \sum_{i=1}^h \cost(B_i)$. 
 It holds that $\cost(B)=\MEDS(H-B)+|B|-\MEDS(H)=|B|-1$ (because $\MEDS(H-B)+1=\MEDS(H)$), and it always holds that $\cost(B_i) \leq |B_i|$ (definition of $\cost$). This implies that there exists at least one $i^* \in [h]$ with $\cost(B_{i^*})<|B_{i^*}|$: otherwise
 \[
  |B|-1 = \cost(B) \geq \sum_{i=1}^h \cost(B_i) = \sum_{i=1}^h |B_i| \geq |B|.
 \]
 Now, $|B|-\cost(B) = 1 \leq |B_{i^*}| - \cost(B_{i^*})$. This is a contradiction to $B$ beneficial, because $B_{i^*} \subsetneq B$ is a proper subset of $B$ and it holds that $|B|-\cost(B) \leq |B_{i^*}| - \cost(B_{i^*})$. Thus, $B$ is strongly beneficial.
\end{proof}

\begin{lemma}[Proposition \ref{proposition::properties} (\ref{proposition::strongly})]
 If $H$ has a beneficial set $B$, then $H$ has also a strongly beneficial set $B' \subseteq B$.
\end{lemma}
\begin{proof}
 This follows from Proposition~\ref{proposition::properties}~(\ref{proposition::BFScost1}).
\end{proof}

\begin{lemma}[Proposition \ref{proposition::properties} (\ref{proposition::BFSedge})]
 Let $F$ be a minimum edge dominating set in $H$. If $e=\{x,y\}$ is an edge in $F$ with $x, y \notin Q$, then $\{x,y\}$ is a strongly beneficial set.
\end{lemma}
\begin{proof}
 Let $F$ be a minimum edge dominating set in $H$ and let $\{x,y\}$ be an edge in $F$ with $x,y \notin Q$. First, we show that $B=\{x,y\}$ is beneficial. 
 It holds that $\MEDS(H-B) < \MEDS(H)$, because $F'=F \setminus \{\{x,y\}\}$ is an edge dominating set in $H-B$: $F'$ covers all edges that are not incident with $x$ and $y$, and these vertices are not contained in $H-B$. 
It follows from Proposition~\ref{proposition::properties}~(\ref{proposition::BFScost1}) that there exists a strongly beneficial set $B' \subseteq B$ with $\MEDS(H-B')+1=\MEDS(H)$. The sets $\emptyset$, $\{x\}$ and $\{y\}$ are not beneficial, because the empty set is not beneficial, and neither $x$ nor $y$ is extendable; hence $\MEDS(H-x)=\MEDS(H-y)=\MEDS(H)$. Note that the only beneficial sets of size consists of one vertex in $Q \setminus W$ (Proposition~\ref{proposition::properties}~(\ref{proposition::Q-W})). Thus, $B$ must be strongly beneficial.
\end{proof}

\begin{lemma}[Proposition \ref{proposition::properties} (\ref{proposition::definition})]
 Let $B$ be a beneficial set. $B$ is strongly beneficial if and only if for every non-trivial partition $B_1,B_2,\ldots,B_h$ of $B$ it holds that $\cost(B) < \sum_{i=1}^h \cost(B_i)$.
\end{lemma}
\begin{proof}
 $(\Rightarrow:)$ This follows directly from the definition of strongly beneficial set, because every non-trivial partition $B_1,B_2, \ldots, B_h$ of $B$ is also a cover of $B$ with $B_1,B_2,\ldots, B_h \subsetneq B$.
 
 $(\Leftarrow:)$ Assume that $B$ is not strongly beneficial. Thus, there exists a cover $B_1, B_2, \ldots, B_h \subsetneq B$ of $B$ with $\cost(B) \geq \sum_{i=1}^h \cost(B_i)$. We construct a non-trivial partition of $B$ as follows: Let $B_1'=B_1$ and let $B_i' = B_i \setminus \left(\bigcup_{j=1}^{i-1} B_j \right)$ for all $2 \leq i \leq h$. It holds that no set $B_i'$ is the set $B$ because $B_i' \subseteq B_i \subsetneq B$. Furthermore, the union of all sets $B_i'$, with $i \in [h]$, is still $B$, and the intersection of two sets $B_i$ and $B_j$ with $i \neq j$ is empty (by construction). Thus, all nonempty set $B_i'$ are a non-trivial partition of $B$. Additionally, it follows from Proposition~\ref{proposition::properties}~(\ref{proposition::monoton}) that $\cost(B_i') \leq \cost(B_i)$ which implies that $\cost(B) \geq \sum_{i=1}^h \cost(B_i)  \geq \sum_{i=1}^h \cost(B_i')$. Hence, there exists a non-trivial partition $B_1'', B_2'', \ldots, B_q''$ of $B$ with $\cost(B) \geq \sum_{i=1}^q \cost(B_i'')$. This concludes the proof. 
\end{proof}

\begin{lemma}[Proposition \ref{proposition::properties} (\ref{proposition::decomposeB})]
 Let $Y \subseteq V \setminus W$. There exists a partition $B_1,B_2,\ldots,B_h$ of $Y$ where  $B_i$ is either strongly beneficial or where $B_i$ has $\cost(B_i)=|B_i|$, for all $i \in [h]$, such that $\cost(Y) \geq \sum_{i=1}^h \cost(B_i)$. (Note that we also allow trivial partitions.) 
\end{lemma}
\begin{proof}
 Assume that the statement does not hold and let $Y \subseteq V \setminus W$ be a minimal set that does not fulfill the properties of the lemma. Hence, $Y$ is neither strongly beneficial nor has $\cost(Y)=|Y|$, because in both cases the trivial partition $Y$ would fulfill the properties of the lemma. 
 
 First, assume that $Y$ is not beneficial. Thus, there exists a set $Y' \subsetneq Y$ beneficial with $\MEDS(H-Y)= \MEDS(H-Y')$ Proposition~\ref{proposition::properties}~(\ref{proposition::BFS}).  
 Since $Y' \subsetneq Y$ is a proper subset of $Y$ and $Y$ is a minimal set that does not fulfill the properties of the lemma, there exists a partition $B_1', B_2'\ldots, B_p'$ of $Y'$ where $B_i'$ is either strongly beneficial or has $\cost(B_i')=|B_i|$ for all $i \in [p]$ such that $\cost(Y') \geq \sum_{i=1}^p \cost(B_i')$. Furthermore, the set $Y'' = Y \setminus Y'$ is a proper subset of $Y$, because $Y'$ is not the empty set (if $Y'$ is the empty set then $Y'$ is not beneficial). Thus, there exists a partition $B_1'', B_2'', \ldots, B_q''$ of $Y''$ where $B_i''$ is either strongly beneficial or has $\cost(B_i'')=|B_i''|$ for all $i \in [q]$ such that $\cost(Y'') \geq \sum_{i=1}^q \cost(B_i'')$.
 
 Now, $B_1=B_1',B_2=B_2',\ldots, B_p=B_p', B_{p+1}=B_1'',B_{p+2}=B_2'',\ldots,B_{p+q}=B_q''$ is a partition of $Y$, because $B_1', B_2'\ldots, B_p'$ is a partition of $Y'$, $B_1'', B_2'', \ldots, B_q''$ is a partition of $Y''$, and $Y', Y''$ is a partition of $Y$. Additionally, every set $B_i$, with $i \in [p+q]$, is either strongly beneficial or has $\cost(B_i)=|B_i|$ (by choice of $B_i$). To show that $Y$ also fulfills the properties of the lemma it remains to show that $\cost(Y) \geq \sum_{i=1}^{p+q} \cost(B_i)$. It holds that
 \begin{align*}
  \cost(Y) &=\MEDS(H-Y) +|Y|-\MEDS(H) \\
	   &= \MEDS(H-Y') + |Y'| + |Y''| - \MEDS(H) &\text{// bc. choice of } Y' \text{ and } Y=Y' \dot\cup Y''\\
	   &\geq \cost(Y') + \cost(Y'') &\text{// bc. definition of cost} \\
	   &\geq \sum_{i=1}^p \cost(B_i') + \sum_{i=1}^q \cost(B_i'') = \sum_{i=1}^{p+q} \cost(B_i)
  \end{align*}
 This implies that $Y$ fulfills the properties of the lemma, which is a contradiction.
 
 Thus, assume that $Y$ is beneficial (but not strongly beneficial). Hence, there exists a non-trivial partition $B_1, B_2,\ldots, B_h$ of $Y$ with $\cost(Y) \geq \sum_{i=1}^h \cost(B_i)$ Proposition~\ref{proposition::properties}~(\ref{proposition::definition}). Every $B_i$, with $i \in [h]$, is a proper subset of $Y$, because $B_1, B_2, \ldots, B_h$ is a non-trivial partition of $Y$. Since $Y$ is a minimal set that does not fulfill the properties of the lemma, there exists, for all $i \in [h]$, a partition $B_{i,1}, B_{i,2},\ldots, B_{i,p_i}$ of $B_i$ where $B_{i,j}$ is either strongly beneficial or has $\cost(B_{i,j})=|B_{i,j}|$, for all $j \in [p_i]$, such that $\cost(B_i) \geq \sum_{j=1}^{p_i} \cost(B_{i,j})$. By construction, the sets $B_{1,1}, B_{1,2},\ldots, B_{1,p_1}, B_{2,1},\ldots, B_{h,p_h}$ are a partition of $Y$ where every $B_{i,j}$ is either strongly beneficial or has $\cost(B_{i,j})=|B_{i,j}|$, for all $i \in [h]$ and all $j \in [p_i]$. Furthermore, $\cost(Y) \geq \sum_{i=1}^h \cost(B_i) \geq \sum_{i=1}^h \sum_{j=1}^{p_i} \cost(B_{i,j})$. Thus, the set $Y$ fulfills the properties of the lemma, which is a contradiction and concludes the proof.
\end{proof}

\section{Conclusion}\label{section::conclusion}

As our main result, we have given a complete classification for \EDS parameterized by the size of a modulator to \cH-component graphs for all finite sets \cH. An obvious follow-up question is to extend this result to infinite sets \cH. Our lower bounds of course continue to work in this setting, and the upper bounds still permit us to reduce the number of connected components (under the same conditions as before, e.g., that relevant beneficial sets have bounded size). However, for infinite \cH, polynomial kernels also require us to shrink connected components of $G-X$, and to derive general rules for this. Moreover, even determining beneficial sets etc.\ for graphs $H\in\cH$ could no longer be dismissed as being constant time. It is conceivable that such a classification is doable whenever graphs in $\cH$ have bounded treewidth, as this simplifies the required additional steps. Since most known tractable graph classes for \EDS have bounded treewidth (and tractability for $G-X$ is required, or else \NP-hardness for $|X|=0$ rules out kernels and fixed-parameter tractability), this seems like a reasonable goal.
Apart from this, it would be nice to close the gap between size $\Oh(|X|^{d+1}\log |X|)$ and the lower bound of $\Oh(|X|^{d-\varepsilon})$, where improvements to the upper bound seem more likely.

\bibliography{lit}

\begin{thebibliography}{10}

\bibitem{BodlaenderDFH09}
Hans~L. Bodlaender, Rodney~G. Downey, Michael~R. Fellows, and Danny Hermelin.
\newblock On problems without polynomial kernels.
\newblock {\em J. Comput. Syst. Sci.}, 75(8):423--434, 2009.
\newblock URL: \url{https://doi.org/10.1016/j.jcss.2009.04.001}, \href
  {http://dx.doi.org/10.1016/j.jcss.2009.04.001}
  {\path{doi:10.1016/j.jcss.2009.04.001}}.

\bibitem{BodlaenderJK14}
Hans~L. Bodlaender, Bart M.~P. Jansen, and Stefan Kratsch.
\newblock Kernelization lower bounds by cross-composition.
\newblock {\em {SIAM} J. Discrete Math.}, 28(1):277--305, 2014.
\newblock URL: \url{https://doi.org/10.1137/120880240}, \href
  {http://dx.doi.org/10.1137/120880240} {\path{doi:10.1137/120880240}}.

\bibitem{CardinalLL09}
Jean Cardinal, Stefan Langerman, and Eythan Levy.
\newblock Improved approximation bounds for edge dominating set in dense
  graphs.
\newblock {\em Theor. Comput. Sci.}, 410(8-10):949--957, 2009.
\newblock URL: \url{https://doi.org/10.1016/j.tcs.2008.12.036}, \href
  {http://dx.doi.org/10.1016/j.tcs.2008.12.036}
  {\path{doi:10.1016/j.tcs.2008.12.036}}.

\bibitem{ChlebikC06}
Miroslav Chleb{\'{\i}}k and Janka Chleb{\'{\i}}kov{\'{a}}.
\newblock Approximation hardness of edge dominating set problems.
\newblock {\em J. Comb. Optim.}, 11(3):279--290, 2006.
\newblock URL: \url{https://doi.org/10.1007/s10878-006-7908-0}, \href
  {http://dx.doi.org/10.1007/s10878-006-7908-0}
  {\path{doi:10.1007/s10878-006-7908-0}}.

\bibitem{Cook71}
Stephen~A. Cook.
\newblock The complexity of theorem-proving procedures.
\newblock In Michael~A. Harrison, Ranan~B. Banerji, and Jeffrey~D. Ullman,
  editors, {\em Proceedings of the 3rd Annual {ACM} Symposium on Theory of
  Computing, May 3-5, 1971, Shaker Heights, Ohio, {USA}}, pages 151--158.
  {ACM}, 1971.
\newblock URL: \url{http://doi.acm.org/10.1145/800157.805047}, \href
  {http://dx.doi.org/10.1145/800157.805047} {\path{doi:10.1145/800157.805047}}.

\bibitem{CyganPPW13}
Marek Cygan, Marcin Pilipczuk, Michal Pilipczuk, and Jakub~Onufry Wojtaszczyk.
\newblock On multiway cut parameterized above lower bounds.
\newblock {\em {TOCT}}, 5(1):3:1--3:11, 2013.
\newblock URL: \url{http://doi.acm.org/10.1145/2462896.2462899}, \href
  {http://dx.doi.org/10.1145/2462896.2462899}
  {\path{doi:10.1145/2462896.2462899}}.

\bibitem{DBLP:conf/soda/DellM12}
Holger Dell and D{\'{a}}niel Marx.
\newblock Kernelization of packing problems.
\newblock In Yuval Rabani, editor, {\em Proceedings of the Twenty-Third Annual
  {ACM-SIAM} Symposium on Discrete Algorithms, {SODA} 2012, Kyoto, Japan,
  January 17-19, 2012}, pages 68--81. {SIAM}, 2012.
\newblock URL:
  \url{http://portal.acm.org/citation.cfm?id=2095122&CFID=63838676&CFTOKEN=79617016},
  \href {http://dx.doi.org/10.1137/1.9781611973099}
  {\path{doi:10.1137/1.9781611973099}}.

\bibitem{DellM14}
Holger Dell and Dieter van Melkebeek.
\newblock Satisfiability allows no nontrivial sparsification unless the
  polynomial-time hierarchy collapses.
\newblock {\em J. {ACM}}, 61(4):23:1--23:27, 2014.
\newblock URL: \url{http://doi.acm.org/10.1145/2629620}, \href
  {http://dx.doi.org/10.1145/2629620} {\path{doi:10.1145/2629620}}.

\bibitem{Diestel12}
Reinhard Diestel.
\newblock {\em Graph Theory, 4th Edition}, volume 173 of {\em Graduate texts in
  mathematics}.
\newblock Springer, 2012.

\bibitem{EscoffierMPX15}
Bruno Escoffier, J{\'{e}}r{\^{o}}me Monnot, Vangelis~Th. Paschos, and Mingyu
  Xiao.
\newblock New results on polynomial inapproximabilityand fixed parameter
  approximability of edge dominating set.
\newblock {\em Theory Comput. Syst.}, 56(2):330--346, 2015.
\newblock URL: \url{https://doi.org/10.1007/s00224-014-9549-5}, \href
  {http://dx.doi.org/10.1007/s00224-014-9549-5}
  {\path{doi:10.1007/s00224-014-9549-5}}.

\bibitem{Fernau06}
Henning Fernau.
\newblock edge dominating set: Efficient enumeration-based exact algorithms.
\newblock In Hans~L. Bodlaender and Michael~A. Langston, editors, {\em
  Parameterized and Exact Computation, Second International Workshop, {IWPEC}
  2006, Z{\"{u}}rich, Switzerland, September 13-15, 2006, Proceedings}, volume
  4169 of {\em Lecture Notes in Computer Science}, pages 142--153. Springer,
  2006.
\newblock URL: \url{https://doi.org/10.1007/11847250_13}, \href
  {http://dx.doi.org/10.1007/11847250_13} {\path{doi:10.1007/11847250_13}}.

\bibitem{FominGSS09}
Fedor~V. Fomin, Serge Gaspers, Saket Saurabh, and Alexey~A. Stepanov.
\newblock On two techniques of combining branching and treewidth.
\newblock {\em Algorithmica}, 54(2):181--207, 2009.
\newblock URL: \url{https://doi.org/10.1007/s00453-007-9133-3}, \href
  {http://dx.doi.org/10.1007/s00453-007-9133-3}
  {\path{doi:10.1007/s00453-007-9133-3}}.

\bibitem{FortnowS11}
Lance Fortnow and Rahul Santhanam.
\newblock Infeasibility of instance compression and succinct pcps for {NP}.
\newblock {\em J. Comput. Syst. Sci.}, 77(1):91--106, 2011.
\newblock URL: \url{https://doi.org/10.1016/j.jcss.2010.06.007}, \href
  {http://dx.doi.org/10.1016/j.jcss.2010.06.007}
  {\path{doi:10.1016/j.jcss.2010.06.007}}.

\bibitem{FujitoN02}
Toshihiro Fujito and Hiroshi Nagamochi.
\newblock A 2-approximation algorithm for the minimum weight edge dominating
  set problem.
\newblock {\em Discrete Applied Mathematics}, 118(3):199--207, 2002.
\newblock URL: \url{https://doi.org/10.1016/S0166-218X(00)00383-8}, \href
  {http://dx.doi.org/10.1016/S0166-218X(00)00383-8}
  {\path{doi:10.1016/S0166-218X(00)00383-8}}.

\bibitem{GargP16}
Shivam Garg and Geevarghese Philip.
\newblock Raising the bar for vertex cover: Fixed-parameter tractability above
  {A} higher guarantee.
\newblock In Robert Krauthgamer, editor, {\em Proceedings of the Twenty-Seventh
  Annual {ACM-SIAM} Symposium on Discrete Algorithms, {SODA} 2016, Arlington,
  VA, USA, January 10-12, 2016}, pages 1152--1166. {SIAM}, 2016.
\newblock URL: \url{https://doi.org/10.1137/1.9781611974331.ch80}, \href
  {http://dx.doi.org/10.1137/1.9781611974331.ch80}
  {\path{doi:10.1137/1.9781611974331.ch80}}.

\bibitem{GolovachHKV15}
Petr~A. Golovach, Pinar Heggernes, Dieter Kratsch, and Yngve Villanger.
\newblock An incremental polynomial time algorithm to enumerate all minimal
  edge dominating sets.
\newblock {\em Algorithmica}, 72(3):836--859, 2015.
\newblock URL: \url{https://doi.org/10.1007/s00453-014-9875-7}, \href
  {http://dx.doi.org/10.1007/s00453-014-9875-7}
  {\path{doi:10.1007/s00453-014-9875-7}}.

\bibitem{Hagerup12}
Torben Hagerup.
\newblock Kernels for edge dominating set: Simpler or smaller.
\newblock In Branislav Rovan, Vladimiro Sassone, and Peter Widmayer, editors,
  {\em Mathematical Foundations of Computer Science 2012 - 37th International
  Symposium, {MFCS} 2012, Bratislava, Slovakia, August 27-31, 2012.
  Proceedings}, volume 7464 of {\em Lecture Notes in Computer Science}, pages
  491--502. Springer, 2012.
\newblock URL: \url{https://doi.org/10.1007/978-3-642-32589-2_44}, \href
  {http://dx.doi.org/10.1007/978-3-642-32589-2_44}
  {\path{doi:10.1007/978-3-642-32589-2_44}}.

\bibitem{DBLP:journals/siamcomp/HopcroftK73}
John~E. Hopcroft and Richard~M. Karp.
\newblock An n\({}^{\mbox{5/2}}\) algorithm for maximum matchings in bipartite
  graphs.
\newblock {\em {SIAM} J. Comput.}, 2(4):225--231, 1973.
\newblock URL: \url{https://doi.org/10.1137/0202019}, \href
  {http://dx.doi.org/10.1137/0202019} {\path{doi:10.1137/0202019}}.

\bibitem{IwaideN16}
Ken Iwaide and Hiroshi Nagamochi.
\newblock An improved algorithm for parameterized edge dominating set problem.
\newblock {\em J. Graph Algorithms Appl.}, 20(1):23--58, 2016.
\newblock URL: \url{https://doi.org/10.7155/jgaa.00383}, \href
  {http://dx.doi.org/10.7155/jgaa.00383} {\path{doi:10.7155/jgaa.00383}}.

\bibitem{JansenB13}
Bart M.~P. Jansen and Hans~L. Bodlaender.
\newblock Vertex cover kernelization revisited - upper and lower bounds for a
  refined parameter.
\newblock {\em Theory Comput. Syst.}, 53(2):263--299, 2013.
\newblock URL: \url{https://doi.org/10.1007/s00224-012-9393-4}, \href
  {http://dx.doi.org/10.1007/s00224-012-9393-4}
  {\path{doi:10.1007/s00224-012-9393-4}}.

\bibitem{KanteLMN12}
Mamadou~Moustapha Kant{\'{e}}, Vincent Limouzy, Arnaud Mary, and Lhouari
  Nourine.
\newblock On the neighbourhood helly of some graph classes and applications to
  the enumeration of minimal dominating sets.
\newblock In Kun{-}Mao Chao, Tsan{-}sheng Hsu, and Der{-}Tsai Lee, editors,
  {\em Algorithms and Computation - 23rd International Symposium, {ISAAC} 2012,
  Taipei, Taiwan, December 19-21, 2012. Proceedings}, volume 7676 of {\em
  Lecture Notes in Computer Science}, pages 289--298. Springer, 2012.
\newblock URL: \url{https://doi.org/10.1007/978-3-642-35261-4_32}, \href
  {http://dx.doi.org/10.1007/978-3-642-35261-4_32}
  {\path{doi:10.1007/978-3-642-35261-4_32}}.

\bibitem{KanteLMNU15}
Mamadou~Moustapha Kant{\'{e}}, Vincent Limouzy, Arnaud Mary, Lhouari Nourine,
  and Takeaki Uno.
\newblock Polynomial delay algorithm for listing minimal edge dominating sets
  in graphs.
\newblock In Frank Dehne, J{\"{o}}rg{-}R{\"{u}}diger Sack, and Ulrike Stege,
  editors, {\em Algorithms and Data Structures - 14th International Symposium,
  {WADS} 2015, Victoria, BC, Canada, August 5-7, 2015. Proceedings}, volume
  9214 of {\em Lecture Notes in Computer Science}, pages 446--457. Springer,
  2015.
\newblock URL: \url{https://doi.org/10.1007/978-3-319-21840-3_37}, \href
  {http://dx.doi.org/10.1007/978-3-319-21840-3_37}
  {\path{doi:10.1007/978-3-319-21840-3_37}}.

\bibitem{KoblerR03}
Daniel Kobler and Udi Rotics.
\newblock Edge dominating set and colorings on graphs with fixed clique-width.
\newblock {\em Discrete Applied Mathematics}, 126(2-3):197--221, 2003.
\newblock URL: \url{https://doi.org/10.1016/S0166-218X(02)00198-1}, \href
  {http://dx.doi.org/10.1016/S0166-218X(02)00198-1}
  {\path{doi:10.1016/S0166-218X(02)00198-1}}.

\bibitem{Kratsch16}
Stefan Kratsch.
\newblock A randomized polynomial kernelization for vertex cover with a smaller
  parameter.
\newblock In Piotr Sankowski and Christos~D. Zaroliagis, editors, {\em 24th
  Annual European Symposium on Algorithms, {ESA} 2016, August 22-24, 2016,
  Aarhus, Denmark}, volume~57 of {\em LIPIcs}, pages 59:1--59:17. Schloss
  Dagstuhl - Leibniz-Zentrum fuer Informatik, 2016.
\newblock URL: \url{https://doi.org/10.4230/LIPIcs.ESA.2016.59}, \href
  {http://dx.doi.org/10.4230/LIPIcs.ESA.2016.59}
  {\path{doi:10.4230/LIPIcs.ESA.2016.59}}.

\bibitem{KratschW12}
Stefan Kratsch and Magnus Wahlstr{\"{o}}m.
\newblock Representative sets and irrelevant vertices: New tools for
  kernelization.
\newblock In {\em 53rd Annual {IEEE} Symposium on Foundations of Computer
  Science, {FOCS} 2012, New Brunswick, NJ, USA, October 20-23, 2012}, pages
  450--459. {IEEE} Computer Society, 2012.
\newblock URL: \url{https://doi.org/10.1109/FOCS.2012.46}, \href
  {http://dx.doi.org/10.1109/FOCS.2012.46} {\path{doi:10.1109/FOCS.2012.46}}.

\bibitem{LokshtanovNRRS14}
Daniel Lokshtanov, N.~S. Narayanaswamy, Venkatesh Raman, M.~S. Ramanujan, and
  Saket Saurabh.
\newblock Faster parameterized algorithms using linear programming.
\newblock {\em {ACM} Trans. Algorithms}, 11(2):15:1--15:31, 2014.
\newblock URL: \url{http://doi.acm.org/10.1145/2566616}, \href
  {http://dx.doi.org/10.1145/2566616} {\path{doi:10.1145/2566616}}.

\bibitem{Prieto05}
Elena Prieto.
\newblock {\em Systematic kernelization in {FPT} algorithm design}.
\newblock PhD thesis, The University of Newcastle, Australia, 2005.

\bibitem{RamanRS11}
Venkatesh Raman, M.~S. Ramanujan, and Saket Saurabh.
\newblock Paths, flowers and vertex cover.
\newblock In Camil Demetrescu and Magn{\'{u}}s~M. Halld{\'{o}}rsson, editors,
  {\em Algorithms - {ESA} 2011 - 19th Annual European Symposium,
  Saarbr{\"{u}}cken, Germany, September 5-9, 2011. Proceedings}, volume 6942 of
  {\em Lecture Notes in Computer Science}, pages 382--393. Springer, 2011.
\newblock URL: \url{https://doi.org/10.1007/978-3-642-23719-5_33}, \href
  {http://dx.doi.org/10.1007/978-3-642-23719-5_33}
  {\path{doi:10.1007/978-3-642-23719-5_33}}.

\bibitem{RamanSS07}
Venkatesh Raman, Saket Saurabh, and Somnath Sikdar.
\newblock Efficient exact algorithms through enumerating maximal independent
  sets and other techniques.
\newblock {\em Theory Comput. Syst.}, 41(3):563--587, 2007.
\newblock URL: \url{https://doi.org/10.1007/s00224-007-1334-2}, \href
  {http://dx.doi.org/10.1007/s00224-007-1334-2}
  {\path{doi:10.1007/s00224-007-1334-2}}.

\bibitem{SchmiedV12}
Richard Schmied and Claus Viehmann.
\newblock Approximating edge dominating set in dense graphs.
\newblock {\em Theor. Comput. Sci.}, 414(1):92--99, 2012.
\newblock URL: \url{https://doi.org/10.1016/j.tcs.2011.10.001}, \href
  {http://dx.doi.org/10.1016/j.tcs.2011.10.001}
  {\path{doi:10.1016/j.tcs.2011.10.001}}.

\bibitem{RooijB12}
Johan M.~M. van Rooij and Hans~L. Bodlaender.
\newblock Exact algorithms for edge domination.
\newblock {\em Algorithmica}, 64(4):535--563, 2012.
\newblock URL: \url{https://doi.org/10.1007/s00453-011-9546-x}, \href
  {http://dx.doi.org/10.1007/s00453-011-9546-x}
  {\path{doi:10.1007/s00453-011-9546-x}}.

\bibitem{WangCFC09}
Jianxin Wang, Beiwei Chen, Qilong Feng, and Jianer Chen.
\newblock An efficient fixed-parameter enumeration algorithm for weighted edge
  dominating set.
\newblock In Xiaotie Deng, John~E. Hopcroft, and Jinyun Xue, editors, {\em
  Frontiers in Algorithmics, Third International Workshop, {FAW} 2009, Hefei,
  China, June 20-23, 2009. Proceedings}, volume 5598 of {\em Lecture Notes in
  Computer Science}, pages 237--250. Springer, 2009.
\newblock URL: \url{https://doi.org/10.1007/978-3-642-02270-8_25}, \href
  {http://dx.doi.org/10.1007/978-3-642-02270-8_25}
  {\path{doi:10.1007/978-3-642-02270-8_25}}.

\bibitem{Xiao10}
Mingyu Xiao.
\newblock Exact and parameterized algorithms for edge dominating set in
  3-degree graphs.
\newblock In Weili Wu and Ovidiu Daescu, editors, {\em Combinatorial
  Optimization and Applications - 4th International Conference, {COCOA} 2010,
  Kailua-Kona, HI, USA, December 18-20, 2010, Proceedings, Part {II}}, volume
  6509 of {\em Lecture Notes in Computer Science}, pages 387--400. Springer,
  2010.
\newblock URL: \url{https://doi.org/10.1007/978-3-642-17461-2_31}, \href
  {http://dx.doi.org/10.1007/978-3-642-17461-2_31}
  {\path{doi:10.1007/978-3-642-17461-2_31}}.

\bibitem{XiaoKP13}
Mingyu Xiao, Ton Kloks, and Sheung{-}Hung Poon.
\newblock New parameterized algorithms for the edge dominating set problem.
\newblock {\em Theor. Comput. Sci.}, 511:147--158, 2013.
\newblock URL: \url{https://doi.org/10.1016/j.tcs.2012.06.022}, \href
  {http://dx.doi.org/10.1016/j.tcs.2012.06.022}
  {\path{doi:10.1016/j.tcs.2012.06.022}}.

\bibitem{XiaoN13a}
Mingyu Xiao and Hiroshi Nagamochi.
\newblock Parameterized edge dominating set in graphs with degree bounded by 3.
\newblock {\em Theor. Comput. Sci.}, 508:2--15, 2013.
\newblock URL: \url{https://doi.org/10.1016/j.tcs.2012.08.015}, \href
  {http://dx.doi.org/10.1016/j.tcs.2012.08.015}
  {\path{doi:10.1016/j.tcs.2012.08.015}}.

\bibitem{XiaoN14}
Mingyu Xiao and Hiroshi Nagamochi.
\newblock A refined exact algorithm for edge dominating set.
\newblock {\em Theor. Comput. Sci.}, 560:207--216, 2014.
\newblock URL: \url{https://doi.org/10.1016/j.tcs.2014.07.019}, \href
  {http://dx.doi.org/10.1016/j.tcs.2014.07.019}
  {\path{doi:10.1016/j.tcs.2014.07.019}}.

\bibitem{YannakakisG80}
Mihalis Yannakakis and Fanica Gavril.
\newblock Edge dominating sets in graphs.
\newblock {\em SIAM Journal on Applied Mathematics}, 38(3):364--372, 1980.

\end{thebibliography}

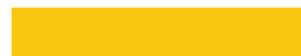
\end{document}